\documentclass[11pt,a4paper]{article}
\usepackage[margin=1in]{geometry}
\usepackage[latin1]{inputenc}
\usepackage[T1]{fontenc}   
\usepackage[english]{babel} % Language hyphenation and typographical rules
\usepackage{enumerate}
\usepackage{soul}

\usepackage{appendix}
\usepackage[ruled,linesnumbered]{algorithm2e} % https://tug.ctan.org/macros/latex/contrib/algorithm2e/doc/algorithm2e.pdf
\usepackage[noend]{algpseudocode}
\usepackage{xcolor}
\usepackage{listings}
\usepackage{csvsimple}
\usepackage{graphicx}

\usepackage{amsmath, amssymb, amscd, amsthm, amsfonts}
\usepackage{thm-restate}
\usepackage[colorlinks=true, allcolors=blue]{hyperref}
\usepackage{tikz}
\usetikzlibrary{positioning,shapes.geometric,calc,math}
\newcommand{\drawTangentLines}[6]{
    \begin{scope}
        \coordinate (C) at (#1,#2);  % Circle center
        \coordinate (P) at (#4,#5);  % External point
        
        \pgfmathparse{sqrt((#4-#1)^2 + (#5-#2)^2)}
        \let\d\pgfmathresult
        
        \pgfmathparse{acos(#3/\d)} 
        \let\tangentangle\pgfmathresult
        
        \pgfmathparse{atan2(#5-#2,#4-#1)}
        \let\baseangle\pgfmathresult
        
        \pgfmathparse{\baseangle + \tangentangle}
        \let\angleOne\pgfmathresult
        \pgfmathparse{\baseangle - \tangentangle}
        \let\angleTwo\pgfmathresult
        
        \coordinate (T1) at ($(C) + (\angleOne:#3)$);
        \coordinate (T2) at ($(C) + (\angleTwo:#3)$);
        
        \draw[#6] (P) -- (T1);
        \draw[#6] (P) -- (T2);
        
    \end{scope}
}
\usepackage{cellspace}
\usepackage{cleveref}
\crefname{enumi}{step}{steps}
\usepackage{xspace}
\usepackage{csquotes}
\usepackage{enumitem}

\usepackage{multirow}

\setlist{itemsep=1pt}

\usepackage{natbib}

\usepackage{pifont}% http://ctan.org/pkg/pifont
\newtheorem{theorem}{Theorem}[section]

\newtheorem{lemma}[theorem]{Lemma}

\newtheorem{observation}[theorem]{Observation}

\newtheorem{corollary}[theorem]{Corollary}

\newtheorem{question}{Open Question}

\newtheorem{definition}[theorem]{Definition}

\usepackage[colorinlistoftodos,prependcaption,textsize=tiny]{todonotes}

\newcommand{\Detect}{\textnormal{\textsc{Detect}}}
\newcommand{\List}{\textnormal{\textsc{List}}}
\newcommand{\MemDetect}{\textnormal{\textsc{MemDetect}}}
\newcommand{\MemList}{\textnormal{\textsc{MemList}}}
\newcommand{\No}{\textnormal{\textsc{No}}}
\newcommand{\Yes}{\textnormal{\textsc{Yes}}}

\newcommand{\SIGNAL}{\textnormal{\texttt{Signal}}}

\newcommand{\Del}{\textnormal{\texttt{Del}}}

\newcommand{\create}{\mathsf{create}}
\newcommand{\firstDiff}{\mathsf{1\textit{-}DistinctBit}}
\newcommand{\diffList}{\mathsf{1\textit{-}DBList}}
\newcommand{\ecc}{\operatorname{ecc}}
\newcommand{\dist}{\operatorname{dist}}
\newcommand{\rad}{\operatorname{rad}}
\newcommand{\diam}{\operatorname{diam}}
\newcommand{\centre}{\operatorname{center}}
\newcommand{\Count}{\operatorname{Count}}
\newcommand{\nerad}{\widetilde{\operatorname{rad}}}
\newcommand{\nediam}{\widetilde{\operatorname{diam}}}
\newcommand{\necenter}{\widetilde{\operatorname{center}}}
\newcommand{\neecc}{\widetilde{\operatorname{ecc}}}
\newcommand{\ID}{\operatorname{ID}}
\newcommand{\poly}{\operatorname{poly}}

\newcommand{\tlast}{t_{\mathsf{final}}}

\newcommand{\LOCAL}{\mathsf{LOCAL}}
\newcommand{\CONGEST}{\mathsf{CONGEST}}

\usepackage{nameref}

\makeatletter
\let\orgdescriptionlabel\descriptionlabel
\renewcommand*{\descriptionlabel}[1]{%
  \let\orglabel\label
  \let\label\@gobble
  \phantomsection
  \edef\@currentlabel{#1\unskip}%
  \let\label\orglabel
  \orgdescriptionlabel{#1}%
}
\makeatother

\title{The Complexity Landscape of \\ Dynamic Distributed Subgraph Finding}
 \author{
 \hspace{1.5cm}
 Yi-Jun Chang\footnote{National University of Singapore. ORCID: 0000-0002-0109-2432. Email: cyijun@nus.edu.sg} 
 \and Lyuting Chen\footnote{National University of Singapore. ORCID: 0009-0002-8836-6607. Email: e0726582@u.nus.edu}
 \and Yanyu Chen\footnote{National University of Singapore. ORCID: 0009-0008-8068-1649. Email: yanyu.chen@u.nus.edu}
 \hspace{1.5cm}
 \and Gopinath Mishra\footnote{National University of Singapore. ORCID: 0000-0003-0540-0292. Email: gopinath@nus.edu.sg}   
 \and Mingyang Yang\footnote{National University of Singapore. ORCID: 0009-0006-8971-2064. Email: myangat@u.nus.edu} 
 }
\date{}

\begin{document}

\maketitle
\begin{abstract}
Bonne and Censor-Hillel (ICALP 2019) initiated the study of distributed subgraph finding in dynamic networks of limited bandwidth. For the case where the target subgraph is a clique, they determined the tight bandwidth complexity bounds in nearly all settings. However, several open questions remain, and very little is known about finding subgraphs beyond cliques. In this work, we consider these questions and explore subgraphs beyond cliques in the \emph{deterministic} setting.

For finding cliques, we establish an $\Omega(\log \log n)$ bandwidth lower bound for one-round membership-detection under edge insertions only and an $\Omega(\log \log \log n)$ bandwidth lower bound for one-round detection under both edge insertions and node insertions. Moreover, we demonstrate new algorithms to show that our lower bounds are \emph{tight} in bounded-degree networks when the target subgraph is a triangle. Prior to our work, no lower bounds were known for these problems.

For finding subgraphs beyond cliques, we present a \emph{complete characterization} of the bandwidth complexity of the membership-listing problem for every target subgraph, every number of rounds, and every type of topological change: node insertions, node deletions, edge insertions, and edge deletions. We also show partial characterizations for one-round membership-detection and listing. 
\end{abstract}
\thispagestyle{empty}
\newpage
\thispagestyle{empty}
\tableofcontents
\newpage
\pagenumbering{arabic}
\section{Introduction}
Detecting small subgraphs in distributed networks has recently attracted significant research interest~\cite{bonne2019distributed,CensorCLL21,Censor+PODC20,chang2024deterministic,ChangPSZ21,Drucker+PODC14,Eden+DISC19,Fischer+SPAA18,fraigniaud2019distributed,fraigniaud2016distributed,Korhonen+OPODIS17,Izumi+PODC17,liu2022noteimprovedresultsround}. Distributed subgraph finding plays an important role in understanding the \emph{bandwidth limitation} in distributed networks: It is a classical problem where a straightforward $O(1)$-round solution exists with unlimited bandwidth, but becomes significantly more complex when bandwidth constraints are imposed.

Previous works on distributed subgraph findings mostly assumed a model in which the underlying network is \emph{static}. However, distributed systems in real life may undergo topological changes over time: A node might crash, and a new connection might be formed between two existing nodes. To address this gap, Bonne and Censor-Hillel~\cite{bonne2019distributed} initiated the study of distributed subgraph finding
in dynamic networks to better capture the real-world behavior of networks of {limited bandwidth}. For the case where the target subgraph is a clique,
they determined the tight bandwidth complexity bounds in nearly all settings. Later, Liu~\cite{liu2022noteimprovedresultsround} extended this study to dynamic graphs with batched updates and resolved an open question of Bonne and Censor-Hillel~\cite[Open question 4]{bonne2019distributed}. However,
several open questions remain, and very little is known about finding subgraphs beyond
cliques. In this paper, we build upon their works~\cite{bonne2019distributed,liu2022noteimprovedresultsround} to consider the remaining open questions and explore other target subgraphs beyond cliques.

\subsection{Models}\label{subsec:model}
We now formally describe the models considered in this paper, which were introduced by Bonne and Censor-Hillel~\cite{bonne2019distributed}.
 A dynamic network $\mathcal{G}$ is a sequence of graphs $\mathcal{G}=\{G^0, G^1, \ldots \}$. The superscript notation should not be confused with graph powers. The initial graph $G^0$ represents the state of the network at some starting point. For each $i > 1$, the graph $G^i$ is either identical to its preceding graph $G^{i-1}$ or obtained from $G^{i-1}$ by a single topological change. %Let $G^i$ denote the final graph structure after the topological change (if exists) at round $i$.

    Each node $u$ in the network is equipped with a unique identifier $\ID(u)$, and it knows the list of identifiers of all its neighbors. 
    The communication is synchronous. In each round of communication, each node can send to each of its neighbors a message of $B$ bits, where $B$ denotes the \emph{bandwidth} of the network.

We assume that each node initially knows the \emph{entire} topology of the initial graph $G^0$. In particular, the set of all identifiers is global knowledge, so we may assume that the range of the identifiers is exactly $[n]$, where $n$ is the number of nodes in the network.

\paragraph{Topological changes} We consider four types of topological changes: node insertions, node deletions, edge insertions, and edge deletions. In the case of a node insertion, the adversary may connect the new node to an arbitrary subset of the existing nodes.
    Each node $u$ can only \emph{indirectly} deduce that a topological change has occurred by comparing its list of neighbors $N_{G^i}(u)$ in the current round $i$ and its list of neighbors $N_{G^{i-1}}(u)$ in the previous round $i-1$. At most one topological change can occur in each round. Suppose at some round $i$, node $u$ detects that exactly one new node $v$ appears in its neighborhood list, then from this information only, node $u$ cannot distinguish whether edge $\{u,v\}$ is added or node $v$ is added in round $i$, if we are in the model where both edge insertions and node insertions are allowed. Similarly, suppose node $u$ detects that exactly one node $v$ disappears in its list. In that case, $u$ cannot distinguish whether edge $\{u,v\}$ is deleted or node $v$ is deleted, if we are in the model where both edge deletions and node deletions are allowed.
    
\paragraph{Algorithms} An algorithm can be designed to handle only one type of topological change or any combination of them. Throughout this paper, we only consider \emph{deterministic} algorithms. We say an algorithm $\mathcal{A}$ is an $r$-round algorithm if $\mathcal{A}$ works in the following setting.
\begin{itemize}
    \item Each topological change is followed with at least $r-1$ \emph{quiet rounds}. Specifically, if a topological change occurs in round $i$, then we must have $G^{i} = G^{i+1} = \cdots = G^{i+r-1}$. Rounds $i+1, \ldots, i+r-1$ are quiet in the sense that no topological changes occur in these rounds.
    \item The output w.r.t.~$G^i$ must be computed correctly by round $i+r-1$.
\end{itemize}    

A one-round algorithm is not called a zero-round algorithm because, within each round, topological change occurs before any communication takes place. This setup allows the nodes to have one round of communication between the topological change and deciding the output within a single round. For example, if an edge $\{u,v\}$ is inserted in a round, then $u$ and $v$ can immediately communicate with each other along this edge within the same round.

We emphasize that all our algorithms and lower bounds in this paper are \emph{deterministic}, although many of our lower bounds also extend to the randomized setting, see \Cref{sect:random}.

\subsection{Problems}\label{subsec:problems}

We consider the four types of distributed subgraph finding problems introduced by Bonne and Censor-Hillel~\cite{bonne2019distributed}.

\begin{description}
    \item[Membership Listing] For the membership-listing ($\MemList(H)$) problem, each node $v$ lists all the copies of $H$ containing $v$. In other words, for each copy of $H$ and each node $u$ in $H$, node $u$ lists $H$. 
    \item[Membership Detection] For the membership-detection ($\MemDetect(H)$) problem,~ each node $v$ decides whether $v$ belongs to at least one copy of $H$. 
    \item[Listing] For the listing ($\List(H)$) problem, every appearance of $H$ is listed by at least one node in the network. In other words, for each copy of $H$, there exists some node $u$ that lists $H$.
    \item[Detection] For the detection ($\Detect(H)$) problem, the existence of any $H$ must be detected by at least one node. Specifically, if the network does not contain $H$ as a subgraph, then all nodes must output $\No$. Otherwise, at least one node must output $\Yes$.
\end{description}

For both membership-detection and detection, the output of each node is binary ($\Yes/\No$) only,  with no requirement to report the actual target subgraph. For both membership-listing and listing, each node outputs a list of the target subgraphs, using the node identifiers in the network. For the listing problem, the node $u$ responsible for listing $H$ is not required to be in $H$, and it is allowed that each copy of $H$ is listed by multiple nodes.

In the literature, the problem of deciding whether a subgraph isomorphic to $H$ exists is often referred to as \emph{$H$-freeness}, whereas \emph{detection} sometimes denotes the task of outputting a copy of $H$. We emphasize that our use of detection differs from this convention: In both $\MemDetect(H)$ and $\Detect(H)$, outputting a copy of $H$ is not required.

 %We use $\MemList(H), \MemDetect(H),\List(H)$ and $\Detect(H)$ to denote the problems of membership listing, membership detection, listing and detection of subgraph $H$. 
    
\paragraph{Bandwidth complexities} The $r$-round \emph{bandwidth complexity} of a problem is defined as the minimum bandwidth $B$ for which there exists an $r$-round algorithm solving that problem with bandwidth $B$.
Fix a target subgraph $H$, a round number $r$, and some type of topological changes.
Let $B_{\MemList}$, $B_{\MemDetect}$, $B_{\List}$, and $B_{\Detect}$ denote the  $r$-round bandwidth complexities of $\MemList(H)$, $\MemDetect(H)$, $\List(H)$, and $\Detect(H)$, respectively. The following observations were made by Bonne and Censor-Hillel~\cite{bonne2019distributed}.
\begin{observation}[{\cite{bonne2019distributed}}]
Given any target subgraph $H$ and any integer $r$, 
\[B_{\Detect}\leq B_{\List}\leq B_{\MemList}\] under any type of topological change.
\end{observation}
\begin{observation}[{\cite{bonne2019distributed}}]\label{obs: detect < memdetect < memlist}
Given any target subgraph $H$ and any integer $r$, 
\[B_{\Detect}\leq B_{\MemDetect}\leq B_{\MemList}\] under any type of topological change. 
\end{observation}

\paragraph{Nontrivial target subgraphs} In this paper, we only focus on \emph{nontrivial} target subgraphs $H$, meaning that we implicitly assume that $H$ is \underline{connected and contains at least three nodes}: %\yanyu{it's not really impossible to solve. if G is connected, we can find the connected components and gather them.}\yijun{should be impossible if the network $G$ is allowed to be disconnected, which should be the case in our setting.}
 \begin{itemize}
     \item If $H$ is not connected, then all four problems are trivially impossible to solve, as we allow the network to be disconnected.
     \item If $H$ is connected with exactly two nodes, then all four problems are trivially solvable with zero communication.
 \end{itemize}   

For example, when we say that the one-round bandwidth complexity of $\MemDetect(K_s)$ under edge insertions is $\Omega(\log \log n)$, we implicitly assume that $s \geq 3$.

\subsection{Our Contributions}
%Our contribution is two-fold: We first investigate the remaining open questions of Bonne and Censor-Hillel~\cite{bonne2019distributed} on clique finding and then explore other target subgraphs beyond cliques.
While Bonne and Censor-Hillel~\cite{bonne2019distributed} settled most of the bandwidth complexity bounds for clique finding, they left five open questions, one of which~\cite[Open question 4]{bonne2019distributed} was resolved by Liu~\cite{liu2022noteimprovedresultsround}. In this paper, we investigate the remaining ones.

\paragraph{Finding cliques under edge insertions} In \cite[Open question 1]{bonne2019distributed} and  \cite[Open question 3]{bonne2019distributed}, Bonne and Censor-Hillel asked for the tight bound on the one-round bandwidth complexity of membership-detection for triangles and larger cliques under \emph{edge insertions} only. %\mingyang{under edge insertions}\yijun{fixed} 
For triangles, they obtained two upper bounds $O(\log n)$ and $O(\sqrt{\Delta \log n})$~\cite{bonne2019distributed}. For larger cliques, they obtained an upper bound $O(\sqrt{n})$ which works even for the membership-listing problem~\cite{bonne2019distributed}. %\mingyang{Note: they use the upper bound from membership-listing directly}. 
In this work, we show that these problems admit a bandwidth lower bound of $\Omega(\log \log n)$, which holds even in bounded-degree networks.

\begin{restatable}{thm}{loglogLB}\label{lem: memdect lb ks}\label{lem: memdect lb tri}
    For any constant $s \geq 3$, the one-round bandwidth complexity of $\MemDetect(K_s)$ under edge insertions is $\Omega(\log \log n)$, even in bounded-degree dynamic networks.   
\end{restatable}

Prior to our work, no lower bound was known for this problem. Moreover, we complement our lower bound with a new $O(\log \log n)$-bandwidth triangle finding algorithm in bounded-degree networks, which is capable of solving the more challenging problem of membership-listing.

\begin{restatable}{thm}{loglogUB}\label{thm: Delta memlist ub loglogn}
There exists a one-round algorithm of $\MemList(K_3)$ under edge insertions with bandwidth $O(\Delta^2 \log \log n)$, where $\Delta$ is the maximum degree of the dynamic network.
\end{restatable}

Combining \Cref{lem: memdect lb ks} and \Cref{thm: Delta memlist ub loglogn}, we obtain a tight bound of membership-detection for triangles in bounded-degree dynamic networks.

%According to \Cref{obs: detect < memdetect < memlist} and \Cref{thm: Delta memlist ub loglogn}, for dynamic networks with bounded degree, the one-round bandwidth of $\MemDetect(K_3)$ under edge insertions is $O(\log\log n)$. The upper bound of \Cref{thm: Delta memlist ub loglogn} matches the lower bound $\Omega(\log \log n)$ of \Cref{lem: memdect lb ks}.
\begin{corollary}\label{cor-1}
     The one-round bandwidth complexity of $\MemDetect(K_3)$ under edge insertions is $\Theta(\log\log n)$ in bounded-degree dynamic networks.
\end{corollary}

\paragraph{Finding cliques under edge insertions and node insertions} The remaining two open questions of Bonne and Censor-Hillel~\cite{bonne2019distributed}  considered the model where two types of topological changes are allowed. Specifically, \cite[Open question 2]{bonne2019distributed} and  \cite[Open question 5]{bonne2019distributed} asked for the tight bound on the one-round bandwidth complexity of the listing problem for triangles and larger cliques under both \emph{edge insertions} and \emph{node insertions}.  %\mingyang{under edge insertions and node insertions}\yijun{fixed}
Previously, these problems were known to have an upper bound of $O(\log n)$~\cite{bonne2019distributed}.  In this work, we show that these problems admit a bandwidth lower bound of $\Omega(\log \log \log n)$, which applies to the easier problem of detection and holds in bounded-degree networks.

\begin{restatable}{thm}{logloglogLB}\label{thm:mixed}
       For any constant $s \geq 3$, the one-round bandwidth complexity of $\Detect(K_s)$ under both node insertions and edge insertions is $\Omega(\log \log \log n)$, even in bounded-degree dynamic networks. 
\end{restatable}

Same as \Cref{lem: memdect lb ks}, prior to our work, no lower bound was known for this problem. Interestingly, we are also able to match this lower bound with a new $O(\log \log \log n)$-bandwidth algorithm for listing triangles in bounded-degree networks. 

\begin{restatable}{thm}{logloglogUB}\label{thm: Delta list lb logloglogn}
There exists a one-round algorithm of $\List(K_3)$ under edge insertions and node insertions with bandwidth $O(\Delta\log\log \log n)$, where $\Delta$ is the maximum degree of the dynamic network.
\end{restatable}

Combining \Cref{thm:mixed} and \Cref{thm: Delta list lb logloglogn}, we obtain a tight bound for triangle detection in bounded-degree dynamic networks.

%According to \Cref{obs: detect < memdetect < memlist} and \Cref{thm: Delta list lb logloglogn}, for dynamic networks with bounded degree, the one-round bandwidth of $\Detect(K_3)$ under edge insertions and node insertions is $O(\log\log\log n)$. The upper bound of \Cref{thm: Delta list lb logloglogn} matches the lower bound $\Omega(\log \log \log n)$ of \Cref{thm:mixed}.
\begin{corollary}\label{cor-2}
    The one-round bandwidth complexity of $\Detect(K_3)$ under edge insertions and node insertions is $\Theta(\log\log\log n)$ in bounded-degree dynamic networks.
\end{corollary}

%\Cref{lem: memdect lb ks} makes progress toward answering \cite[Open question 1]{bonne2019distributed} and  \cite[Open question 3]{bonne2019distributed}, which asks for the tight bound on the  one-round bandwidth complexity of $\MemDetect(H)$ for $H = K_3$ and $H = K_s$.  

%\Cref{thm:mixed} considers the model where two types of topological changes, node insertions and edge insertions, are allowed. As any lower bound for  $\Detect(H)$ is automatically a lower bound for $\List(H)$, similarly, \Cref{thm:mixed} makes progress toward answering \cite[Open question 2]{bonne2019distributed} and  \cite[Open question 5]{bonne2019distributed}, which asks for the tight bound on the  one-round bandwidth complexity of $\List(H)$ for $H = K_3$ and $H = K_s$. Before our work, no lower bound was known for the two problems in \Cref{lem: memdect lb ks,thm:mixed}.  

%For finding cliques, we establish an $\Omega(\log \log n)$ bandwidth lower bound for one-round membership-detection under edge insertions only and an $\Omega(\log \log \log n)$ bandwidth lower bound for one-round detection under both edge insertions and node insertions. Moreover, we demonstrate new algorithms to show that our lower bounds are \emph{tight} in bounded-degree networks when the target subgraph is a triangle. Prior to our work, no lower bounds were known for these problems.

\paragraph{Beyond cliques} It appears to be a challenging task to extend the current results beyond cliques. In the static setting, while the round complexity for $k$-clique listing has been settled~\cite{CensorCLL21,Censor+PODC20,chang2024deterministic,ChangPSZ21,Eden+DISC19,Izumi+PODC17}, far less is known about target subgraphs that are not cliques. In the dynamic setting, Bonne and Censor-Hillel~\cite{bonne2019distributed} highlighted that cliques are \emph{unique} in that they can be found trivially in \emph{one round} if bandwidth is unrestricted. 

In this work, we demonstrate that it is possible to achieve meaningful results beyond cliques, despite the inherent difficulties.
For subgraph finding beyond cliques, we present a \emph{complete characterization} of the bandwidth complexity for the membership-listing problem across all target subgraphs, all numbers of rounds, and all four types of topological changes: node insertions, node deletions, edge insertions, and edge deletions (\Cref{table:memList1}). Moreover, we show partial characterizations for one-round membership-detection (\Cref{tab:memDect}) and listing (\Cref{tab:listing}).%\yijun{Maybe discuss these results a bit more here -- I will do this once I finish editing all the remaining sections.}

Our contribution lies in finding structures in the apparent chaos: We identify relevant graph classes (e.g., complete multipartite graphs) and parameters (e.g., node-edge versions of distance, radius, and diameter) in the area of dynamic distributed subgraph finding. For the cases where a full characterization has yet to be achieved, we identify remaining challenging open problems and outline future research directions. Addressing these challenges will likely require developing novel techniques.

\subsection{Additional Related Work} 
K{\"{o}}nig and Wattenhofer~\cite{koenig2013local} conducted a systematic study of classical network problems under various types of topological changes, such as node or edge insertions and deletions. They said that a combination of a problem and a topological change is \emph{locally fixable} if an existing solution can be repaired in $O(1)$ rounds following a topological change in the network.
 Several subsequent studies have investigated dynamic distributed algorithms for local problems, such as independent set~\cite{assadi2018fully,censor2016optimal}, matching~\cite{solomon2016fully}, and coloring~\cite{parter2016local}. The dynamic-$\LOCAL$ model was formalized in~\cite{akbari2021locality}. Distributed algorithms for highly dynamic networks, where multiple topological changes can occur within a single communication round, were examined in~\cite{bamberger2019local,censorhillel_et_al:LIPIcs.OPODIS.2020.28}. Global distributed problems, such as consensus and information dissemination, have also been studied in the dynamic setting~\cite{dinitz2018smoothed,dutta2013complexity,haeupler2011faster,jahja2020sublinear,kuhn2010distributed,kuhn2011coordinated,yu2018cost}. Dynamic distributed algorithms are closely related to the concept of \emph{self-stabilization}---a key notion in distributed computing---where a distributed network undergoes various changes and must rapidly return to a stable state after some quiet time~\cite{selfstablizationbook}.

%to do: distributed subgraph finding:
%congested clique
%congest
%property testing
%mpc
There is a long line of research studying distributed subgraph finding in the $\CONGEST$ model. The first breakthrough in triangle detection was achieved by Izumi and Le Gall~\cite{Izumi+PODC17}, who demonstrated that triangle detection and listing can be completed in $\tilde{O}(n^{2/3})$ and $\tilde{O}(n^{3/4})$ rounds, respectively. After a series of work~\cite{CensorCLL21,Censor+PODC20,Censor2022deterministic,chang2024deterministic,ChangPSZ21,Eden+DISC19}, it was shown that $\tilde{\Theta}(n^{1-2/k})$ is the tight bound for $k$-clique listing via the use of expander decomposition and routing. Distributed cycle findings have also been extensively studied~\cite{Drucker+PODC14,Eden+DISC19,Fischer+SPAA18,Korhonen+OPODIS17}.
Property-testing variants of the distributed subgraph finding problem have been explored~\cite{Brakerski2009distributed,Censor2019fast,EvenFFGLMMOORT17,fraigniaud2019distributed,fraigniaud2016distributed}. The subgraph finding problem has also been investigated in other computational models beyond distributed and parallel computing~\cite{Alon+SIDMA08,becchetti+KDD08,Eden+SICOMP17}.
 For more on the distributed subgraph finding problem, see the survey by Censor-Hillel~\cite{censorhillel:LIPIcs.ICALP.2021.3}. 

The study of distributed subgraph finding is partially motivated by the fact that many algorithms can be significantly improved if the underlying network does not contain certain small subgraphs. %\yijun{While this is indeed one motivation, this is not really a main one, maybe try editing this later} 
For instance, Pettie and
Su showed distributed coloring algorithms for triangle-free graphs using fewer colors and rounds \cite{PETTIE2015263}. Similarly, Hirvonen, Rybicki, Schmid, and Suomela developed an efficient distributed algorithm to find large cuts in triangle-free graphs~\cite{hirvonen2014largecutslocalalgorithms}. 
%\yijun{discuss prior work on distributed subgraph finding and dynamic distributed graph algorithms}

\subsection{Roadmap}
In \Cref{sect:prelim}, we review essential graph terminology and parameters and demonstrate how certain graph parameters determine the minimum number of rounds required for certain subgraph-finding tasks. In \Cref{sect:overview}, we present a technical overview of our proofs. In \Cref{sect:cliqueLB}, we show lower bounds for finding cliques. In \Cref{sect:trianglesUB}, we show upper bounds for finding triangles.
In \Cref{sect:memList}, we show a complete characterization of the bandwidth complexity of the membership-listing problem for every target subgraph, every round number, and every type of topological change.
In \Cref{sect:memDect}, we show a partial characterization for one-round membership-detection.
In \Cref{sect:list}, we show a partial characterization for one-round listing. 
In \Cref{sect:conclusions}, we conclude the paper with some open questions.
In  \Cref{sect:random}, we show that many of our lower bounds can be extended to the randomized setting.

%\yijun{write paper organization}

\section{Preliminaries}\label{sect:prelim}
In this section, we present the basic definitions used in the paper.
In \Cref{subsec:terminology}, we review essential graph terminology and parameters. 
In \Cref{sect:properties}, we discuss some basic properties of these graph parameters.
In \Cref{subsec: local constraints}, we demonstrate how specific graph parameters determine the minimum number of rounds required for $\MemList(H)$ and $\MemDetect(H)$, independent of any bandwidth constraints.

%and techniques we used in this paper.

%In \Cref{subsec: local constraints}, we investigate the locality constraints of Membership Listing and Membership Detection problem, which is the minimum round requirement for $\MemList(H)$ and $\MemDetect(H)$ regardless any bandwidth constraints. 

\subsection{Graph Terminology}\label{subsec:terminology}
Given a graph $H$, an edge $e=\{u,v\}\in E(H)$, and a node subset $S \subseteq V(H)$, we write $H[S]$ %\yijun{Can we write $H[S]$?}\mingyang{Updated}
to denote the subgraph of $H$ induced by node set $S$, write $H-S$ to denote the subgraph of $H$ induced by node set $V(H)\setminus S$, and write $H-e$ to denote the subgraph of $H$ induced by edge set $E(H)\setminus \{e\}$.
%\gopinath{May be we use $``\setminus"$ instead of $``-"$ through out the paper.}\yanyu{different meaning here}
% \gopinath{May be we expand expand the list of notations used in this paper in this section?}

For a graph $G=(V, E)$, we have the following definitions.
\begin{definition}[Eccentricity, diameter, radius, and center]\label{def:parameters} %We define Eccentricity, diameter, radius, and center as follows.
    \begin{align*}
    &\text{Eccentricity of $u$:} \ & \ecc_G(u)&=\max\{\dist_G(u,v): v\in V\}\\
        &\text{Diameter of $G$:} \ & \diam(G) &= \max \{\ecc_G(u): u\in V\}\\
        &\text{Radius of $G$:} \ & \rad(G) &= \min \{\ecc_G(u): u\in V\}\\
         &\text{Center of $G$:} \ &  \centre(G) &= \{u \in V: \ecc_G(u) = \rad(G)\}
\end{align*}
\end{definition}

In other words, the eccentricity of a node is the maximum of its distance to other nodes, the diameter of a graph is the maximum eccentricity of its nodes, the radius of a graph is the minimum eccentricity of its nodes, and the center of a graph is a node subset containing all nodes with eccentricity equal to the radius.

% \begin{definition}[Eccentricity]
% The eccentricity of a node is the maximum of its distance to other nodes:
% $\ecc_G(u)=\max\{\dist_G(u,v): v\in V\}$.
% \end{definition}
% \begin{definition}[Diameter and Radius]
%     The diameter of $G$ is the maximum eccentricity of its nodes:
%     $\diam(G) = \max \{\ecc_G(u): u\in V\}$.
%     The radius of $G$ is the minimum eccentricity of its nodes:
%     $\rad(G) = \min \{\ecc_G(u): u\in V\}$.
% \end{definition}
% \begin{definition}[Center]
%     The center of $G$ is a node subset containing all nodes with eccentricity equal to the radius:
%     $\centre(G) = \{u \in V: \ecc_G(u) = \rad(G)\}$.
% \end{definition}

We consider the ``node-edge'' version of the definitions above.

\begin{definition}[Node-edge distance]%\mingyang{$\dist_G$}\yijun{changed}
    For any $u \in V$ and any $e=\{v,w\} \in E$, the node-edge distance between them is defined as $\dist_G(u, e) := 1+\min \left\{ \dist_G(u, w), \dist_G(u,v)\right\}$.
\end{definition}

We overload the notation $\dist(\cdot,\cdot)$ since it is easy to distinguish $\dist(u,v)$ and $\dist(u, e)$.

\begin{definition}[Node-edge eccentricity, diameter, radius, and center]\label{def:NE_parameters} %We define Eccentricity, diameter, radius, and center as follows.
    \begin{align*}
    &\text{Node-edge eccentricity of $u$:} \ & \neecc_G(u)&=\max\{\dist_G(u,e): e\in E\}\\
        &\text{Node-edge diameter of $G$:} \ & \nediam(G) &= \max \{\neecc_G(u): u\in V\}\\
        &&&=\max_{u\in V} \max_{e \in E} \dist_G(u,e)\\
        &\text{Node-edge radius of $G$:} \ & \nerad(G) &= \min \{\neecc_G(u): u\in V\}\\
        &&&=\min_{u\in V} \max_{e \in E} \dist_G(u,e)\\
         &\text{Node-edge center of $G$:} \ &  \necenter(G) &= \{u \in V: \neecc_G(u) = \nerad(G)\}\\
         &&&= \{u \in V: \max_{e \in E} \dist_G(u,e) = \nerad(G)\}
\end{align*}
\end{definition}

% \begin{definition}[Node-Edge Radius]
%     $\nerad(G) := \min_{v\in V} \max_{e \in E} \dist_G(v,e)$.
% \end{definition}

% \begin{definition}[Node-Edge Center]
%     The node-edge center of $G$ is a node subset containing all nodes $u$ such that the node-edge distance $\dist_G(u, e)$ between $u$ and the furthest edge $e$ from $u$ is equal to the node-edge radius
%     $\necenter(G) = \{u \in V: \max_{e \in E} \dist_G(u,e) = \nerad(G)\}$.
% \end{definition}

The node-edge distance definition is needed to capture the complexity bounds for various problems discussed in this paper. For example, in $C_4$, any node can detect an edge deletion within one round, whereas in $C_5$, the node opposite the deleted edge cannot detect the topological change in a single round. This is reflected by the fact that $\nerad(C_4) = 2 < 3 = \nerad(C_5)$, while the standard radius definition $\rad(C_4) = 2 = \rad(C_5)$ is insufficient to capture the difference.

\subsection{Properties of Graph Parameters}\label{sect:properties}

% \mingyang{shall we discuss the relation between $\nediam$ and $\diam$ here? It is $\diam\leq \nediam\leq\diam +1$ due to \Cref{obs: r_H r_H'} \Cref{obs:r_H=nediam(H)-1}
% \Cref{obs: r_H'}. Maybe we should change the presentation.}
%\yijun{I rewrote this paragraph.}\yijun{Note: the presentation of the above definitions is really messy -- I will think about how to make them look cleaner}

%The motivation for this node-edge distance definition arises from our observations of complexity differences between $C_4$ and $C_5$ across various problems discussed in this paper. This distinction is exemplified by the fact that $\rad(C_4)=2=\nerad(C_4)$, while $\rad(C_5)=2\neq 3=\nerad(C_5)$. This difference has practical implications: in $C_4$, any node can detect an edge change within 1 round, whereas in $C_5$, the node opposite the deleted edge cannot detect the change in a single round.

These distance-based parameters and their node-edge versions are closely related. Specifically, we have the following observation.
%\yanyu{if okay, remove the color}
\begin{observation}\label{obs: diam nediam}
For any connected graph $H$, $\diam(H) \leq \nediam(H)\leq\diam(H) +1$.
\end{observation}
\begin{proof}
Since for any node $u$ and any edge $\{v,w\}$, $|\dist_H(u, w)-\dist_H(u,v)| \leq 1$, the observation follows directly from the definition of $\diam$ and $\nediam$. 
%\yanyu{this proof is quite trivial feel free to shorten it}
\begin{align*}
    \diam(H)&=\max_{u\in V} \max_{v \in V} \dist_H(u,v)\\
    % &= \max_{u\in V} \max_{\{v,w\} \in E} 1+\min \left\{ \dist(u, w), \dist(u,v)\right\}\\
    &\leq \max_{u\in V} \max_{\{v,w\} \in E} 1+\min \left\{ \dist_H(u, w), \dist_H(u,v)\right\}= \nediam(H)\\
    &\leq 1 + \max_{u\in V} \max_{v \in V} \dist_H(u,v)=\diam(H)+1.
    \qedhere
\end{align*}
% $$\nediam(H)
%     = \max_{u\in V} \max_{\{v,w\} \in E} 1+\min \left\{ \dist(u, w), \dist(u,v)\right\}
%     \geq \max_{u\in V} \max_{v \in V} \dist(u,v)=\diam(H)$$
% $$\nediam(H)
%     = \max_{u\in V} \max_{\{v,w\} \in E} 1+\min \left\{ \dist(u, w), \dist(u,v)\right\}
%     \leq 1 + \max_{u\in V} \max_{v \in V} \dist(u,v)=\diam(H)+1$$
\end{proof}

The class of \emph{complete multipartite graphs} plays a crucial role in the study of the complexity landscape of dynamic distributed subgraph finding, as it can be characterized in terms of node-edge diameter or node-edge independence.

\begin{definition}[Complete multipartite graphs]
    A graph $G$ is \underline{complete $k$-partite} if its node set $V(G)$ can be partitioned into $k$ independent sets $S_1, S_2, \ldots, S_k$ such that for any two nodes $u,v \in V(G)$, $\{u,v\} \in E(G)$ if and only if $u \in S_i$ and $v \in S_j$ for some $i \neq j$. For any $k$, a complete $k$-partite graph is called a \underline{complete multipartite} graph. 
%\yanyu{shorten definition a little bit}
    % \begin{enumerate}
    %     \item Each $S_i$ is an independent set, i.e. for all nodes $u,v \in S_i$, $\{u,v\} \notin E(G)$.
    %     \item Any pair of distinct parts $S_i$, $S_j$ induces a complete bipartite graph, i.e. for all nodes $u,v \in G$, if $u \in S_i$, $v \in S_j$, and $i \neq j$, then $\{u,v\} \in E(G)$.
    % \end{enumerate}
\end{definition}

% \yanyu{maybe better to put this at 2.1}
\begin{definition} [Node-edge independence]\label{def:neindependent}
    In a graph $G$, a node $w$ and an edge $\{u,v\}$ are \underline{independent} if $w$ is adjacent to neither $u$ nor $v$ in $G$.
    A graph $G$ is \underline{node-edge independent} if {at least one} of its node-edge pairs are independent.
    % if there exist a node $w \in V(G)$ and an edge $e = \{u,v\} \in E(G)$ such that $w$ is adjacent to neither $u$ nor $v$. 
\end{definition}

In other words, a graph $G$ is node-edge independent if and only if it contains a co-$P_3$ (complement of a path with three nodes) as an induced subgraph. See \Cref{fig:co-p3}.

\begin{lemma} \label{lem: P3 free}
    A graph $G$ has no induced $P_3$ if and only if its connected components are cliques.
\end{lemma}
\begin{proof}
%\yanyu{now I am a little bit more inclined to omit the proof of \ref{lem: P3 free} and \ref{lem: multipartite iff not node-edge indp}}
    If every connected component of $G$ is a clique, then any three nodes selected cannot induce a $P_3$. 
    Conversely, suppose a graph $G$ has no induced $P_3$ and there is a connected component that is not a clique, then there is a pair of disconnected nodes $u$ and $v$ in the same component. The two nodes $u$ and $v$ must be connected by some minimum-length path $P_{uv}$ whose length is at least $2$, since they are in the same connected component. Any three consecutive nodes on the path induce a $P_3$, which is a contradiction. 
\end{proof}

\begin{lemma} \label{lem: multipartite iff not node-edge indp}
    A graph $G$ is complete multipartite if and only if it is \emph{not} node-edge independent.
\end{lemma}

\begin{proof}
    Observe that a graph is complete multipartite if and only if its complement is a disjoint union of cliques. Hence \Cref{lem: P3 free} implies that a graph $G$ is \emph{not} node-edge independent $\iff$ $G$ does not contain a co-$P_3$ as an induced subgraph $\iff$ complement of $G$ is $P_3$-free $\iff$ $G$ is complete multipartite.
\end{proof}

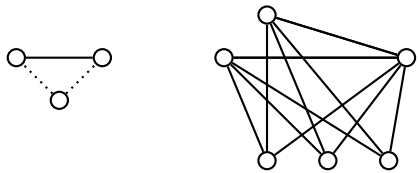
\begin{figure}[!htbp]
    \centering
    \begin{tikzpicture}[node distance={8mm}, thick, main/.style = {draw, circle, inner sep = .8mm}] 
        \node[main] (1) {}; 
        \node[main] (2) [above left of=1] {}; 
        \node[main] (3) [above right of=1] {}; 
        \draw[dotted](1) -- (2); 
        \draw[dotted](1) -- (3); 
        \draw(2) -- (3); 

        \node(4) [right of=3] {};
        \node[main] (11) [right of=4] {}; 
        \node[main] (12) [above right of=11] {}; 
        \node(13) [right of=11] {};
        \node(133) [right of=13] {};
        \node[main] (21) [right of=133] {};
        \node(14) [below right of=11] {}; 
        \node[main] (31) [below of=14] {}; 
        \node[main] (32) [right of=31] {}; 
        \node[main] (33) [right of=32] {}; 

        \draw(11) -- (21); 
        \draw(11) -- (31); 
        \draw(11) -- (32); 
        \draw(11) -- (33); 

        \draw(12) -- (21); 
        \draw(12) -- (31); 
        \draw(12) -- (32); 
        \draw(12) -- (33); 

        \draw(21) -- (11); 
        \draw(21) -- (12); 
        \draw(21) -- (31); 
        \draw(21) -- (32); 
        \draw(21) -- (33); 
    \end{tikzpicture}
    \caption{Co-$P_3$ (left, dotted lines denote non-edges) and a complete multi-partite graph (right).}
    \label{fig:co-p3}
\end{figure}

% Now we are ready to prove the following theorem.
Now we characterize complete multipartite graphs in terms of node-edge diameter.

\begin{observation} \label{thm: nediam=2}
    For any connected graph $H$ with at least three nodes, $\nediam(H)=2$ if and only if $H$ is complete multipartite.
% \yanyu{I feel that as a theorem it is better to include our implicit assumption here}
\end{observation}
\begin{proof}
    % By \Cref{obs:r_H=nediam(H)-1} that $r_H = \nediam(H) - 1$, it suffices to show that $\nediam(H) = 2$ if and only if $H$ is complete multipartite. 
    If $\nediam(H) = 2$, then $H$ is not node-edge independent, since if node $w$ is independent of edge $e=\{u,v\}$, then $\dist_H(w, e)\geq 3$. Hence, $H$ is complete multipartite by \Cref{lem: multipartite iff not node-edge indp}. 
    % By definition of $\nediam$, for any node $w \in V(H)$ and edge $e = \{u,v\} \in E(H)$, $w$ is either one of $u$ or $v$, or adjacent to at least one of them. Hence $H$ is not node-edge independent and $H$ is complete multipartite by \Cref{lem: multipartite iff not node-edge indp}. 
    Conversely, suppose $H$ is complete multipartite, then it is not node-edge independent, and $\dist_H(w, e)\leq 2$ for all $w\in V(H)$ and $e\in E(H)$. Hence, $\nediam(H)\leq 2$. Since $H$ has at least three nodes, for any edge $e$, there is a node $w$ such that $w\notin e$. We have $\dist_H(w,e)\geq 2$ and $\nediam(H)\geq 2$. 
    %\yanyu{shorten a lot, please check coherence.}
    % The proof is complete.
    % For the converse, suppose $H$ is a complete multipartite graph. Consider a node $w \in V(H)$ and an edge $e = \{u,v\} \in E(H)$ where $w,u,v$ are distinct (they must exist since $H$ is connected with at least three nodes). By definition of complete multipartite, $u$ and $v$ must be in different parts. Hence $w$ must be in a different parts from at least one of them. Without loss of generality, suppose $w$ and $u$ are in different parts. Then $\{w,u\} \in E(H)$ so $\dist(w, e) = 2$. Since our choices of $w$ and $e$ were arbitrary, we must have $\nediam(H) = 2$.
\end{proof}

\subsection{Locality Constraints}\label{subsec: local constraints}

In this section, we investigate the minimum number of rounds required for the two problems $\MemList(H)$ and $\MemDetect(H)$ regardless of any bandwidth constraints, for any given target subgraph $H$. 
Due to the \emph{local} nature of topological changes, the appearance or disappearance of a copy of $H$ might not be detectable for some nodes in $H$ in the same round or even after several rounds of communication. 

See \Cref{fig: locality constraints} for an example. Suppose at some round $r$, edge $\{u,v\}$ is inserted, so the current graph is $C_5$. Assuming unlimited bandwidth, after one round of messaging, all the nodes highlighted can detect the appearance of $C_5$ due to the messages from $u$ or $v$. The only remaining node $w$ is unable to detect the appearance of $C_5$ after one round of messaging, as the information about the topological change has not reached $w$.

%Suppose at round $r$, the first copy of $H$ occurs after edge $\{u,v\}$ is inserted. For nodes $u$ and $v$, they may be able to detect the appearance of $H$. However for some other nodes, they may not be able to detect it even after one round of communication. See \Cref{fig: locality constraints} for illustration.

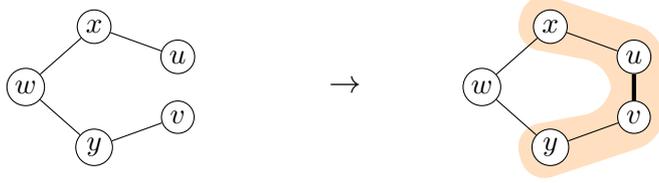
\begin{figure}[htbp]
    \centering
    \begin{tikzpicture}[
vtx/.style = {draw, circle, inner sep=0.6mm, fill=white}, 
bg/.style={fill=orange!25, draw=none}, 
x=1mm, y=1mm
]
    \begin{scope}        
        \node[vtx] (w) at (-2,0) {$w$};
        \node[vtx] (x) at (7,8) {$x$};
        \node[vtx] (y) at (7,-8) {$y$};
        \node[vtx] (u) at (18,4) {$u$};
        \node[vtx] (v)at (18,-4) {$v$};
        \draw (x)--(w);
        \draw (y)--(w);
        \draw (x)--(u);
        \draw (y)--(v);
    \end{scope}

    \node[font=\large] at (40,0) {$\rightarrow$};

    \begin{scope}[xshift=60mm]
    \filldraw[bg, rounded corners=3.5mm] 
            (4, 13) -- (22,7) -- (22,-7) -- (4,-13) -- (2, -5.5) -- (15,-3) -- (15,3) -- (2,5.5) -- cycle;
        \node[vtx] (w) at (-2,0) {$w$};
        \node[vtx] (x) at (7,8) {$x$};
        \node[vtx] (y) at (7,-8) {$y$};
        \node[vtx] (u) at (18,4) {$u$};
        \node[vtx] (v)at (18,-4) {$v$};
        \draw (x)--(w);
        \draw (y)--(w);
        \draw (x)--(u);
        \draw (y)--(v);
        \draw[ultra thick] (u)--(v);
    \end{scope}
\end{tikzpicture}
    \caption{Locality constraint for $\MemDetect(C_5)$.
    }
    \label{fig: locality constraints}
\end{figure}

In the subsequent discussion, we define two graph parameters $r_H$ and $r_H'$, and use them to show the locality constraints for $\MemList(H)$ and $\MemDetect(H)$.

% To handle edge insertions/deletions, 
First, we define threshold $r_H$ and show that $r_H$ is a lower bound on the number of rounds needed for $\MemList(H)$ or $\MemDetect(H)$ under edge insertions or under edge deletions.
% as follows. 
% For any edge $e$ in $H$, we write  $H/e$ to denote the graph resulting from contracting $e$ into one node $v_e$. Compute the eccentricity $\ecc_{H/e}(v_e)$ of $v_e$ in graph $H/e$. 
% Let $r_H$ be the largest eccentricity among all edges $e$.
\begin{definition}[Threshold $r_H$]\label{def: r_H}
For any given graph $H$, define $r_H = \nediam(H) - 1$. 
% \[
% r_H = \max_{e \in E(H)} \ecc_{H/e}(v_e),
% \]
% where $H/e$ is the graph resulting from contracting edge $e$ into one node $v_e$.
\end{definition}

% {\color{orange}
% \begin{observation} \label{obs:r_H=nediam(H)-1}
%     $r_H = \nediam(H) - 1$. 
% \end{observation}
% \begin{proof}
%     Given $H$ and $e\in E$, note that $\ecc_{H/e}(v_e) = \max \{\dist_G(u,e): u\in V\}-1$. By definition of $r_H$, we have $r_H = \max_{e \in E} \ecc_{H/e}(v_e) = 
%     \max_{e \in E} \max_{u\in V} \dist_G(u,e) -1 = \nediam(H) - 1$.
% \end{proof} \yijun{maybe discuss the connection to $\nediam(G)$ and add a proof of the observation below and define complete multipartite graphs}\mingyang{done}
% }

% We show that $r_H$ is a lower bound on the number of rounds needed for $\MemList(H)$ or $\MemDetect(H)$ under edge insertions or under edge deletions.

\begin{theorem}[Locality constraints, edge insertions/deletions]\label{lem: locality edge}
    Given any graph $H$, for any $T < r_H$, there exists no $T$-round algorithm for $\MemList(H)$ or $\MemDetect(H)$ under edge insertions or under edge deletions.
\end{theorem}

\begin{proof}
    Suppose $\dist_H(w, e) = \nediam(H)$, for a node $w\in V(H)$ and an edge $e=\{u,u'\}\in E(H)$.
    % Select $e=\{u,u'\}\in E(H)$ such that $r_H = \ecc_{H/e}(v_e)$. There exists a node $w \in V(H)$ such that $r_H = \min \{\dist_H(w,u), \dist_H(w,u')\}$.
    Consider the following dynamic graph $\mathcal{G}$ under edge insertions:
    \begin{enumerate}
        \item
        Initially, $G^0 =H - e$.% is induced by all edges in $H$ except edge $\{u,u'\}$.
        \item At round $1$, consider two cases:
        \begin{enumerate}
            \item The edge $\{u,u'\}$ is inserted.
            Thus $G^1 = H$.
            \item There is no change. Thus $G^1$ does not contain any subgraph isomorphic to $H$.
        \end{enumerate}
    \end{enumerate}

Since $T < r_H = \nediam(H)-1 = \min \{\dist_H(w,u'), \dist_H(w,u)\}$, within $T$ rounds of communication, node $w$ must receive identical messages in both cases and cannot distinguish two cases, so any correct algorithm requires at least $r_H$ rounds. 

A similar proof applies to the case of edge deletion. The only required modification is to start with $G^0 = H$ and then replace the insertion of $e$ with the deletion of $e$.
\end{proof}

Next, we define the threshold $r_H'$ and show that $r_H'$ is a lower bound on the number of rounds needed for $\MemList(H)$ or $\MemDetect(H)$ under node insertions or under node deletions.
% For any node $u$ in $H$, define $N_H^+(u) = N_H(u) \cup \{u\}$ as the union of node $u$ and its neighbors in $H$. Contract $N_H^+(u)$ into one node, which is also called $u$, and form a contracted graph $H/N^+_H(u)$. Compute the eccentricity of $u$ in $H/N^+_H(u)$. Define $r_H'$ to be the largest eccentricity among all choices of $u$.

\begin{definition}[Threshold $r_H'$]\label{def: r_H'}
For any given graph $H$, define 
$r_H'=\diam(H)-1$.

% \[
% r_H' = \max_{u \in V(H)} \ecc_{H/ N^+_H(u)}(u),
% \]
% where $H/N^+_H(u)$ is the graph resulting from contracting $N_H(u) \cup \{u\}$ into one node $u$.
\end{definition}

\begin{theorem}[Locality constraints, node insertions/deletions]\label{lem: locality node deletion}
    Given any graph $H$, for any $T < r_H'$, there exists no $T$-round algorithm for $\MemList(H)$ or $\MemDetect(H)$ under node insertions or under node deletions.
\end{theorem}
\begin{proof}
    Suppose $\dist_H(u,w)=\diam(H)$ for nodes $u, w \in V(H)$.
    % We focus on two nodes $u$ and $w$ in $H$ with the property:
    % \[r_H' = \ecc_{H/N^+_H(u)}(u) = \min_{u'\in N_H(u)} \dist_H(w,u').\]
    For node insertions, we start with the graph $G^0=H-\{u\}$, the subgraph of $H$ induced by $V(H)\setminus \{u\}$. Consider two cases:
    \begin{enumerate}
        \item Insert node $u$, along with all its incident edges in $H$. Thus $G^1=H$.
        \item No topological change.
    \end{enumerate}
    Since $T < r_H' = \diam(H)-1 = \min \{\dist_H(w,u'): u' \in N_H(u)\}$, node $w$ receives the same information within $T$ rounds in both cases. Therefore, for any $T$-round algorithm, $w$ is unable to correctly decide whether $H$ appears. 
    We can use a similar design to prove the case of node deletions: Start with graph $H$, and then at round $1$, either delete node $u$ or do nothing. Since $T < r_H' = \min \{\dist_H(w,u'): u' \in N_H(u)\}$,  the same analysis shows that node $w$ cannot decide whether $H$ disappears.  
\end{proof}

We show a tighter bound in terms of $r_H$ for the case of node insertions. We remark that the same argument does work for node deletion since nodes are assumed to know the entire topology before any deletion and hence may decide based on one-sided information along the shortest path from $u$ to $x$ in the following example.

\begin{theorem}[Locality constraints, node insertions]\label{lem: locality node insertion}
    Given any graph $H$, for any $T < r_H$, there exists no $T$-round algorithm for $\MemList(H)$ or $\MemDetect(H)$ under node insertions.
\end{theorem}
\begin{proof} According to \Cref{obs: diam nediam}, we have either $r_H=\nediam(H)-1 = \diam(H)-1=r'_H$ or $r_H=\nediam(H)-1 = \diam(H)=r_H'+1$. If $r_H=r'_H$, then the proof follows from \Cref{lem: locality node deletion}. 
For the rest of the proof, we focus on the case where $r_H=\nediam(H)-1 = \diam(H)=r_H'+1$. 

% We first find a node $x$ and an edge $e=\{u,v\}$ such that 
% \[
% \dist_H(x,v)= \dist_H(x,u) = r_H = r_H'+1.
% \]
%     Let the initial graph be $G^0 = H$. We consider two cases for the first round: $G^1=H$ or $G^1=H-e$. We claim that node $x$ cannot distinguish the two cases within $T$ rounds of communication if $T< r_H = \dist_H(x,v)= \dist_H(x,u)$.
    Suppose $\dist_H(x, e) = \nediam(H) = \diam(H) + 1$ for node $x\in V(H)$ and edge $e=\{u,v\}\in E(H)$.
    % In the subsequent discussion, we write $w=r'_H$. Let $e=\{u,v\}\in E(H)$ and node $x\in V(H)$ be chosen such that $r_H = \ecc_{H/e}(v_e) = \min \{\dist_H(x,u), \dist_H(x,v)\}=w+1$. 
    We must have $\dist_H(x,u) = \dist_H(x,v) =\diam(H)$. Otherwise, if $\dist_H(x,u)\leq \diam(H) -1$, then $\dist_H(x,e)\leq \diam(H)$, contradicting our assumption.
    % Without loss of generality, we assume that $\dist_H(x,u)=w+1$. We must also have $\dist_H(x,v) = w+1$. 
    % If $\dist_H(x,v) \geq w+2$, then $\diam(H)\geq w+2$, contradicting \Cref{obs: r_H'}. 
    Now, consider the following dynamic graph $\mathcal{G}$.
    \begin{enumerate}
            \item Initially $G^0=H-\{u\}$.% is induced by all nodes in $H$ except $u$.
            \item At round $1$, consider two cases:
        \begin{enumerate}
            \item Node $u$ is inserted with all its incident edges in $H$.
            Thus $G^1=H$.
            \item Node $u$ is inserted  with all its incident edges in $H$ excluding  $\{u,v\}$. Thus $G^1 = H - e$.            
        \end{enumerate}
    \end{enumerate}
    For any $T < r_H = \diam(H)$, node $x$ receives the same information within $T$ rounds in both cases. Therefore, for any $T$-round algorithm, $x$ cannot distinguish the two cases and cannot output correctly.
\end{proof}
%\yanyu{updated, removed contraction}

\paragraph{One-round solvability and complete multipartite graphs} By \Cref{lem: locality edge,lem: locality node insertion} with $T = 1$, we know that $r_H \leq 1$, or equivalently $\nediam(H) \leq 2$, characterizes the class of target subgraphs $H$ that permits one-round $\MemList(H)$ and $\MemDetect(H)$ algorithms for edge insertions, edge deletions, and node insertions.
Observe that $\nediam(H) = 1$ if and only if $H$ is a single edge, so all non-trivial target subgraphs $H$ satisfy $\nediam(H) \geq 2$.
Therefore, \Cref{thm: nediam=2} implies that, excluding trivial target subgraphs, complete multipartite graphs are exactly the class of graphs that permits one-round $\MemList(H)$ and $\MemDetect(H)$ algorithms for edge insertions, edge deletions, and node insertions.

%Applying \Cref{thm: nediam=2}, we know that the complete multipartite graphs are the class of graphs that permits one-round $\MemList(H)$ and $\MemDetect(H)$ algorithms under edge insertions or edge deletions.\yijun{Maybe turn this into a theorem with a proper proof? There are some missing steps, e.g., how to select the parameters to apply the theorem above and what about the case where node-edge diameter is smaller than 2.}
% Hence, we might be interested in the class of graphs which has $r_H = 1$, as this would be the class of graphs for which one-round $\MemList(H)$ and $\MemDetect(H)$ algorithms under edge insertions or edge deletions possibly exist. As shown previously in \Cref{thm: nediam=2}, a connected complete multipartite $H$ with at least three nodes has
% It turns out that this is exactly the class of all complete multipartite graphs. To prove this, we first show some unique properties of complete multipartite graphs.

\section{Technical Overview}\label{sect:overview}

In this section, we present a technical overview of our proofs.
%\yijun{work in progress}

\subsection{Lower Bounds for Finding Cliques}

We start with the lower bounds for finding cliques.

\paragraph{\boldmath $\Omega(\log \log n)$ lower bound} We present the core idea underlying the $\Omega(\log \log n)$ bandwidth lower bound for one-round $\MemDetect(K_s)$ under edge insertions, as shown in \Cref{lem: memdect lb tri}. For simplicity and clarity, we focus on the case of membership-detecting triangles ($K_3$).

We start by describing the hard instance. Consider a set of $t=\log^{0.1} n$ nodes $\{x_1, x_2, \ldots, x_t\}$ and an independent set $I$. By connecting each $x_i$ to two distinct nodes from $I$, we form $t$ disjoint paths of length two, each resembling a triangle missing one edge. The graph is constructed by edge insertions over $2t$ rounds.
 
After constructing these paths, a single edge $e$ is inserted. Two scenarios are possible:
\begin{itemize}
    \item Edge $e$ connects the two endpoints of one of the constructed paths, completing a triangle.
    \item Edge $e$ connects endpoints from two different paths, forming no triangle.
\end{itemize}

Let $(a, x_i, b)$ be one such path, and suppose $a$ is one endpoint of the newly inserted edge $e$. To correctly detect whether a triangle has been formed, node $a$ must determine whether the other endpoint of $e$ is node $b$. We show that if the bandwidth of the algorithm is $B = o(\log \log n)$, then no mechanism exists for $a$ to reliably make this distinction. 

\begin{itemize}
    \item Before the insertion of $e$, node $a$ might try to learn $\ID(b)$ through $x_i$ during the initial $2t$ rounds. However, due to bandwidth limitations, $a$ can receive only $O(Bt) = o(\log n)$ bits, which is far too little to uniquely identify $b$ from the space of possible identifiers $[n]$.
    \item After the insertion of $e$, the two endpoints of $e$ can exchange $B$-bit messages. However, distinguishing whether they are part of the same path among $t$ candidates requires $\Omega(\log t)$ bits, while $B = o(\log \log n) = o(\log t)$ is again insufficient.
\end{itemize}

\paragraph{\boldmath $\Omega(\log \log \log n)$ lower bound} 
We now turn to the proof of \Cref{thm:mixed}, which establishes that the one-round bandwidth complexity of $\Detect(K_s)$ under both node insertions and edge insertions is $\Omega(\log \log \log n)$. For clarity, we again focus on the case of membership-detecting triangles.

 The key reason this lower bound is exponentially smaller than the previous one is that $\Detect(K_s)$ is inherently a much simpler problem than $\MemDetect(K_s)$. In particular, if \emph{only edge insertions are allowed}, then triangle detection is solvable with bandwidth $B = O(1)$: When an edge is inserted, its endpoints can simply broadcast a signal to their neighbors, and any node receiving two such signals can locally infer the existence of a triangle.

However, this strategy breaks down when \emph{node insertions} are also permitted. Again, consider a path of the form $(a, x_i, b)$. Now, it is impossible to distinguish whether an edge $\{a,b\}$ has been inserted---completing a triangle---or whether a new node $c$ has been inserted with incident edges $\{a,c\}$ and $\{b,c\}$, where no triangle is created. The ambiguity arises because both scenarios lead to $x_i$ receiving signals from $a$ and $b$.

 To formalize this, we use a construction similar to the one from the previous lower bound, but with a smaller parameter: $t = \log^{0.1} \log n$. We show that distinguishing the two scenarios requires $\Omega(\log \log \log n)$ bits of bandwidth. 
 
 Here is a brief sketch of the argument. We label each pair $(a, b)$, for $a, b \in I$, according to the messages transmitted across the edges $\{a, x_i\}$ and $\{b, x_i\}$ in the first $2t$ rounds, under the assumption that the path $(a, x_i, b)$ is formed. Since each message consists of $B$ bits and communication lasts $2t$ rounds, the total number of distinct labels is $s = 2^{O(Bt)}$. Since node insertion is allowed, nodes in $I$ without incident edges are considered as not yet inserted.

A Ramsey-type argument then implies the existence of a subset $\{a, b, c\} \subseteq I$ such that all of $(a, b)$, $(a, c)$, $(c, b)$ receive the same label. This means that, from the perspective of node $x_i$, the insertion of an edge $e=\{a,b\}$ versus the insertion of a node $c$ with two incident edges $\{a,c\}$ and $\{b,c\}$ becomes indistinguishable.

For the proof to work, the set $I$ must be sufficiently large relative to the number of distinct labels, which is $s = 2^{O(Bt)}$. This requirement is precisely why we set $t = \log^{0.1} \log n$ instead of the larger value $t = \log^{0.1} n$ used in the previous argument, resulting in a weaker lower bound of $\Omega(\log \log \log n)$. The full proof of \Cref{thm:mixed} is more intricate, as it involves analyzing not only the perspective of each $x_i$, but also the views of the endpoints of the newly inserted edge $e$. For instance, the actual labeling in the proof requires quantifying over all $i \in [t]$, which increases the number of distinct labels to $2^{O(Bt^2)}$.

\subsection{Upper Bounds for Finding Triangles}

We now discuss our two triangle-finding algorithms, both of which achieve optimal bandwidth complexity in bounded-degree networks, matching our lower bounds.

Our algorithms build upon the one-round algorithm for $\MemDetect(K_3)$ under edge insertions by Bonne and Censor-Hillel~\cite{bonne2019distributed}, which operates with bandwidth $B = O(\sqrt{\Delta \log n})$. Their approach uses two types of messages: one with $d$ bits, and another with $O\left(\frac{\Delta \log n}{d}\right)$ bits. The overall bandwidth complexity is optimized by choosing $d = \sqrt{\Delta \log n}$.

\paragraph{\boldmath $O( \log \log n)$ upper bound} To prove \Cref{thm: Delta memlist ub loglogn}, we extend the algorithm of Bonne and Censor-Hillel to design a one-round algorithm for $\MemList(K_3)$ under edge insertions with a bandwidth complexity of $O(\Delta^2 \log \log n)$. In their original algorithm for $\MemDetect(K_3)$, each $d$-bit message is a binary string where the $i$th bit indicates whether an incident edge was inserted $i$ rounds ago. We observe that this binary string can be compactly represented using $O(\Delta \log d)$ bits: For each recently inserted incident edge, use a number in $[d]$ to indicate its insertion time. Setting $d = \log n$ gives a bandwidth complexity of $O(\Delta \log \log n)$ for $\MemDetect(K_3)$.

To handle the more demanding $\MemList(K_3)$ problem, we introduce several modifications to the algorithm. Most notably, we must account not only for the edges incident to a node but also for those incident to its neighbors. This additional layer of information increases the size of the message by a factor of $\Delta$, resulting in a total bandwidth complexity of $O(\Delta^2 \log \log n)$.

\paragraph{\boldmath $O( \log \log \log n)$ upper bound} We now turn to the proof of \Cref{thm: Delta list lb logloglogn}, where we design a one-round algorithm for $\List(K_3)$ that handles both edge and node insertions with bandwidth $O(\Delta \log \log \log n)$. The exponential improvement stems from improving the size of the $O\left(\frac{\Delta\log n}{d}\right)$-bit message in the algorithm of Bonne and Censor-Hillel~\cite{bonne2019distributed} to $O\left(\frac{\Delta\log \log n}{d}\right)$ bits. The purpose of this message is to transmit the list of neighborhood $\ID$s, which contains $O(\Delta\log n)$ bits of information. As the transmission is done in $d$ rounds, the required message size is $O\left(\frac{\Delta\log n}{d}\right)$.

Our key idea lies in a new method for identifying triangles. Recall that a core challenge in our $\Omega(\log \log \log n)$ lower bound proof is to understand the inherent difficulty for a node $x$, with two non-adjacent neighbors $a$ and $b$, to distinguish between the insertion of an edge $\{a,b\}$ and the insertion of a node $c$ with two incident edges $\{a,c\}$ and $\{b,c\}$. 

We develop a new algorithm to handle this instance. We let $x$ select an index $i$ such that the $i$th bits of $\ID(a)$ and $\ID(b)$ differ. Nodes $a$ and $b$ then report the $i$th bit of the identifier of their new neighbor. If an edge $\{a,b\}$ is inserted, then the reported bits will differ. If a node $c$, along with two incident edges $\{a,c\}$ and $\{b,c\}$, is inserted, then the bits will match. This comparison allows $x$ to distinguish between the two cases. 

Sending an index requires only $O(\log \log n)$ bits, which is exponentially more efficient than sending the full identifier, which requires $O(\log n)$ bits. Consequently, the size of the $O\left(\frac{\Delta \log n}{d}\right)$-bit message in the algorithm of Bonne and Censor-Hillel~\cite{bonne2019distributed} is reduced to $O\left(\frac{\Delta \log \log n}{d}\right)$ bits. Setting $d = \log \log n$ then yields the improved bandwidth complexity $O(\Delta \log \log \log n)$.

\subsection{Membership-Listing}

We obtain a \emph{complete characterization} of the bandwidth complexity of $\MemList(H)$ for {every} target subgraph $H$, every number of rounds $r$, and every type of topological change: node insertions, node deletions, edge insertions, and edge deletions. The full characterization is summarized in \Cref{table:memList1}. In this overview, we omit node insertions and deletions, as they closely mirror the respective cases of edge insertions and deletions.

We begin with the case of \emph{edge insertions}. For $r < r_H$, we have an impossibility result from \Cref{lem: locality edge}. For $r \geq r_H$, the $r$-round bandwidth complexity depends on the structure of $H$:
\begin{itemize}
  \item If $H$ is a {complete multipartite graph that is not a clique}, the complexity is $\Theta(n / r)$.
  \item If $H$ is {not} a complete multipartite graph, the complexity increases to $\Theta(n^2 / r)$.
\end{itemize}

\paragraph{Upper bounds} 
The upper bounds follow from the observation that a node $v$ can list all subgraphs $H$ that contain it once it has an accurate view of its $r_H$-radius neighborhood. A brute-force approach would be to flood the entire graph topology using messages of size $O(n^2)$ after each topological change. Within $r_H$ rounds, this ensures all nodes acquire the required local view. If $r \geq r_H$, this communication can be spread over multiple rounds, reducing the bandwidth requirement to $O(n^2 / r)$.

For the special case where $H$ is a complete multipartite graph, we have $r_H = 1$. This allows a more efficient approach: Each node simply broadcasts its list of neighbors as a binary string of length $n$, leading to a reduced bandwidth complexity of $O(n / r)$.

\paragraph{Lower bounds} We first discuss a general $\Omega(n / r)$ lower bound for any non-clique subgraph $H$. Let $\{u,v\}$ be a non-edge in $H$.  Consider the graph resulting from replacing $u$ with an independent set $U'$, so each member of $U'$ corresponds to a copy of $H$. We then construct the graph via edge insertions, ensuring that the edges incident to $v$ are added last, leaving $v$ only $O(r)$ rounds to gather information about the graph. For $v$ to list all copies of $H$, it must learn the set $U'$, which requires $\Omega(n)$ bits of information, yielding a lower bound of $\Omega(n / r)$.

To strengthen the bound to $\Omega(n^2 / r)$ when $H$ is not a complete multipartite graph, we use the fact that such a graph $H$ must contain an edge $e$ and a node $v$ such that neither endpoint of $e$ is adjacent to $v$. Now we apply a similar construction where $e$ is replaced with a bipartite graph, forcing $v$ to learn the bipartite graph in order to list all copies of $H$, which requires $\Omega(n^2)$ bits of information.

\paragraph{Edge deletions} We now consider the case of \emph{edge deletions}. As with insertions, we have an impossibility result for $r < r_H$ from \Cref{lem: locality edge}. For $r \geq r_H$, the $r$-round bandwidth complexity is $\Theta(n / r)$ whenever $H$ is not a clique.

A key difference between insertions and deletions is that edge insertions can merge disjoint components, potentially introducing many new subgraphs and requiring extensive information dissemination. In contrast, to handle deletions, it suffices to propagate the identifiers of the two endpoints of the deleted edge to a radius of $r_H$, which requires only $O(\log n)$ bits of information. This yields an upper bound of $O((\log n) / r)$.

To prove a matching lower bound, we reuse the construction from the $\Omega(n / r)$ lower bound. Suppose an edge deletion causes one of the $|U'|$ copies of $H$ to disappear. For a node $v$ to correctly identify \emph{which} copy was affected, it must learn $\Omega(\log |U'|)$ bits. By letting $U'$ contain a polynomial number of nodes, this yields the desired $\Omega((\log n) / r)$ lower bound.

% Next, we consider the case of edge deletions. In this case, When $r < r_H$, we also already have an impossibility result from \Cref{lem: locality edge}. When $r \geq r_H$, the $r$-round bandwidth complexity is $\Theta(n / r)$ whenever $H$ is not a clique. The intuition for why edge insertions require much higher bandwidth is that edge insertions can cause two disjoint graphs to merge, leading to a lot of new information to learn. When an edge deletion occurs, essentially we just need to propagate the identifier of the two endpoints of the deleted edge up to radius $r_H$, so only $O(\log n)$ bits of information. This explains the $O((\log n) / r)$ upper bound. To prove a matching lower bound, we again use the lower bound graph construction behind the $\Omega(n/r)$ lower bound. An edge deletion can cause one of the $|U'|$ copies of $H$ to disappear. For $v$ to tell specifically which one disappears, $v$ needs to learn $\Omega(\log |U'|)$ bits of information, which gives the desired $\Omega(\log n)$ lower bound by setting $|U'| = n^{\Theta(1)}$.

\subsection{One-Round Membership-Detection}
We now turn to the one-round bandwidth complexity of the $\MemDetect(H)$ problem, for which we provide a partial characterization, see \Cref{tab:memDect}. Compared to $\MemList(H)$, establishing lower bounds for $\MemDetect(H)$ is more challenging, as it is harder to quantify the minimum information a node must obtain to detect the presence of a subgraph. On the algorithmic side, obtaining optimal upper bounds is also trickier: Since detection does not require listing the subgraph, there is greater flexibility in how the subgraph can be found. This flexibility enables a wider range of algorithmic techniques. In this overview, we focus on two representative results. %: one upper bound and one lower bound.

\paragraph{Lower bound.} We show that the one-round bandwidth complexity of $\MemDetect(H)$ under edge insertions is $\Omega(n)$ for any complete multipartite graph $H$ that is neither a star nor a clique. While the corresponding $\Omega(n / r)$ lower bound for $\MemList(H)$ appears inherently tied to the listing requirement, we demonstrate that, with suitable modifications, the core idea can be adapted to the detection setting. 

Since we cannot require a node $v$ to list all copies of $H$ in the constructed graph, we instead frame the argument in a preprocessing-plus-query model. We first build a graph that initially contains no copy of $H$, but is structured so that a copy can be formed in many different ways. The goal is to ensure that, in order for $v$ to detect the presence of $H$ following an edge insertion, it must have already learned a significant amount of information about the initial graph during the construction phase.

Specifically, we identify two non-adjacent nodes $u$ and $v$ in $H$ that share a common neighbor $w$. We construct a graph by removing the edge $\{u, w\}$ from $H$ and replacing node $u$ with an independent set $U'$. The graph is built via edge insertions, with the edges incident to $v$ and $w$ added at the end, leaving them only $O(1)$ rounds to learn about $U'$. We then insert a new edge $e$ incident to $w$, creating two possible scenarios: If the other endpoint of $e$ lies in $U'$, a copy of $H$ is formed; otherwise, it is not. For $v$ to decide correctly, $v$ and $w$ must learn the set $U'$ before the insertion of $e$, which requires $\Omega(n)$ bits of information.

%%%%%%%%%%%%%

% \paragraph{Lower bound} We show that the one-round bandwidth complexity of $\MemDetect(H)$ is $\Omega(n)$ under edge insertions for any complete multipartite graph $H$ that is neither a star nor a clique. The idea is that, while the corresponding $\Omega(n / r)$ lower bound proof for $\MemList(H)$ seems to inherently only work for membership-listing, with some modifications, it can be adapted to work for membership-detection.

% First of all, we can find two non-adjacent nodes $u$ and $v$ with a common neighbor $w$ in $H$. Consider the graph resulting from removing the edge $\{u,w\}$ in $H$ and replacing the node $u$ with an independent set $U'$ of size $\Theta(n)$.
% Again, the graph is constructed by a series of edge insertions, where the edges incident to $v$ and $w$ are added in the end, so $v$ and $w$ have only $O(1)$ rounds to learn the set $U'$. Next, we insert an edge $e$ incident to $w$. There are two cases. If the other endpoint of $e$ is in $U'$, then a copy of $H$ is formed. Otherwise, no copy of $H$ is formed. For $v$ to output correctly, $v$ and $w$ must learn the set $U'$, which requires $\Omega(n)$ bits of information. 

\paragraph{Upper bound} As in the case of $\MemList(H)$, membership-detection becomes significantly easier when the allowed topological change is a deletion rather than an insertion. In particular, we present a one-round $O(1)$-bandwidth algorithm for $\MemDetect(H)$ that applies to any complete multipartite graph $H$, improving upon the $O(\log n)$ bound required for $\MemList(H)$ under the same conditions.

For clarity, we describe our algorithm for the special case $H = C_4$, which captures the key ideas behind the more general algorithm that works for an arbitrary complete multipartite graph $H$. The core observation is that to detect a $C_4$, it suffices for each node $w$ to maintain, for every pair of its neighbors $u$ and $v$, the number of common neighbors they share, excluding $w$. Each such shared neighbor corresponds to a distinct copy of $C_4$ containing $w$, $u$, and $v$.

Maintaining these counters requires only one-bit messages. When a node detects that one of its neighbors has been deleted, it sends a one-bit signal to all its remaining neighbors. If a node $w$ receives such a signal from two neighbors $u$ and $v$ in the same round, it decrements the counter for the pair $\{u, v\}$ by one.

\subsection{One-Round Listing}

We study the one-round bandwidth complexity of the $\List(H)$ problem for edge deletions and node deletions. See \Cref{tab:listing} for a summary of our results. In this overview, we focus on the case of edge deletions, as the case of node deletions is analogous.

We begin with the simpler cases. When $\nerad(H) = 1$, $H$ is a star, so $\List(H)$ is trivially solvable with zero bandwidth by allowing the star center to handle the listing. When $\nerad(H) = 3$, $\List(H)$ cannot be solved in one round, since there exists at least one edge $e$ in $H$ whose deletion cannot be communicated to all nodes within one round.

When $\rad(H) = 1$, the graph $H$ has a center node $v$ adjacent to all other nodes, enabling a simple one-round one-bit algorithm: Whenever an edge is deleted, its endpoints send a signal to all their neighbors, allowing the center of $H$ to determine whether $H$ still exists.

We now turn to the remaining nontrivial case, where $\nerad(H) = 2$ and $\rad(H) = 2$. In this setting, we prove a tight $\Theta(\log n)$ bound. The upper bound is achieved by a simple protocol: When an edge $e = \{u, v\}$ is deleted, both $u$ and $v$ broadcast $\ID(u)$ and $\ID(v)$ to all their neighbors. This guarantees that if the deletion eliminates a copy of $H$, its center, who is responsible for listing the subgraph, can detect the change. 

The lower bound is more involved and requires a novel construction. Let $V(H)=\{u_1$, $u_2$, $\ldots$, $u_m\}$. We replace each node $u_i$ in $H$ with an independent set of $n$ nodes $S_i = \{v_{i,1}, \ldots, v_{i,n}\}$. Moreover, we assume, based on the structure of $H$, that there exists a path of length two from $u_1$ to $u_m$ via $u_{m-1}$, but no direct edge between $u_1$ and $u_m$.

Using symmetry and the pigeonhole principle, we can assume that node $v_{1,1}$ must list at least $\Omega(n)$ distinct copies of $H$ of the form $\{v_{1,i}, v_{2,1}, v_{3,1}, \ldots, v_{m-1,1}, v_{m,j}\}$ for $i, j \in [n]$. Among these copies of $H$, there must be at least $\Omega(\sqrt{n})$ copies with distinct $i$-values, or at least $\Omega(\sqrt{n})$ with distinct $j$-values.

Assume the former holds, and let $I$ be the set of distinct $i$ values. Now consider deleting the edge $\{v_{1,i^*}, v_{m-1,1}\}$ for some $i^* \in I \setminus \{1\}$. Node $v_{1,1}$ must then stop listing all copies of $H$ that include this edge. Since there is no direct connection between $v_{1,1}$ and $v_{1,i^*}$, it must receive a message from $v_{m-1,1}$ that uniquely identifies $i^*$ among $\Omega(\sqrt{n})$ candidates, which requires $\Omega(\log n)$ bits. The argument for the latter case (distinct $j$-values) is similar.

\section{Lower Bounds for Finding Cliques}\label{sect:cliqueLB}

In this section, we prove our two lower bounds, \Cref{lem: memdect lb ks} and \Cref{thm:mixed}, for finding cliques in dynamic networks. We emphasize that both lower bounds hold even for bounded-degree networks. Both lower bounds are proved using the same framework, which is described in \Cref{subsect:setup}. \Cref{lem: memdect lb ks} is proved in \Cref{subsect:edgeIns}. \Cref{thm:mixed} is proved in \Cref{subsect:mixed}.

We begin by discussing the intuition behind \Cref{thm: Delta list lb logloglogn}.
Recall that the problem $\List(K_3)$ has a bandwidth complexity of $B = O(1)$ under any single type of topological change~\cite{bonne2019distributed}. However, when both edge insertions and node insertions are allowed, \Cref{thm:mixed} establishes a bandwidth lower bound of $B = \Omega(\log\log\log n)$. The proof of this lower bound relies on the inherent difficulty for a node $w$, with two non-adjacent neighbors $x$ and $y$, to distinguish between the insertion of an edge $\{x,y\}$ and the insertion of a node  $w'$ with two incident edges $\{x,w'\}$ and $\{y,w'\}$.
%given $B=o(\log\log\log n)$.

We develop a new technique to handle the hard instance with $B = O(\Delta \log \log \log n)$. For $w$ to distinguish between the above two cases, we fix an index $i$ such that the $i$th bit of $\ID(x)$ does not equal the $i$th bit of $\ID(y)$, and then we ask $x$ and $y$ to report to $w$ the $i$th bit of their new neighbor $\ID$. The two bits differ for the case of insertion of $\{x,y\}$ by our choice of $i$. The two bits are identical for the case of node insertion of $w'$ with two incident edges $\{x,w'\}$ and $\{y,w'\}$, as they are both the $i$th bit of the same string $\ID(w')$. 

We implement this approach within the framework of the proof for \Cref{thm: Delta memlist ub loglogn} by modifying the algorithm for \ref{technique: pu}. The improvement from $O(\log \log n)$ to $O(\log \log \log n)$ arises from the observation that transmitting such an index $i$ requires only $O(\log \log n)$ bits, which is exponentially more efficient than sending the entire $\ID$, which requires $O(\log n)$ bits.

\subsection{Hard Instances}\label{subsect:setup}

The construction of the hard instances underlying \Cref{lem: memdect lb ks} and \Cref{thm:mixed} is parameterized by a parameter $t$. We set $t = \log^{0.1} n$ in the proof of \Cref{lem: memdect lb ks} and $t = \log^{0.1} \log n$ in the proof of \Cref{thm:mixed}.

\paragraph{Initialization} At the beginning, the initial graph $G^0$ at round zero consists of $t$ $(s-2)$-cliques $K^1, K^2, \ldots, K^t$ and an independent set $I$ of size $n - (s-2)t$, so the graph contains exactly $n$ nodes. 

\paragraph{Almost-clique creations} During the first $2(s-2)\cdot t$ rounds, we may perform edge insertions to connect some nodes in $I$ to some cliques. For any $i \in [t]$ and two distinct nodes $u \in I$ and $v \in I$, we write $\create(u,K^i,v)$ to denote the operation of adding an edge between each $x \in \{u,v\}$ and each $y \in K^i$ in the $2(s-2)$-round interval $[2(s-2)\cdot (i-1) + 1, \ldots 2(s-2)\cdot i]$. This operation makes $K^i \cup \{u,v\}$ an $s$-clique minus an edge $\{u,v\}$. Observe that the intervals $[2(s-2)\cdot (i-1) + 1, \ldots 2(s-2)\cdot i]$ for all $i \in [t]$ are disjoint. We write $\tlast = 2(s-2)\cdot t + 1$ to denote the first round number after the last time interval.

\paragraph{Hard instances} The construction of our hard instances is based on the description of a given distributed algorithm $\mathcal{A}$ whose existence we want to disprove. In the proof of \Cref{lem: memdect lb ks}, we assume that $\mathcal{A}$ is an $o(\log \log n)$-bandwidth one-round algorithm for $s$-clique membership-detection under edge insertions. Similarly, In the proof of \Cref{thm:mixed}, we assume that $\mathcal{A}$ is an $o(\log \log \log n)$-bandwidth one-round algorithm for $s$-clique detection under both edge insertions and node insertions. In the subsequent discussion, we let $B$ be the bandwidth of $\mathcal{A}$, i.e., $B=o(\log \log n)$ in the proof of \Cref{lem: memdect lb ks} and $B=o(\log \log \log n)$ in the proof of \Cref{thm:mixed}.

Given the description of $\mathcal{A}$, we select four distinct nodes $a$, $b$, $c$, and $d$ from the independent set $I$ and two distinct indices $i \in [t]$ and $i^\ast \in [t]$. We consider the following four scenarios.

\begin{description}
    \item[(S1)\label{item:SS1}] $\create(a,K^i,b)$ and insert the edge $\{a,b\}$ at round $\tlast$.
    \item[(S2)\label{item:SS2}] $\create(a,K^i,c)$, $\create(d,K^{i^\ast},b)$, and insert the edge $\{a,b\}$ at round $\tlast$.
    \item[(S3)\label{item:SS3}] $\create(d,K^i,b)$, $\create(a,K^{i^\ast},c)$, and insert the edge $\{a,b\}$ at round $\tlast$.
    \item[(S4)\label{item:SS4}] $\create(a,K^i,b)$ and insert the node $c$, together with two incident edges $\{a,c\}$ and $\{c,b\}$, at round $\tlast$. 
\end{description} 

For \ref{item:SS4} to make sense, in the setting where node insertion is allowed, nodes in $I$ without incident edges are considered as not yet inserted.

We emphasize that \ref{item:SS4} is the only scenario that involves a node insertion, so \ref{item:SS4} is considered in the proof of  \Cref{thm:mixed} only. 
Observe that \ref{item:SS1} is the only scenario where an $s$-clique is formed. The lower bounds are based on \emph{indistinguishability} arguments between \ref{item:SS1} and other scenarios, which we briefly explain as follows.
\begin{itemize}
    \item In the proof of \Cref{lem: memdect lb ks}, we only use \ref{item:SS1} and \ref{item:SS2}. Since the underlying problem is membership-detection, all members of the $s$-clique $K^i \cup \{a,b\}$ in \ref{item:SS1} need to detect the presence of the clique. We aim to show that it is possible to select the nodes in such a way that $a$ cannot distinguish between \ref{item:SS1} and \ref{item:SS2}, so the considered algorithm $\mathcal{A}$ must fail in at least one of the scenarios.
    \item In the proof of \Cref{thm:mixed}, we use all of \ref{item:SS1}, \ref{item:SS2}, \ref{item:SS3}, and \ref{item:SS4}.  Since the underlying problem is detection, we just need one member of the $s$-clique $K^i \cup \{a,b\}$ in \ref{item:SS1} to detect the presence of the clique. We aim to show that it is possible to select the nodes to satisfy the following indistinguishability requirements: ($a$ cannot distinguish between \ref{item:SS1} and \ref{item:SS2}), ($b$ cannot distinguish between \ref{item:SS1} and \ref{item:SS3}), and ($K^i$ cannot distinguish between \ref{item:SS1} and \ref{item:SS4}).
    % \begin{itemize}
    %     \item $a$ cannot distinguish between \ref{item:SS1} and \ref{item:SS2}.
    %     \item $b$ cannot distinguish between \ref{item:SS1} and \ref{item:SS3}.
    %     \item $K^i$ cannot distinguish between \ref{item:SS1} and \ref{item:SS4}.
    % \end{itemize}
    Therefore, regardless of the choice of the nodes responsible for detecting the $s$-clique $K^i \cup \{a,b\}$ in \ref{item:SS1}, $\mathcal{A}$ must fail in at least one of the scenarios.
\end{itemize}
 
%\subsection{An \texorpdfstring{$\Omega(\log \log n)$}{Omega(log log n)} Lower Bound Under Edge Insertions}
\subsection{Proof of \texorpdfstring{\Cref{lem: memdect lb ks}}{Lemma \ref{lem: memdect lb ks}}}
\label{subsect:edgeIns}

% \begin{theorem}\label{lem: memdect lb tri}
%     The deterministic one-round bandwidth complexity of $MemDetect(K_3)$ under edge insertions is $\Omega(\log \log n)$.
% \end{theorem}

%an $\Omega(\log \log n)$ bandwidth complexity lower bound for one-round algorithms for $s$-clique membership-detection under edge insertions, for any constant $s \geq 3$. The lower bound holds even for bounded-degree graphs. This result 

% \begin{description}
%     \item[Scenario 1\label{item:SS1}] $\create(a,K^i,b)$ and insert the edge $\{a,b\}$ at round $\tlast$.
%     \item[Scenario 2\label{item:SS2}] $\create(a,K^i,b')$, $\create(c,K^{i'},b)$, and insert the edge $\{a,b\}$ at round $\tlast$.
% \end{description} 

We prove \Cref{lem: memdect lb ks} with the parameter choices $B=o(\log \log n)$ and $t = \log^{0.1} n$.

\loglogLB*

As discussed earlier, it suffices to show that, by the end of round $\tlast$, node $a$ cannot distinguish between \ref{item:SS1}, where an $s$-clique $K^i \cup \{a,b\}$ involving $a$ is formed, and \ref{item:SS2}, where $a$ is not contained in an $s$-clique. We briefly explain the intuition behind the indistinguishability proof as follows. Informally, there are two ways that $a$ can distinguish between the two scenarios.
\begin{enumerate}
    \item Node $a$ can try to tell the difference between $\ID(b)$ and $\ID(d)$ through the communication with $K^i$ in the first $\tlast$ rounds. However, the total amount of bits $a$ can receive from $K^i$ is $O(Bt) = o(\log n)$, which is not enough to even learn one identifier.\label{item-1}
    \item Node $a$ can try to tell the difference between $i$ and $i^\ast$ through the $B$-bit message from $b$ immediately after the edge $\{a,b\}$ is inserted. However, this is insufficient to learn one index in $[t]$, as $B = o(\log \log n)=o(\log t)$. \label{item-2}
\end{enumerate}

\paragraph{Remark} The choice of the exponent $0.1$ in $t = \log^{0.1} n$ is arbitrary and can be replaced with any constant in the interval $(0,1)$. This approach inherently leads to a bandwidth lower bound of $\Omega(\log \log n)$. While a substantially larger value of $B$ could still satisfy the constraint $O(Bt) = o(\log n)$ in Item~\ref{item-1}, the condition $B = o(\log t)$ in Item~\ref{item-2} imposes a stricter limitation. Since $t = o(\log n)$ is required by the constraint in Item~\ref{item-1}, it follows that $B = o(\log \log n)$.

\paragraph{Parameter selection} In the subsequent discussion, we show how to select $\{a,b,c,d,i,i^\ast\}$ to ensure that, in the first $\tlast$ rounds, $a$ receives identical messages in both \ref{item:SS1} and \ref{item:SS2}, so algorithm $\mathcal{A}$ fails in at least one of the two scenarios, implying \Cref{lem: memdect lb tri}.

%\paragraph{Friends} 
We first focus on controlling the messages sent from $K^i$ to $a$. For any $w \in I$ and $j \in [t]$, we say that two distinct nodes $u \in I\setminus\{w\}$ and $v \in I\setminus\{w\}$ are \emph{friends} w.r.t.~$(w,j)$ if $w$ receives identical messages from $K^j$ in the first $\tlast$ rounds in both %\ref{item:SS1} and \ref{item:SS2} 
$\create(w,K^j,u)$ and $\create(w,K^j,v)$. If $u \in I$ does not have any friends  w.r.t.~$(w,j)$, then we say that $u$ is \emph{lonely} w.r.t.~$(w,j)$.

\begin{lemma}\label{lem:friends}
There exists a node $b \in I$ and a subset $I^\star \subseteq I$ with $|I^\star| \geq n^{0.99}$ such that $b$ is not lonely w.r.t.~$(w,j)$ for all $w \in I^\star$ and $j \in [t]$.
\end{lemma}
\begin{proof}
Select a subset $I^\star \subseteq I$ of $\lceil n^{0.99} \rceil$ nodes arbitrarily. For any $w \in I^\star$ and $j \in [t]$, the number of lonely nodes w.r.t.~$(w,j)$ is upper bounded by $2^{O(Bt)} = n^{o(1)}$, since $O(B \tlast) = O(Bt)$ is the total number of bits a node can receive within $\tlast=O(t)$ rounds in a bounded-degree network with bandwidth $B$. 

By a summation over all $w \in I^\star$ and $j \in [t]$, the total number of nodes that are lonely w.r.t.~$(w,j)$ for some  $w \in I^\star$ and $j \in [t]$ is upper bounded by $|I^\star|\cdot t \cdot 2^{O(Bt)}$. Observe that \[ |I^\star|\cdot t \cdot 2^{O(Bt)} < |I|,\] because $|I| = \Theta(n)$, $|I^\star| = \lceil n^{0.99} \rceil$, $B=o(\log \log n)$, and $t = \log^{0.1} n$. Therefore, there exists a node $b \in I$ that is not lonely w.r.t.~$(w,j)$ for all $w \in I^\star$ and $j \in [t]$.
\end{proof}

For the rest of the proof, we fix $b \in I$ and $I^\star \subseteq I$  according to \Cref{lem:friends}, and then we select $a$ and $d$ from $I^\star$ and $i$ and $i^\ast$ from $[t]$ in such a way that the message sent from $b$ to $a$ after the insertion of $\{a,b\}$ in both \ref{item:SS1} and \ref{item:SS2} are identical. After fixing all of $\{a,b,d,i,i^\ast\}$, we select $c \in I$ as any friend of $b$ w.r.t.~$(a,i)$, whose existence is guaranteed by \Cref{lem:friends}.

Next, we focus on controlling the messages sent from $K^j$ to $b$, where $j = i$ in \ref{item:SS1} and $j = i^\ast$ in \ref{item:SS2}, which influence the message sent from $b$ to $a$.

\begin{lemma}\label{lem:select}
There exist two distinct nodes $a \in I^\star$ and $d \in I^\star$ such that, for every $j \in [t]$, the messages sent from $K^j$ to $b$ in the first $\tlast$ rounds are identical in both %\ref{item:SS1} and \ref{item:SS2} 
 $\create(a,K^j,b)$ and $\create(d,K^j,b)$.
\end{lemma}
\begin{proof}
We classify each node $v \in I^\star$ based on the sequence of messages sent from $K^j$ to $b$ during the first $\tlast$ rounds in $\create(v,K^j,b)$, for all $j \in [t]$. To prove the lemma, it suffices to show that the number of such classes is smaller than $|I^\star|$. This ensures that at least one class contains two or more nodes, allowing us to select any two nodes from that class as $a$ and $d$.

The total number of classes is at most $t \cdot 2^{O(Bt)}$, since there are $t$ possible choices for $j \in [t]$, and the number of distinct message patterns from $K^j$ to $b$ in the first $\tlast$ rounds is bounded by $2^{O(Bt)}$. Indeed, we have
\[t \cdot 2^{O(Bt)} = n^{o(1)} < |I^\star|,\]
as $B=o(\log \log n)$ and $t = \log^{0.1} n$.
\end{proof}

We fix $a \in I^\star$ and $d \in I^\star$ according to \Cref{lem:select}. 

\begin{lemma}\label{lem:select2}
There exist two distinct indices $i \in [t]$ and $i^\ast \in [t]$ such that the message sent from $b$ to $a$ after the insertion of $\{a,b\}$ is the same in both \ref{item:SS1} and \ref{item:SS2}.
\end{lemma}
\begin{proof}
By \Cref{lem:select}, our choice of nodes $a$ and $d$ ensures that the messages sent from $K^j$ to $b$---where $j = i$ in \ref{item:SS1} and $j = i^\ast$ in \ref{item:SS2}---depend only on the index $j$. In particular, swapping the roles of $a$ and $d$ has no effect on the messages sent from $K^j$ to $b$.

%are the same in both \ref{item:SS1} and \ref{item:SS2}.\yijun{wrong - should be is a function of $j$, and also it has nothing to do with $c$, should be $d$} 

Given that $a$, $b$, and $d$ have been fixed, the $B$-bit message sent from $b$ to $a$ after the insertion of $\{a,b\}$ is a function of the index $j \in \{i, i^\ast\}$ of the $(s-2)$-clique $K^j$ only. By the pigeonhole principle, since the number $t=\log^{0.1} n$ of choices of $j$ is larger than the number $2^B = \log^{o(1)} n$ of possible $B$-bit messages, there exist two distinct indices $i \in [t]$ and $i^\ast \in [t]$ resulting in the same message.
\end{proof}

We fix $i \in [t]$ and $i^\ast \in [t]$ according to \Cref{lem:select2}. Now all of $\{a,b,d,i,i^\ast\}$ have been fixed. As discussed earlier, we select $c \in I$ as any friend of $b$ w.r.t.~$(a,i)$, whose existence is guaranteed by \Cref{lem:friends}. We are now ready to finish the proof of \Cref{lem: memdect lb tri}.

\begin{proof}[Proof of \Cref{lem: memdect lb tri}]
We just need to show that, by the end of round $\tlast$, node $a$ cannot distinguish between \ref{item:SS1} and \ref{item:SS2}. By \Cref{lem:select2}, $a$ receives the same message from $b$ after the insertion of $\{a,b\}$ in both \ref{item:SS1} and \ref{item:SS2}. Since $b$ and $c$ are friends w.r.t.~$(a,i)$, $a$ receives the same message from $K^i$ in the first $\tlast$ rounds in both \ref{item:SS1} and \ref{item:SS2}. Therefore, $a$ cannot distinguish between \ref{item:SS1} and \ref{item:SS2}, so the algorithm $\mathcal{A}$ produces an incorrect output in at least one of the two scenarios. Hence there is no one-round algorithm solving $\MemDetect(K_s)$ with $o(\log \log n)$ bandwidth.
\end{proof}

\subsection{Proof of \texorpdfstring{\Cref{thm:mixed}}{Theorem \ref{thm:mixed}}}
%\subsection{Clique Detection Under Multiple Types of Changes}
\label{subsect:mixed}

We prove \Cref{thm:mixed}  with the parameter choices $B=o(\log \log \log n)$ and $t = \log^{0.1} \log n$.

\logloglogLB*

As discussed earlier, it suffices to select the parameters $\{a,b,c,d,i,i^\ast\}$ in such a way that, by the end of round $\tlast$, the following indistinguishability requirements are met: ($a$ cannot distinguish between \ref{item:SS1} and \ref{item:SS2}), ($b$ cannot distinguish between \ref{item:SS1} and \ref{item:SS3}), and ($K^i$ cannot distinguish between \ref{item:SS1} and \ref{item:SS4}). The parameter selection is done in two steps. In the first step, we select $\{a,b,c,d\}$ to ensure indistinguishability for the messages communicated between the $(s-2)$-cliques and other nodes. In the second step, we select $\{i, i^\ast\}$ to ensure indistinguishability for the messages sent across the edge $\{a,b\}$ inserted in round $\tlast$. 

\paragraph{The first step} Consider the complete directed graph $G^\ast$ over the node set $I$, where any two distinct nodes $x$ and $y$ in $I$ are connected by two directed edges $(x,y)$ and $(y,x)$. We color each edge $(x,y)$ in $G^\ast$ by a vector of \[t \cdot \left(\tlast \cdot 4\cdot (s-2) + 2\right) = O(t^2)\] messages of $B$ bits that correspond to the messages communicated in the first $\tlast$ rounds between  $x$, $y$, and $K^i$ in \ref{item:SS1} with $a=x$ and $b=y$ \underline{for all choices of $i \in [t]$}. Here, we do not consider the messages sent within the edges in $K^i$. The term $+2$ in the formula reflects the two messages sent across the edge $\{a,b\}$ at round $\tlast$ immediately after inserting the edge $\{a,b\}$.

\begin{lemma}\label{lem:first_step}
    There exist $\{a,b,c,d\} \subseteq I$ such that the edges $(a,b)$, $(a,c)$, $(a,d)$, $(c,b)$, and $(d,b)$ are colored the same in $G^\ast$.
\end{lemma}
\begin{proof}
Let $x = 2^{O(Bt^2)}=2^{O(\log^{0.2} \log n)\cdot o(\log \log \log n)} = 2^{O(\log^{0.3} \log n)}$ be the total number of colors.
We select the nodes $v_1, v_2, \ldots. v_{2x+1}$ in $I$ sequentially as follows. The first node $v_1$ is selected arbitrarily, and we select the color $c_1$ that appears the highest number of times among the edges emanating from $v_1$. Now we restrict to the subgraph $G_1$ of $G^\ast$ induced by the neighbors $u$ of $v_1$ such that $(v_1, u)$ is colored by $c_1$. The selection of the remaining nodes is done recursively in the subgraph. The colors $c_2, c_3, \ldots$ and the subgraphs $G_2, G_3, \ldots$ are defined similarly. By the pigeonhole principle, among $\{c_1, c_2, \ldots, c_{2x+1}\}$, there exist three indices $i_1 < i_2 < i_3$ such that $c_{i_1} = c_{i_2} = c_{i_3}$. We set $a = v_{i_1}$, $c = v_{i_2}$, and $d = v_{i_3}$, and select $b$ to be any node in $G_{2x+1}$. To show the correctness of the selection procedure, it suffices to show that $G_{2x+1}$ is non-empty. Since $c_i$ is the color that appears the highest number of times among the edges emanating from $v_i$ in $G_{i-1}$, we must have $|V(G_{i})| \geq |V(G_{i-1})-1|/x$, so $|V(G_{2x+1})| = \Omega(n / x^{2x+1}) = n^{1 - o(1)} > 0$, as $x^{2x+1}=2^{O(x \log x)} = 2^{2^{O(\log^{0.3} \log n)}O(\log^{0.3} \log n)} = 2^{2^{O(\log^{0.3} \log n)}} = n^{o(1)}$.
\end{proof}

We select $\{a,b,c,d\} \subseteq I$ according to \Cref{lem:first_step}. The selection already guarantees that from the perspective of $K^i$, \ref{item:SS1} and \ref{item:SS4} are indistinguishable in the first $\tlast$ rounds, regardless of the choice of $i \in [t]$. However, depending on the choices of $\{i, i^\ast\} \subseteq [t]$, $a$ and $b$ might still be able to distinguish between \ref{item:SS1} and \ref{item:SS4} at round $\tlast$ by the messages they receive from $c$. 

\paragraph{The second step} Now we consider the selection of $\{i, i^\ast\} \subseteq [t]$. While our choice of  $\{a,b,c,d\}$ already ensures that the message sent from $a$ to $b$ at round $\tlast$ is the same in both \ref{item:SS1} and \ref{item:SS2}, the message sent from $b$ to $a$ can be different in these two scenarios, as $i \neq i^\ast$. By properly selecting $\{i, i^\ast\} \subseteq [t]$, we can ensure indistinguishability for the messages sent across the edge $\{a,b\}$ inserted in round $\tlast$ across all scenarios. 

\begin{lemma}\label{lem:second_step}
    There exist two distinct indices $i \in [t]$ and $i^\ast \in [t]$ meeting the two conditions.
    \begin{itemize}
        \item The message sent from $b$ to $a$ at round $\tlast$ is the same in both \ref{item:SS1} and \ref{item:SS2}.
        \item The message sent from $a$ to $b$ at round $\tlast$ is the same in both \ref{item:SS1} and \ref{item:SS3}.
    \end{itemize}
\end{lemma}
\begin{proof}
Since $a$ and $b$ are fixed, in \ref{item:SS1}, the two messages sent across $\{a,b\}$ in round $\tlast$ depends only on $i \in [t]$. The number of possibilities for these two messages is $2^{2B} = 2^{o(\log \log \log n)}= \log^{o(1)} \log n$ is smaller than $t = \log^{0.1} \log n$. Therefore, by the pigeonhole principle, there exist two distinct indices $i \in [t]$ and $i^\ast \in [t]$ resulting in identical messages in \ref{item:SS1}. 

To see that the message sent from $b$ to $a$ at round $\tlast$ is the same in both \ref{item:SS1} and \ref{item:SS2}, observe that the messages sent from $b$ in the first $\tlast$ rounds in \ref{item:SS2} remain the same even if we replace $d$ with $a$ (due to \Cref{lem:first_step}) and replace $i^\ast$ with $i$ (due to our choice of $\{i, i^\ast\} \subseteq [t]$). Similarly, we see that the message sent from $a$ to $b$ at round $\tlast$ is the same in both \ref{item:SS1} and \ref{item:SS3}.
\end{proof}

We fix $\{i, i^\ast\} \subseteq [t]$ according to \Cref{lem:second_step}.
Now all the parameters have been fixed. We are ready to finish the proof of \Cref{thm:mixed}.

\begin{proof}[Proof of \Cref{thm:mixed}]
Assuming that $\mathcal{A}$ is correct, at least one node in $K^i \cup \{a,b\}$ is responsible for reporting the $s$-clique $K^i \cup \{a,b\}$ at round $\tlast$ in \ref{item:SS1}. If one node in $K^i$ is responsible for that, then the same node must incorrectly report an $s$-clique in \ref{item:SS4}, as the node cannot distinguish between \ref{item:SS1} and \ref{item:SS4} by \Cref{lem:first_step}. If node $a$ is responsible for that, then the same node must incorrectly report an $s$-clique in \ref{item:SS2}, as node $a$ cannot distinguish between \ref{item:SS1} and \ref{item:SS2} by \Cref{lem:first_step,lem:second_step}. Similarly, if node $b$ is responsible for that, then the same node must incorrectly report an $s$-clique in \ref{item:SS3}, as node $a$ cannot distinguish between \ref{item:SS1} and \ref{item:SS3} by \Cref{lem:first_step,lem:second_step}. Hence $\mathcal{A}$ must be incorrect in at least one scenario, so we conclude that the bandwidth complexity of the considered problem must be $\Omega(\log \log \log n)$.
\end{proof}

\section{Upper Bounds for Finding Triangles}\label{sect:trianglesUB}
In this section, we prove our two upper bounds for finding triangles in dynamic networks. Both two upper bounds match the lower bounds in \Cref{sect:cliqueLB} for bounded-degree networks. The proofs of both theorems are built upon the techniques of Bonne and Censor-Hillel \cite{bonne2019distributed}, which we briefly review in \Cref{subsect:review}. We prove \Cref{thm: Delta memlist ub loglogn} in \Cref{subsec: Delta memlist ub loglogn} and prove \Cref{thm: Delta list lb logloglogn} in \Cref{subsect:logloglogUB}.

%\paragraph{Organization}  
\subsection{Review of the Previous Approach}\label{subsect:review}
Bonne and Censor-Hillel~\cite{bonne2019distributed} designed a one-round algorithm for $\MemDetect(K_3)$ under edge insertions with bandwidth $B=O(\sqrt{\Delta\log n})$, where $\Delta$ is the maximum degree of the dynamic network. Their algorithm is parameterized by a threshold $d$. The intuition is that the information about the topological changes within the last $d$ rounds is handled differently. In their algorithm, after an edge $\{u,v\}$ is inserted, there are two main types of messages communicated over the edge $\{u,v\}$. %,\gopinath{Edited a bit.}
    \begin{itemize}
        \item In the same round where $\{u,v\}$ is inserted, $u$ and $v$ exchange their \underline{recent records}. Here the recent record of a node $w$ consists of two binary strings of length $d$. The $i$th bit in one string indicates whether $w$ obtains a new edge $i$ rounds ago.  The $i$th bit in the other string indicates whether some neighbor of $w$ obtains a new edge $i$ rounds ago.
        
        %Recent records. For each of $u$ and $v$, it has two strings of length $d$ where the $i$th bit indicates whether or not it obtains a new edge $i$ rounds ago and whether or not its neighbor obtains a new edge $i$ rounds ago.
        
        % Construct $R^e(u)$ as the set of all previous $d$ rounds when $u$ obtains a new edge. Construct $R^s(u)$ as the set of all previous $d$ rounds when $u$ has a neighbor obtaining a new edge. Encode each set using $d$ bits.
        \item After the insertion of $\{u,v\}$, $u$ and $v$ use $d$ rounds of communication to exchange their \ul{lists of neighborhood $\ID$s} in the current graph by splitting the information into $d$ blocks of equal size $O\left(\frac{\Delta\log n}{d}\right)$.
        %Neighborhood $\ID$s. For each of $u$ and $v$, it computes the $\ID$s of all its current neighbors and split them into $d$ blocks equally. Each block takes $O(\Delta\log n/d)$ bits. Send one block at each following round.
    \end{itemize}
    
    For the correctness of the algorithm, suppose an edge insertion of $\{u,v\}$ leads to a triangle $\{u,v,w\}$. Without loss of generality, assume $t_{uv}>t_{vw}>t_{uw}$, where $t_{xy}$ is the round number where $\{x,y\}$ is inserted. There are two cases.
    %it has two cases according to the interval between the second-inserted edge and final-inserted edge.
    \begin{description}
    \item[Case 1\label{item: short interval}] If  $t_{uv} - t_{vw} \leq d$, then both $u$ and $v$ can identify $w$ as their common neighbor and list the triangle $\{u,v,w\}$ using the recent records communicated along the edge $\{u,v\}$.
    \item[Case 2\label{item: long interval}] Otherwise, $t_{uv} - t_{vw} > d$, so $v$ has already received from $w$ the complete list of neighborhood $\ID$s of $w$ in round $t_{vw}$, which includes $\ID(u)$. Therefore, $v$ can list the triangle $\{u,v,w\}$ without any further communication. Moreover, $v$ can  \emph{inform} $u$ that a triangle involving $u$ is detected using one bit of communication along the edge $\{u,v\}$. 
    \end{description}
    
    The overall complexity $B=O\left(d+\frac{\Delta\log n}{d}\right)$ can be optimized to $B=O(\sqrt{\Delta\log n})$ by taking  $d=\Theta(\sqrt{\Delta\log n})$. This algorithm works for $\MemDetect(K_3)$ and not $\MemList(K_3)$  since $v$ can only \emph{inform} $u$ the existence of a triangle in \ref{item: long interval}. Due to the bandwidth constraint, $v$ cannot include $\ID(w)$ in the message.

\subsection{Proof of \texorpdfstring{\Cref{thm: Delta memlist ub loglogn}}{Theorem \ref{thm: Delta memlist ub loglogn}}}
%{An \texorpdfstring{$O(\Delta^2 \log\log n)$}{O(Delta2 log log n} Upper Bound of \texorpdfstring{$\MemList(K_3)$}{\MemList(K3)} Under Edge Insertions}
\label{subsec: Delta memlist ub loglogn}

In this section, we prove \Cref{thm: Delta memlist ub loglogn}.

\loglogUB*

% \begin{theorem}\label{thm: Delta memlist ub loglogn}
%     For any graph $G$ with maximum degree $\Delta$, there exists an one-round algorithm of $\MemList(K_3)$ under edge insertions with bandwidth $O(\Delta^2 \log \log n)$.
% \end{theorem}
\begin{proof} 
The proof relies on modifying the two main types of messages presented in \Cref{subsect:review}:
    \begin{description}
        \item[Recent records\label{technique: rr}] Given an integer $d$, we define $R_d^{\mathsf{edge}}(u)\subseteq [d]$ as the subset of all previous $d$ rounds where $u$ obtains a new edge, and we define $R_d^{\mathsf{signal}}(u)\subseteq [d]$ as the subset of all previous $d$ rounds where $u$ has a neighbor obtaining a new edge. 
        
        Due to the degree upper bound $\Delta$, $\left|R_d^{\mathsf{edge}}(u)\right| \leq \Delta$ and $\left|R_d^{\mathsf{signal}}(u)\right| \leq \Delta^2$. As each element of $[d]$ can be encoded using $O(\log d)$ bits, the recent records $R_d^{\mathsf{edge}}(u)$ and $R_d^{\mathsf{signal}}(u)$ can be encoded into a message of $B=O(\Delta^2\log d)=O(\Delta^2\log\log n)$ bits, for any $d=O(\log n)$. This encoding is more efficient than the one of Bonne and Censor-Hillel~\cite{bonne2019distributed} when $\Delta$ is small.
        
        \item[Periodic updates\label{technique: pu}] For each edge $\{u,v\}$ in the network, we let $u$ and $v$ update to each other their list of neighbors periodically with period length $T$, where $T$ is some given integer.

        %For each node $u$ with edge $\{u,v\}$, as graph is changing, node $u$ can keep updating its neighborhood information to $v$ periodically. Given integer $T\geq 1$, we define a period with length $T$ as $T$ consecutive rounds of communication. 

        Suppose node $u$ detects that an edge $\{u,v\}$ is inserted, then $u$ sends to $v$ its current list of neighborhood $\ID$s by splitting the information into $T$ blocks of equal size $O\left(\frac{\Delta\log n}{T}\right)$ and sending one block to $v$ in each of following $T$ rounds. After the transmission is done, the procedure is repeated for a new period, and so on. 
        %It requires $B=O(\Delta\log n/T)$.
        %We call $T$ the period length.

The difference between our approach and the one of Bonne and Censor-Hillel~\cite{bonne2019distributed} here is that we make the exchange of neighborhood $\ID$s periodic. As we will later see, the modification allows us to make the algorithm work for not only $\MemDetect(K_3)$ but also $\MemList(K_3)$.
%    By taking periodic updates of neighborhood, we claim that both two nodes have the complete neighbor $\ID$s in \ref{item: long interval} so that it works for $\MemList(K_3)$. See formal proof below.
    \end{description}

\paragraph{Algorithm} We now describe our algorithm for $\MemList(K_3)$. We choose the parameters $d= \lceil \log n \rceil$ and $T=\lfloor d/2\rfloor$. If node $u$ detects that an edge $\{u,v\}$ is inserted at round $t$, then $u$ performs the following steps.
        \begin{enumerate}
            \item In round $t$, $u$ sends $(\SIGNAL)$ to all its neighbors. 
            \item In round $t$, $u$ sends $R_d^{\mathsf{edge}}(u)$ and $R_d^{\mathsf{signal}}(u)$ to $v$.
            \label{step: recent record}
            \item $u$ starts the periodic neighborhood updates for $\{u,v\}$ with period length $T$ from round $t$.\label{step: periodic update}
            \end{enumerate}

We emphasize that one of the purposes of $(\SIGNAL)$ is to help the neighbors $w$ of $u$ to prepare their recent records  $R_d^{\mathsf{signal}}(w)$.

      \paragraph{Correctness} Suppose a triangle $\{u,v,w\}$ appears at some round. We claim that all three nodes in the triangle have enough information to list the triangle after the communication within the same round. Without loss of generality, we assume $t_{vw} < t_{uw}<t_{uv}$, where $t_{xy}$ is the round number when the edge $\{x,y\}$ is inserted.        
        \begin{itemize}
            \item We first consider node $w$. At the round where edge $\{u,v\}$ is inserted, node $w$ receives $(\SIGNAL)$ from both $u$ and $v$. Thus node $w$ can determine the existence of edge $\{u,v\}$ and list the triangle $\{u,v,w\}$.
            \item For node $u$, we consider two cases:
            \begin{itemize}
                \item Suppose $t_{uv} - t_{uw} > d$. At round $t_{uw} + T-1 \leq t_{uv}$, due to periodic updates, node $u$ receives all neighborhood $\ID$s of $w$ from $w$, which contains $\ID(v)$. Therefore, $u$ is aware of the two edges $\{u,w\}$ and $\{w,v\}$. Immediately after edge $\{u,v\}$ is inserted, $u$ can list the triangle $\{u,v,w\}$.
                \item Suppose $t_{uv} - t_{uw} \leq d$. 
                Set $j=t_{uv} - t_{uw}$. 
               At round $t_{uw}=t_{uv}-j$, $v$ receives $(\SIGNAL)$ from $w$ due to the insertion of $\{u,w\}$ in that round. When edge $\{u,v\}$ is inserted, node $u$ receives $R_d^{\mathsf{signal}}(v)$, with $j\in R_d^{\mathsf{signal}}(v)$, from $v$. Since $j\in R_d^{\mathsf{edge}}(u)$, $u$ can infer the existence of edge $\{v,w\}$ by the fact that in each round at most one topological change occurs. 
                Thus $u$ can list the triangle $\{u,v,w\}$.
            \end{itemize}
            \item For node $v$, we again consider two cases:
            \begin{itemize}
                \item Suppose $t_{uv} - t_{uw} > d$. The interval $[t_{uw}, t_{uw}+2T-1]$ covers one period of neighborhood updates from $w$ to $v$ along the edge $\{v,w\}$ entirely. Therefore, at round $t_{uw} + 2T-1\leq t_{uv}$, node $v$ already knows that the list of neighborhood $\ID$s of $w$ contains $\ID(v)$, so $v$ can infer the existence of edge $\{u,w\}$. Immediately after edge $\{u,v\}$ is inserted, $v$ can list the triangle $\{u,v,w\}$. This case is the reason why we set $T=\lfloor d/2\rfloor$ and not simply $T=d$.%\yijun{More detail needed...note: the argument here should be that an interval of $2T$ guarantees to contain a full period.}
                \item Suppose $t_{uv} - t_{uw} \leq d$. Set $j=t_{uv} - t_{uw}$. 
                %Since edge $\{v,w\}$ is inserted and $u$ receives $(\SIGNAL)$ from $w$ at round $t_{vw}=t_{uv}-j$, 
                At round $t_{uw}=t_{uv}-j$, $v$ receives $(\SIGNAL)$ from $w$ due to the insertion of $\{u,w\}$ in that round.  
                When edge $\{u,v\}$ is inserted, node $v$ receives $R_d^{\mathsf{edge}}(u)$ with $j\in R_d^{\mathsf{edge}}(u)$. Since $j\in R_d^{\mathsf{signal}}(v)$, $v$ can infer the existence of edge $\{u,w\}$ by the fact that in each round at most one change occurs. 
                Thus $v$ can list the triangle $\{u,v,w\}$. This case is similar but not identical to the corresponding case of node $u$: Its correctness relies on $v$ receiving $R_d^{\mathsf{edge}}(u)$ from $u$ and not $u$ receiving $R_d^{\mathsf{signal}}(v)$ from $v$\label{analysis: R^signal} %\yijun{This case is repeated... I think the presentation will be simpler if you do it in a way similar to Section 4.1.}\mingyang{This one is not exactly the same as above: previously $u$ receives $R_d^{\mathsf{signal}}(v)$ and now $u$ receives $R_d^{\mathsf{edge}}(v)$. I thought I should do it separately for the formal analysis.}\yijun{I mean the item that is even above that case, which is symmetrical to this case.}
            \end{itemize}
        \end{itemize}

        \paragraph{Bandwidth} The message sizes for the three steps of the algorithm are $1$, $O(\Delta^2\log d) = O(\Delta^2\log\log n)$, and $O\left(\frac{\Delta\log n}{T}\right) = O(\Delta)$, so the overall bandwidth is $B=O(\Delta^2 \log\log n)$.
\end{proof}

\subsection{Proof of \texorpdfstring{\Cref{thm: Delta list lb logloglogn}}{Theorem \ref{thm: Delta list lb logloglogn}}}
%{An \texorpdfstring{$O(\Delta \log\log\log n)$}{O(Delta log log log n} Upper Bound of \texorpdfstring{$\List(K_3)$}{List(K3)} Under Edge Insertions and Node Insertions}
\label{subsect:logloglogUB}
In this section, we prove \Cref{thm: Delta list lb logloglogn}.
%\yijun{I have not checked this section yet but I find it strange that the section and the section above look quite different. My feeling is that the proof and the algorithm should be exactly the same with just one difference: change neighborhood ID list with the special distinct bit stuff and adjust the parameters accordingly. Ideally, this should be mentioned very clearly and we should not repeat the same proof.}

%we give an one-round algorithm of $\List(K_3)$ under edge insertions and node insertions with bandwidth $O(\Delta\log\log\log n)$ for any dynamic graph with $n$ nodes and maximum degree $\Delta$.
% \begin{theorem}\label{thm: Delta list lb logloglogn}
%     For any graph $G$ with maximum degree $\Delta$, there exists an one-round algorithm of $\List(K_3)$ under edge insertions and node insertions with bandwidth $O(\Delta\log\log \log n)$.
% \end{theorem}

\logloglogUB*

%\paragraph{Intuition} 
We begin by discussing the intuition behind \Cref{thm: Delta list lb logloglogn}.
Recall that the problem $\List(K_3)$ has a bandwidth complexity of $B = O(1)$ under any single type of topological change~\cite{bonne2019distributed}. However, when both edge insertions and node insertions are allowed, \Cref{thm:mixed} establishes a bandwidth lower bound of $B = \Omega(\log\log\log n)$. The proof of this lower bound relies on the inherent difficulty for a node $w$, with two non-adjacent neighbors $x$ and $y$, to distinguish between the insertion of an edge $\{x,y\}$ and the insertion of a node  $w'$ with two incident edges $\{x,w'\}$ and $\{y,w'\}$.
%given $B=o(\log\log\log n)$.

We develop a new technique to handle the hard instance with $B = O(\Delta \log \log \log n)$. For $w$ to distinguish between the above two cases, we fix an index $i$ such that the $i$th bit of $\ID(x)$ does not equal the $i$th bit of $\ID(y)$, and then we ask $x$ and $y$ to report to $w$ the $i$th bit of their new neighbor $\ID$. The two bits differ for the case of insertion of $\{x,y\}$ by our choice of $i$. The two bits are identical for the case of node insertion of $w'$ with two incident edges $\{x,w'\}$ and $\{y,w'\}$, as they are both the $i$th bit of the same string $\ID(w')$. 

We implement this approach within the framework of the proof for \Cref{thm: Delta memlist ub loglogn} by modifying the algorithm for \ref{technique: pu}. The improvement from $O(\log \log n)$ to $O(\log \log \log n)$ arises from the observation that transmitting such an index $i$ requires only $O(\log \log n)$ bits, which is exponentially more efficient than sending the entire $\ID$, which requires $O(\log n)$ bits.

%We design an algorithm to handle the above hard instance. Roughly speaking, if the time interval between the last two changes is short, one of $x$ and $y$ can do the listing using \ref{technique: rr} technique. Otherwise, both $x$ and $y$ sends one specific bit of its new neighbor ID to $w$ with following property: the reported bits are the same if and only if the two new neighbors are the same. Since it is a long interval, node $w$ can use \ref{technique: pu} to ensure $x,y$ knows the specific bit before the last change. We use $\firstDiff(\ID(x), \ID(y))$ as the specific bit index with definition below:
\begin{definition}[$\firstDiff$]
    For any two distinct bit strings $x$ and $y$, define $\firstDiff(x, y)$ as the first index whose bit value differs in $x$ and $y$.%\yijun{there is some redundancy. Can just define this for two strings}
\end{definition}

For example, if $x=011011$ and $y=011101$, then $\firstDiff(x, y) = 4$, as $x[i]=y[i]$ for all $i\in\{1,2,3\}$ and $x[4] = 0 \neq 1 =  y[4]$. If $x$ and $y$ are $O(\log n)$-bit identifiers, then $\firstDiff(x, y)$ is $O(\log n)$, which can be written as an $O(\log \log n)$-bit string.

Observe that selecting $i = \firstDiff(\ID(u), \ID(v))$ satisfies the requirement of index $i$ in the discussion above.

\begin{definition}[$\diffList$]
For any node $u$ and any neighbor $v\in N(u)$, define $\diffList_v(u)$ as the list of indices $\firstDiff(\ID(v), \ID(u'))$ for all $u'\in N(u)\setminus \{v\}$. 
\end{definition}

 Since each element of $\diffList_v(u)$ can be encoded using $O(\log \log n)$ bits, $\diffList_v(u)$ can be represented using $O(\Delta\log \log n)$ bits. Next, we present the modified algorithm for \ref{technique: pu}.

 \paragraph{Periodic updates of $\diffList_v(u)$} Suppose node $u$ detects that an edge $\{u,v\}$ is inserted, then $u$ sends to $v$ its current $\diffList_v(u)$ by splitting the information into $T$ blocks of equal size $O\left(\frac{\Delta\log \log n}{T}\right)$ and sending one block to $v$ in each of following $T$ rounds. After the transmission is done, the procedure is repeated for a new period, and so on. Compared with \ref{technique: pu} in the proof of \Cref{thm: Delta memlist ub loglogn}, the only difference is that the list of neighborhood $\ID$s is replaced with  $\diffList_v(u)$.

% \begin{description}
%     \item[Property of one-DistinctBit\label{item: firstDiff}] Given node $w$ with two non-adjacent neighbors $x,y$. Set $j=\firstDiff(\ID(x), \ID(y))$. Suppose both $x,y$ know the value of $j$. At some round, node $x$ has a new neighbor $x'$ and $y$ has a new neighbor $y'$. If it is node insertions, $x'=y'$, the $j$th bit of $\ID(x')$ and $j$th bit of $\ID(y')$ is the same. If it is edge insertions, $x'=y$ and $y'=x$. Since $j=\firstDiff(\ID(x), \ID(y))$, the $j$th bit of $\ID(x')$ differs from that of $\ID(y')$. Therefore, node $u$ can distinguish two cases according to the received $j$th bit value from $x$ and $y$.
%     \item[one-DBList] Given node $u$ with its neighborhood $N(u)$, for each $v\in N(u)$, we construct a list of bit index $\firstDiff(\ID(v), \ID(u'))$ for each $u'\in N(u)\setminus \{v\}$. Denote it by $\diffList_v(u)$.
%     Since each $\ID$ has length $O(\log n)$, its index can be encoded using $O(\log \log n)$ bits. Given maximum degree $\Delta$, $\diffList_v(u)$ takes total $O(\Delta\log \log n)$ bits.
%     \item[Periodic updates of one-DBList] Starting from edge $\{u,v\}$ inserted, $u$ computes $\diffList_v(u)$, split into $T$ blocks evenly and send one block to $v$ in each of following $T$ rounds. When it finishes at the end of one period, repeat it for a new period. 
% \end{description}

\begin{proof}[Proof of \Cref{thm: Delta list lb logloglogn}]
Our algorithm for $\List(K_3)$ is as follows.
\paragraph{Algorithm} We choose the parameters $d= \lceil \log \log n \rceil$ and $T=\lfloor d/2\rfloor$. If node $u$ detects that an edge $\{u,v\}$ is inserted at round $t$, then $u$ performs the following steps.
        \begin{enumerate}
            \item In round $t$, $u$ sends $(\SIGNAL)$ to all its neighbors.  \label{step1}
            \item In round $t$, for each $u' \in N(u)$, $u$ sends to $u'$ the $j$th bit of $\ID(v)$ for each $j\in \diffList_u(u')$ according to the latest $\diffList_u(u')$ that $u$ received from $u'$. \label{step:report j-th value}
            \item In round $t$, $u$ sends $R_d^{\mathsf{edge}}(u)$ to $v$ using $O(\Delta \log d)$ bits.\label{step:rr edge only}
            \item $u$ starts the periodic neighborhood updates of $\diffList_v(u)$ for $\{u,v\}$ with period length $T$ from round $t$.\label{step:pu diffList}
            \end{enumerate}
Other than the use of $\diffList$ and the parameter choice $d= \lceil \log \log n \rceil$, a major difference between the algorithm here and the algorithm of \Cref{thm: Delta memlist ub loglogn} is that the nodes do not communicate $R_d^{\mathsf{signal}}$ in the algorithm above, which allows us to improve the $O(\Delta^2)$ factor in the bandwidth complexity to $O(\Delta)$. The communication of $R_d^{\mathsf{signal}}$ is not required here because the problem under consideration is listing and not membership-listing.

% Take $d= \lceil \log \log n \rceil$. For each existing node $u$, suppose $u$ has a new neighbor $v$ at round $t$, send the following messages:
%         \begin{enumerate}
%             \item $u$ sends $(\SIGNAL)$ to all its neighbors.\yijun{why not $v$?}\mingyang{fixed.}
%             \item If $u$ receives the latest $\diffList_u(u')$ from its neighbor $u'$, $u$ sends $u'$ the $j$th bit of $\ID(v)$ for each $j\in \diffList_u(u')$.\label{step:report j-th value}
%             \item $u$ sends $R_d^{\mathsf{edge}}(u)$ to $v$ using $O(\Delta \log d)$ bits.\label{step:rr edge only}
%             \item $u$ starts the periodic updates of $\diffList_v(u)$ to $v$ with period length $T=\lfloor d/2\rfloor$. \label{step:pu diffList}
%         \end{enumerate}
    \paragraph{Correctness} Suppose a triangle  $\{u,v,w\}$ appears at some round $t$. We show that at least one node has enough information to list the triangle after the communication in the same round.
    \begin{description}
        \item [Case 1: node insertion\label{item: node insertion}] Suppose the triangle $\{u,v,w\}$ is formed by a node insertion. Without loss of generality, we assume that edge $\{u,v\}$ is already there before round $t$, and then at round $t$, node $w$ is inserted along with incident edges $\{w,v\}$ and $\{w,u\}$. In \Cref{step1}, $u$ and $v$ send $(\SIGNAL)$ to each other. This happens if and only if $u$ and $v$ obtain a common new neighbor $w$. Thus both $u$ and $v$ can list the triangle $\{u,v,w\}$ in this case. 
        \item [Case 2: edge insertion\label{item: edge insertion}] Suppose the triangle $\{u,v,w\}$ is formed by an edge insertion. We claim that one node can list the triangle after the communication within the same round. Without loss of generality, we assume $t_{vw} < t_{uw}<t_{uv} = t$, where $t_{xy}$ is the round number when the edge $\{x,y\}$ is inserted.        
        \begin{itemize}
            \item Suppose $t_{uv} - t_{uw} > d$. We claim that $w$ can correctly list the triangle $\{u,v,w\}$ in this case. The interval $[t_{uw}, t_{uw}+2T-1]$ covers one period of updates of $\diffList$ entirely, from $w$ to both $u$ and $v$. Therefore, at round $t_{uw} + 2T-1\leq t_{uv} = t$, both the list $\diffList_u(w)$ that node $u$ receives from $w$ and the list $\diffList_v(w)$ that node $v$ receives from $w$ already contain $\firstDiff(\ID(u), \ID(v))$. In the subsequent discussion, we write $j = \firstDiff(\ID(u), \ID(v))$. In  \Cref{step:report j-th value} of round $t$, node $w$ receives the $j$th bit of $\ID(v)$ from $u$ and the $j$th bit of $\ID(u)$ from $v$. Based on the value of these two bits, $w$ can correctly tell whether the triangle $\{u,v,w\}$ is formed due to the insertion of edge $\{u,v\}$.
            \begin{itemize}
            \item Observe that the edge $\{u,v\}$ is inserted \emph{if and only if} $w$ receives $(\SIGNAL)$ from exactly two of its neighbors $u$ and $v$ \emph{and} the two  $\firstDiff(\ID(u), \ID(v))=j$th bits reported from $u$ and $v$ are distinct. In this case, the triangle $\{u,v,w\}$ is formed, and $w$ can correctly list the triangle. 
                \item On the other hand, if the two reported bits are the same, this indicates that a new node was inserted along with edges incident to both $u$ and $v$, allowing $w$ to conclude that the triangle $\{u,v,w\}$ is not formed.
            \end{itemize}
            
            \item Suppose $t_{uv} - t_{uw} \leq d$. By using the same argument as the case of $t_{uv} - t_{uw} \leq d$ for node $v$ in the proof of \Cref{thm: Delta memlist ub loglogn}, we infer that $v$ can correctly list the triangle $\{u,v,w\}$ in this case. We emphasize that the argument only relies on the fact that $v$ can locally calculate $R_d^{\mathsf{signal}}(v)$ and does not require $v$ to send  $R_d^{\mathsf{signal}}(v)$ to $u$, so the proof still works here.%\yijun{I don't know why this case was commented before. Now added back and will check later...}
        \end{itemize}
        \item[Bandwidth] We analyze the bandwidth complexity of our algorithm. The cost for \Cref{step1} is $B = 1$. The cost for \Cref{step:report j-th value} is $B=O(\Delta)$  since 
        $\diffList_{u}(u')$ contains at most $\Delta-1$ indices.
        The cost for  \Cref{step:rr edge only} is $B=O(\Delta \log d)=O(\Delta\log\log\log n)$ since $d= \lceil \log \log n \rceil$.
        The cost for \Cref{step:pu diffList} is $B= O\left(\frac{\Delta\log \log n}{T}\right)=O(\Delta)$ since $T=\lfloor d/2\rfloor$ and $d= \lceil \log \log n \rceil$. Therefore, the overall bandwidth is  $B=O(\Delta\log \log\log n)$.\qedhere
    \end{description}
\end{proof}

%\newpage
\section{Membership-Listing}\label{sect:memList}
In this section, we establish a complete characterization of the bandwidth complexity of $r$-round dynamic $\MemList(H)$ for any target subgraph $H$, for any number of rounds $r$, under any single type of topological change. We emphasize that, while $H$ is assumed to be a constant-size graph, here we allow $r$ to be a function of $n$. See \Cref{table:memList1} for a summary of our results. Refer to \Cref{def: r_H,def: r_H'} for the definition of $r_H$ and $r_H'$.

% \usepackage{multirow}
% \usepackage{multirow}
\begin{table}[ht]
\begin{center}
\begin{tabular}{llllll}
\cline{2-6}
\multicolumn{1}{l|}{} & \multicolumn{4}{l|}{$r\geq r_H$} & \multicolumn{1}{l|}{\multirow{4}{*}{$r<r_H$}} \\ \cline{2-5}
\multicolumn{1}{l|}{} & \multicolumn{3}{l|}{Complete multipartite graphs $(r_H=1)$} & \multicolumn{1}{l|}{\multirow{3}{*}{$r_H \geq 2$}} & \multicolumn{1}{l|}{} \\ \cline{2-4}
\multicolumn{1}{l|}{} & \multicolumn{2}{l|}{Cliques} & \multicolumn{1}{l|}{\multirow{2}{*}{Others}} & \multicolumn{1}{l|}{} & \multicolumn{1}{l|}{} \\ \cline{2-3}
\multicolumn{1}{l|}{} & \multicolumn{1}{l|}{$r = 1$} & \multicolumn{1}{l|}{$r\geq 2$} & \multicolumn{1}{l|}{} & \multicolumn{1}{l|}{} & \multicolumn{1}{l|}{} \\ \cline{2-6} 
\multicolumn{1}{l|}{\multirow{2}{*}{\textbf{Edge insertions}}} & \multicolumn{1}{l|}{$\Theta(\sqrt{n})$} & \multicolumn{1}{l|}{$\Theta(1)$} & \multicolumn{1}{l|}{$\Theta(n/r)$} & \multicolumn{1}{l|}{$\Theta(n^2/r)$} & \multicolumn{1}{l|}{Impossible} \\
\multicolumn{1}{l|}{} & \multicolumn{1}{l|}{\cite{bonne2019distributed}} & \multicolumn{1}{l|}{\cite{bonne2019distributed}} & \multicolumn{1}{l|}{[\ref{Lem: memlist Multipartite graph}][\ref{lem: memlist lb edge insertions complete}]} & \multicolumn{1}{l|}{[\ref{Lem: memlist other graph}][\ref{lem: memlist lb edge insertions other}]} & \multicolumn{1}{l|}{[\ref{lem: locality edge}][\ref{lem: locality node insertion}]} \\ \cline{2-5}
\multicolumn{1}{l|}{\multirow{2}{*}{\textbf{Node insertions}}} & \multicolumn{2}{l|}{$\Theta(n/r)$} & \multicolumn{1}{l|}{$\Theta(n/r)$} & \multicolumn{1}{l|}{$\Theta(n^2/r)$} & \multicolumn{1}{l|}{} \\
\multicolumn{1}{l|}{} & \multicolumn{2}{l|}{\cite{bonne2019distributed}} & \multicolumn{1}{l|}{[\ref{Lem: memlist Multipartite graph}][\ref{lem: memlist lb node insertions complete}]} & \multicolumn{1}{l|}{[\ref{Lem: memlist other graph}][\ref{lem: memlist lb node insertions other}]} & \multicolumn{1}{l|}{} \\ \cline{2-5}
\multicolumn{1}{l|}{\multirow{2}{*}{\textbf{Edge deletions}}} & \multicolumn{2}{l|}{$\Theta(1)$} & \multicolumn{2}{l|}{$\Theta((\log n)/r)$} & \multicolumn{1}{l|}{} \\
\multicolumn{1}{l|}{} & \multicolumn{2}{l|}{\cite{bonne2019distributed}} & \multicolumn{2}{l|}{[\ref{lem: memlist edge deletion}][\ref{lem: memlist lb edge deletion}]} & \multicolumn{1}{l|}{} \\ \cline{2-6} 
 &  &  &  &  &  \\ \cline{2-6} 
\multicolumn{1}{l|}{} & \multicolumn{4}{l|}{$r\geq r_H'$} & \multicolumn{1}{l|}{\multirow{2}{*}{$r<r_H'$}} \\ \cline{2-5}
\multicolumn{1}{l|}{} & \multicolumn{2}{l|}{Cliques $(r_H'=0)$} & \multicolumn{2}{l|}{Others $(r_H'\geq 1)$} & \multicolumn{1}{l|}{} \\ \cline{2-6} 
\multicolumn{1}{l|}{\multirow{2}{*}{\textbf{Node deletions}}} & \multicolumn{2}{l|}{$0$} & \multicolumn{2}{l|}{$\Theta((\log n)/r)$} & \multicolumn{1}{l|}{Impossible} \\
\multicolumn{1}{l|}{} & \multicolumn{2}{l|}{\cite{bonne2019distributed}} & \multicolumn{2}{l|}{[\ref{lem: memlist node deletion}][\ref{lem: memlist lb node deletion}]} & \multicolumn{1}{l|}{[\ref{lem: locality node deletion}]} \\ \cline{2-6} 
\end{tabular}
\caption{The bandwidth complexity of $r$-round $\MemList(H)$.\label{table:memList1}}
\end{center}
\end{table}

In \Cref{subsec: ub}, we present our two new upper bounds  $O(n/r)$ and $O(n^2/r)$ for $\MemList(H)$ under any single type of topological change. 
In \Cref{subsec: lb edge insertion}, we present the  corresponding matching lower bounds for edge insertions.
In \Cref{subsec: lb node insertion}, we present the  corresponding matching lower bounds for node insertions.
In \Cref{subsec: deletion}, we show the tight $r$-round bandwidth complexity bound $\Theta((\log n)/r)$ for $\MemList(H)$ for edge deletions and node deletions.

\subsection{Upper Bounds}\label{subsec: ub}

In this section, we prove \Cref{Lem: memlist Multipartite graph,Lem: memlist other graph}, which show \emph{tight} $r$-round bandwidth complexity upper bounds for complete multipartite graphs and other graphs, respectively.

%In this section, we prove \Cref{Lem: memlist Multipartite graph}, which gives $r$-round bandwidth $O(n/r)$ of $\MemList(H)$ under any type of topological change for complete multipartite graph $H$. For all other graphs, we prove \Cref{Lem: memlist other graph}, which gives $r$-round bandwidth $O(n^2/r)$ for $r\geq r_H$.

\begin{theorem}\label{Lem: memlist Multipartite graph}
    For any connected complete multipartite graph $H$ and for any $r \geq 1$, the $r$-round bandwidth complexity of $\MemList(H)$ under any type of topological change is $B=O(n/r)$.
\end{theorem}
\begin{theorem}\label{Lem: memlist other graph}
    For any connected graph $H$ and for any $r\geq r_H$, the $r$-round bandwidth complexity of $\MemList(H)$ under any type of topological change is $B=O(n^2/r)$.
\end{theorem}
The proofs of the above theorems rely on the notion of the \emph{$r$-radius local view} of a node $u$, which is defined as the subgraph that $u$ can possibly learn after $r$ rounds of communication with unlimited bandwidth, assuming that each node initially only knows its neighbors. A more precise definition is as follows.
\begin{description}
    \item[Initial local view] In a graph $G$, without any communication, each node $u\in V(G)$ initially knows its neighbors $N_G(u)$, so we define the zero-radius local view $B_G^0(u)$ as follows.
    \begin{itemize}
        \item $V(B_G^0(u))=N_G(u)\cup \{u\}= \{w\in V(G): \dist_G(u,w)\leq 1\}$.
        \item $E(B_G^0(u))=\{\{u,u'\}: u' \in N_G(u)\}$.
    \end{itemize}
    %We say each node $u$ has a $0$-radius local view, $B_G^0(u)$, without any communication.
    \item [Local view of a given radius] After $r$ rounds of communication, a node $u$ can receive information from all nodes with distance at most $r$, each of which has a zero-radius local view initially, so we define the $r$-radius local view $B_G^r(u)$ as the union of $B_G^0(v)$ over all nodes $v$ such that $\dist_G(u,v)\leq r$.
    In other words, $B_G^r(u)$ is induced by the set of edges with at least one endpoint whose distance from $u$ is at most $r$.
    %For each integer $r\geq 1$, we define $B_G^r(u)\subseteq G$ as follows:
    \begin{itemize}
        \item $V(B_G^r(u)) = \{w\in V(G): \dist_G(u,w)\leq r+1\}$.
        \item $E(B_G^r(u)) = \{\{w,v\}\in E(G): \min \{\dist_G(u,w), \dist_G(u,w)\} \leq r \}$.
    \end{itemize}
    We emphasize that, in general, $B^r_G(u)$ is not an induced subgraph, as $B^r_G(u)$ does not contain edges whose both endpoints are at exactly distance $r+1$ from $u$.
\end{description}
In the subsequent discussion, we say that a node $u$ has the $r$-radius local view $B_G^r(u)$ if $u$ knows $B_G^r(u)$. We make the following observation.

\begin{observation}\label{obs: isomorphic}
    A graph $H$ is a connected complete multipartite graph if and only if $B^1_H(u) = H$ for all $u \in V(H)$.
\end{observation}
\begin{proof}
    For the forward direction, for any given $u \in V(H)$ in any complete multipartite graph $H$, at least one endpoint $v$ of each edge $e \in E(H)$ belongs to a part that does not contain $u$, so $\dist(u,v)=1$ and $e \in E(B_H^1(u))$. For the opposite direction, suppose $B^1_H(u) = H$ for all $u \in V(H)$. We write $u \leftrightarrow v$ if $u$ and $v$ are non-adjacent. Observe that $\leftrightarrow$ is an equivalence relation: If $\{u,v\}\in E(H)$, $u \leftrightarrow w$, and $v \leftrightarrow w$, then $\{u,v\} \notin E(B_H^1(u))$, contradicting the assumption $B^1_H(u) = H$. Therefore, $H$ is a complete multipartite graph where the parts are the equivalence classes of $\leftrightarrow$.
\end{proof}

\begin{lemma}\label{lem: 2-N}
     For any $r \geq 1$, there exists an algorithm that uses $r$ rounds of communication with bandwidth $B=O(n/r)$ to let every node have the one-radius local view.
\end{lemma}
\begin{proof}
    Consider the following one-round algorithm: Each node $u$ encodes $N_G(u)$ as an $n$-bit string and sends it to all its neighbors $N_G(u)$. Before the algorithm starts, each node $u$ has the local view $B^0_G(u)$. For each neighbor $w\in N_G(u)$, node $u$ receives $N_G(w)$ from $w$. Thus, $u$ learns the set of all edges incident to any neighbor of $u$. By combining the messages from all neighbors, node $u$ learns $E(B_G^1(u))$ and $V(B_G^1(u))$, which gives the one-radius local view $B_G^1(u)$. 
    
    The aforementioned one-round algorithm can be simulated using $r$ rounds with bandwidth $B=O(n/r)$ by breaking each $n$-bit message into $r$ blocks of equal size $O(n/r)$ and sending them using $r$ rounds.
\end{proof}

\begin{proof}[Proof of \Cref{Lem: memlist Multipartite graph}]   
    Let $G$ be the graph resulting from the last topological change. For each node $u\in V(G)$, let $\mathcal{H}_{u}$ be the collection of all copies of $H$ in $G$ containing node $u$. To solve the membership-listing problem, we need to let every node $u$ learn $\mathcal{H}_{u}$.
    Using the algorithm of \Cref{lem: 2-N}, after $r$ rounds of communication, every node $u\in V(G)$ has the local view $B_{G}^1(u)$. By \Cref{obs: isomorphic}, this is sufficient for $u$ to output $\mathcal{H}_{u}$ correctly. 
    % For each $H'\in \mathcal{H}_u$, it has $H'\subseteq B^1_G(u)$ since $H$ is a complete multipartite graph based on \Cref{obs: isomorphic}. 
    % For each $H'\in B^1_G(u)$ containing node $u$, $H'\subseteq G$ and $H'\in \mathcal{H}_u$. 
    % Thus each node $u$ can list all elements in $\mathcal{H}_u$ correctly based on $B^1_G(u)$.
\end{proof}
\begin{lemma}\label{lem: algo 2}
     For any $t \geq 1$ and $r \geq t$, there exists an algorithm that uses $r$ rounds of communication with bandwidth $B=O(n^2t/r)$ to let every node have the $t$-radius local view.
\end{lemma}
\begin{proof}
We can let every node $u$ learn $B^t_G(u)$ in $t$ rounds by letting every node broadcast to its neighbors all the edges it has learned in every round. Specifically, we write $M_u$ to denote the set of edges that $u$ has learned. Initially, $M_u = E(B^0_G(u))$.
For each of $t$ rounds, each node $u$ sends $M_u$ to all its neighbors. At the end of each round, each node $u$ updates $M_u$ to $M_u \cup \bigcup_{v \in N_G(u)} M_v$. Since $M_u$ can be encoded as an $n(n-1)/2$-bit string, the algorithm has bandwidth complexity $O(n^2)$.

Similar to the proof of \Cref{lem: 2-N}, we can reduce the bandwidth complexity from $O(n^2)$ to $O(n^2t/r)$ at the cost of increasing the number of rounds from $t$ to $r$. This is achieved by breaking each $O(n^2)$-bit message into $d=\lfloor \frac{r}{t}\rfloor$ blocks of equal size $O(n^2t/r)$ and sending them using $d$ rounds.
    % We claim that for any node $u$, $\{w,w'\}\in E(B_G^t(u))$ if and only if $\{w,w'\}\in M_u$ at the end of round $t$. 
    % Suppose $\{w,w'\}\in E(B_G^t(u))$, let $w$(w.l.o.g.) satisfy that $\dist_G(u,w)\leq t$. Therefore edge $\{w,w'\}$ is added into $M_u$ within $t$ rounds. On the other hand, suppose $M_u$ contains $\{w,w'\}$ after $t$ rounds, either $\dist_G(u,w)\leq t$ or $\dist_G(u,w')\leq t$. It implies that $\{w,w'\}\in E(B_G^t(u))$.
    % Now consider the general $r$-round algorithm with any $r\geq t$. Set $d=\lfloor \frac{r}{t}\rfloor$. For each round in the above $t$-round algorithm, it can be simulated using $d$ rounds: break the $n(n-1)/2$-bit string into $d$ blocks equally and send it to all its neighbors, one block on every round.
    % Therefore, there exists an $r$-round algorithm for $t$-radius local view with $B=O(n^2t/r)$.
\end{proof}

\begin{proof}[Proof of \Cref{Lem: memlist other graph}]
    Let $G$ be the graph resulting from the last topological change. 
    %Fixed any node $u$. 
    Set $t=r_H$ and consider any $r\geq t$. The $r$-round algorithm of \Cref{lem: algo 2} lets every node $u$ have the local view $B^{t}_G(u)$ with bandwidth $B=O(n^2t/r) = O(n^2 /r)$, as $t=r_H = O(1)$. 
    
    We claim that the edge set of $B^{t}_G(u)$ is sufficient for $u$ to list all copies of $H$ in $G$ containing $u$. 
    %re exists subgraph $H$ containing $u$ if and only if there exists $H\subseteq B^{r_H}_G(u)$. 
    Fix any subgraph $H$ of $G$ that contains $u$. Consider any edge $e=\{v,w\}\in E(H)$. Observe that \[1 + \min\{\dist_{G}(u,w), \dist_{G}(u,v)\} \leq 1 + \min\{\dist_{H}(u,w), \dist_{H}(u,v)\} \leq \nediam(H) 
 = 1+ r_H= 1+t.\] 
    Therefore, the distance from $u$ to at least one of $v$ and $w$ is at most $t$, so $e=\{v,w\}\in B^{t}_G(u)$.
    % Since $\ecc_{H/e}(v_e) \leq r_H$, it has $\dist_{H/e}(u,v_e)\leq r_H$, It implies either $\dist_{G}(u,w)\leq r_H$ or $\dist_{G}(u,v)\leq r_H$. Therefore, $H \subseteq B^{r_H}_G(u)$. On the other hands, For each copy $H\subseteq B^{r_H}_G(u)\subseteq G$, graph $G$ contains $H$. Therefore we claim that $u$ can output correctly using $B^{r_H}_G(u)$.
\end{proof}

\subsection{Lower Bounds Under Edge Insertions}\label{subsec: lb edge insertion}
In this section, we show $r$-round bandwidth complexity lower bounds for $\MemList$ under edge insertions that match the upper bounds in \Cref{Lem: memlist Multipartite graph,Lem: memlist other graph}. We show that the bandwidth complexity is $B=\Omega(n/r)$ for any target subgraph that is not a clique (\Cref{lem: memlist lb edge insertions complete}) and is $B=\Omega(n^2/r)$ for any target subgraph that is not a complete multipartite graph (\Cref{lem: memlist lb edge insertions other}). 

\begin{theorem}\label{lem: memlist lb edge insertions complete}
    For any graph $H$ that is not a clique and for any $r\geq 1$, 
    the $r$-round bandwidth complexity of $\MemList(H)$ under edge insertions is $B=\Omega(n/r)$.
\end{theorem}
\begin{proof}
Since $H$ is not a clique, there exist two nodes $u,v\in V(H)$ such that $\{u,v\} \notin E(H)$. Let $W=V(H)\setminus \{u,v\}$ be the set of all nodes in $H$ excluding $u$ and $v$. 
    To obtain the desired bandwidth complexity lower bound $B=\Omega(n/r)$, the proof idea is to consider the graph resulting from replacing $u$ in $H$ with an independent set $U'$ of size $\Omega(n)$. This graph contains $\Omega(n)$ copies of $H$ that share all nodes except for $u$. 
    We show that if the graph is constructed in such a way that the edges incident to $v$ are added in the end, then $B=\Omega(n/r)$ is necessary for $v$ to learn all $\Omega(n)$ copies of $H$.

    We now give the precise construction of the dynamic network. Refer to \Cref{fig:memlist-ind-pair} for an illustration.   
    %The graph is constructed as follows. Starting from the subgraph of $H$ induced by $W$, we choose $\Omega(n)$ copies of $u$ and each copy is connected to all neighbors of $u$ in $H$. Finally, we add edges to connect $v$ with all its neighbors in $H$. See Figure \ref{fig:memlist-ind-pair} for illustration. We prove that node $v$ requires $B=\Omega(n/r)$ to ensure correctness. See the formal description below:
    Set $U = \{u_1, u_2, \ldots u_n\}$. Consider the dynamic graph $\mathcal{G}$ with node set $W \cup \{v\} \cup U$ of size $O(n)$. Here is construction sequence of $\mathcal{G}$:
    \begin{enumerate}
        \item Start from the induced subgraph $H[W]=H-\{u,v\}$ with all other nodes isolated.
        \item Pick $n/2$ nodes from $U = \{u_1, u_2, \ldots u_n\}$ to form node set $U'$. For each $u_i \in U'$, add edges to connect $u_i$ with $N_H(u)\subseteq W$.\label{step: V'}
        \item Add edges to connect node $v$ with $N_H(v) \subseteq W$. We denote the resulting graph by $G_{U'}$.
    \end{enumerate}   
    Intuitively, each $u_i\in U'$ can be seen as a duplicate of $u$, as for each $w \in W$, $\{u_i,w\} \in E(G_{U'})$ if and only if $\{u,w\} \in E(H)$. 
    Indeed, $G_{U'}$ is the result of replacing $u$ in $H$ with an independent set $U'$, so each $u_i\in U'$ belongs to one unique copy of $H$, which contains node $v$. We write $H_i$ to denote this copy of $H$, so $\mathcal{H}_{U'}=\{H_i: u_i \in U'\}$ is the collection of all copies of $H$ containing node $v$ in $G_{U'}$. To solve the membership-listing problem, node $v$ is required to output $\mathcal{H}_{U'}$ after $r$ rounds of communication from the last topological change. 
    
    There are $\binom{n}{n/2}$ distinct choices of $U'\subseteq U$ at \Cref{step: V'}, where each choice of $U'$ corresponds to a distinct output $\mathcal{H}_{U'}$ of node $v$.
    Since $v$ is not adjacent to any $u_i \in U'$, the output of $v$ only depends on the messages that $v$ receive from $N_H(v)$. Set $d=|N_H(v)| = O(1)$. The $d$ edges incident to $v$ are added in the last $d$ topological changes. Since each topological change is followed by $r$ rounds of communication, $v$ receives $x = dr + (d-1)r + \cdots + r= O(r)$ messages from its neighbors in total.
    To ensure the correctness of the output from $v$, we must have
    \[
    2^{x \cdot B} \geq \binom{n}{n/2} = 2^{\Omega(n)}, 
    \]
    where $B$ is the message size.
    Therefore, we have $B = \Omega(n/r)$.
\end{proof}

\usetikzlibrary{math}
\tikzmath{\d = 2;} 

% \makeatletter
% \usetikzlibrary{decorations,backgrounds}
% \def\pgfdecoratedcontourdistance{0pt}
% \pgfset{
%   decoration/contour distance/.code=%
%     \pgfmathsetlengthmacro\pgfdecoratedcontourdistance{#1}}
% \pgfdeclaredecoration{contour lineto closed}{start}{%
%   \state{start}[
%     next state=draw,
%     width=0pt,
%     persistent precomputation=\let\pgf@decorate@firstsegmentangle\pgfdecoratedangle]{%
%     \pgfpathmoveto{\pgfpointlineattime{.5}
%       {\pgfqpoint{0pt}{\pgfdecoratedcontourdistance}}
%       {\pgfqpoint{\pgfdecoratedinputsegmentlength}{\pgfdecoratedcontourdistance}}}%
%   }%
%   \state{draw}[next state=draw, width=\pgfdecoratedinputsegmentlength]{%
%     \ifpgf@decorate@is@closepath@%
%       \pgfmathsetmacro\pgfdecoratedangletonextinputsegment{%
%         -\pgfdecoratedangle+\pgf@decorate@firstsegmentangle}%
%     \fi
%     \pgfmathsetlengthmacro\pgf@decoration@contour@shorten{%
%       -\pgfdecoratedcontourdistance*cot(-\pgfdecoratedangletonextinputsegment/2+90)}%
%     \pgfpathlineto
%       {\pgfpoint{\pgfdecoratedinputsegmentlength+\pgf@decoration@contour@shorten}
%       {\pgfdecoratedcontourdistance}}%
%     \ifpgf@decorate@is@closepath@%
%       \pgfpathclose
%     \fi
%   }%
%   \state{final}{}%
% }
% \makeatother
% \tikzset{
%   contour/.style={
%     decoration={
%       name=contour lineto closed,
%       contour distance=#1
%     },
%     decorate}}

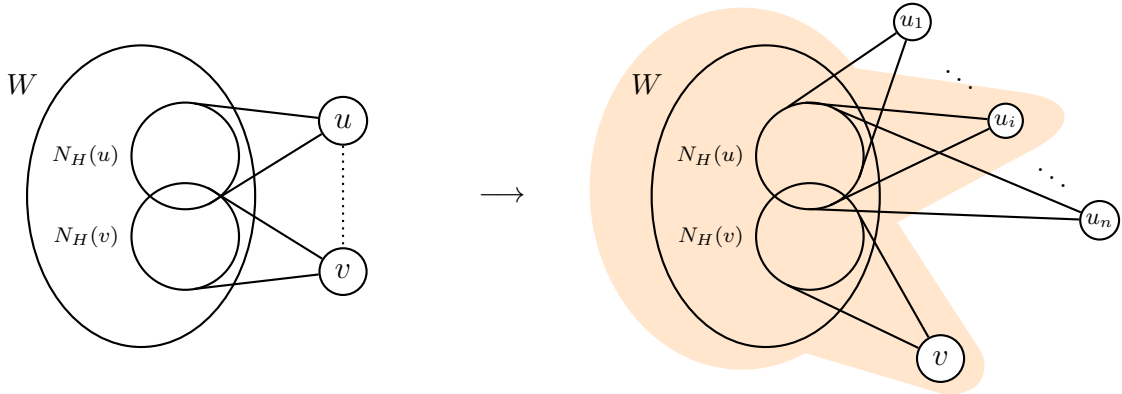
\begin{figure}[htbp]
    \centering
    \begin{tikzpicture}[node distance={10mm},thick, vtx/.style = {draw, circle, fill=white, inner sep=1.2mm, font=\large}, E/.style = {ellipse, draw=black, minimum height=4cm, minimum width=3cm}, dup/.style= {draw, circle, inner sep=0.4mm, font=\footnotesize, fill=white}, blob/.style={draw, circle,inner sep=5mm}, bg/.style={fill=orange!20, draw=none}]
        % \path[draw=black!0]
        %     (8, 2) coordinate (A)
        %     (9, 2) coordinate (B)
        %     (10,-1) coordinate (C)
        %     (10,-4) coordinate (D)
        %     (7, 0) coordinate (E);
        
        % \draw[preaction={contour=10pt, rounded corners, fill=blue!20}]
        %   (A) -- (B) -- (C) -- (D) -- (E) -- cycle;
        \draw[bg] (79mm,1mm) ellipse (2cm and 24mm);

        \filldraw[bg, rounded corners=12mm] 
            (80mm, 19mm) -- (128mm,12mm) -- (80mm,-14mm) -- cycle;
        \filldraw[bg, rounded corners=3mm] 
            (85mm, 19mm) -- (112mm,-23mm) -- (108mm,-27mm) -- (68mm,-15mm) -- cycle;

        \node[E, label=above left: $W$] (0) {};
        \node (00) [right = 10mm of 0]{};
        \node[vtx] (1) [above of=00] {$u$};
        \node[vtx] (2) [below of=00] {$v$};
        \node[blob, label=left: {\scriptsize$N_H(u)$}] (3) [above right = -14mm and -10mm of 0] {};
        \node[blob, label=left: {\scriptsize$N_H(v)$}] (3) [below right = -14mm and -10mm of 0] {};
        \draw[dotted] (1)--(2);
        \draw (1)--(7mm, 12.3mm);
        \draw (1)--(8.8mm, -1.1mm);
        \draw (2)--(8.8mm, 1.1mm);
        \draw (2)--(7mm, -12.3mm);

        \node (arrow) [right = 15mm of 00] {$\longrightarrow$};
        \node[E, label=above left: $W$] (G0) [right = 15mm of arrow] {};
        \node (G00) [right = 15mm of G0]{};
        \node[vtx] (Gv) [below right = 5mm and 10mm of G0] {$v$};
        \node[dup] (G3) [above of=G00] {$u_i$};
        \node[inner sep=0mm] (G2) [above left = \d mm and \d mm of G3] {$\ddots$};
        \node[dup] (G1) [above left = \d mm and \d mm of G2] {$u_1$};
        \node[inner sep=0mm] (Gx) [below right = \d mm and \d mm of G3] {$\ddots$};
        \node[dup] (Gn) [below right = \d mm and \d mm of Gx] {$u_n$};
        
        \node[blob, label=left: {\scriptsize$N_H(u)$}] (3) [above right = -14mm and -10mm of G0] {};
        \node[blob, label=left: {\scriptsize$N_H(v)$}] (3) [below right = -14mm and -10mm of G0] {};
        \draw (Gv)--(94.5mm, -2.6mm);
        \draw (Gv)--(85mm, -11.8mm);
        
        \draw (G1)--(85mm, 11.8mm);
        \draw (G1)--(94.4mm, 2.4mm);
        
        % \draw (G2)--(86.1mm, 12.2mm);
        % \draw (G2)--(93.1mm, 1mm);

        \draw (G3)--(87.4mm, 12.4mm);
        \draw (G3)--(90.6mm, -1.2mm);

        \draw (Gn)--(91.2mm, 11.6mm);
        \draw (Gn)--(87.8mm, -1.7mm);
    \end{tikzpicture}
    \caption{Left: graph $H$, where $\{u,v\}\notin E(H)$. Right: graph $G_{U'}$, where the unique copy $H_i$ of $H$ containing $u_i \in U'$ is highlighted.}
    %\caption{Graph $H$ is on the left with node $u,v$ not adjacent. Graph $G_{U'}$ is on the right. It has node set $W\cup \{v\} \cup U$. It contains $n/2$ copies of $H$ which share node set $W \cup \{v\}$. For each distinct subset $U' \subseteq U$ of size $n/2$, it corresponds a distinct graph $G_{U'}$ and a distinct output $\mathcal{H}_{U'}$ of node $v$.}
    \label{fig:memlist-ind-pair}
\end{figure} 

\begin{theorem}\label{lem: memlist lb edge insertions other}
    For any graph $H$ that is not a complete multipartite graph and for any $r\geq 1$, 
    the $r$-round bandwidth complexity of $\MemList(H)$ under edge insertions is $B=\Omega(n^2/r)$.
\end{theorem}
\begin{proof}
    Since $H$ is not a complete multipartite graph, there exist a node $v\in V(H)$ and an edge $\{u,w\}\in E(H)$ such that $\{v,u\} \notin E(H)$ and $\{v,w\} \notin E(H)$. Let $S=V(H)\setminus \{v,u,w\}$ be the set of all nodes in $H$ excluding $v$, $u$, and $w$. 
    To obtain the desired bandwidth complexity lower bound $B=\Omega(n^2/r)$, the proof idea is to consider the graph resulting from replacing $\{u,w\}$ in $H$ with a bipartite graph with $\Omega(n^2)$ edges, where each edge is contained in a unique copy of $H$. We show that if the graph is constructed in such a way that the edges incident to $v$ are added in the end, then $B=\Omega(n^2/r)$ is necessary for $v$ to learn all $\Omega(n^2)$ copies of $H$.

    % we use a graph containing $\Omega(n^2)$ copies of $H$, which share all nodes except $u$ and $w$. Denote it by $G_C$. 
    % It is constructed as follows: starting from graph $H[S]$, we set up a set $U$ containing $n$ copies of $u$, each connected to all neighbors of $u$ in $H$. Similarly, set up a set $W$ containing $n$ copies of $w$, each connected to all neighbors of $w$ in $H$. 
    % We construct $\Omega(n^2)$ copies of crossing edge $\{u,w\}$ between $U$ and $W$. Finally, we add edges to connect $v$ with all its neighbors in $H$. 
    %See \Cref{fig: memlist lb other} for an illustration. We prove that node $v$ requires $B=\Omega(n^2/r)$ to ensure correctness. See the formal description below:

        We now give the precise construction of the dynamic network. Refer to \Cref{fig: memlist lb other} for an illustration.   
    Set $U = \{u_1, u_2, \ldots u_n\}$ and $W = \{w_1, w_2, \ldots w_n\}$. Consider the dynamic graph $\mathcal{G}$ with node set $V(G)=S \cup \{v\} \cup U \cup W$ of size $O(n)$. Here is construction sequence of $\mathcal{G}$:
    \begin{enumerate}
        \item Start from induced subgraph $H[S]=H-\{v,u,w\}$ with all other nodes isolated.
        \item For each $u_i \in U = \{u_1, u_2, \ldots u_n\}$, add edges to connect $u_i$ with $N_H(u)\subseteq S$.
        \item For each $w_j \in W = \{w_1, w_2, \ldots w_n\}$, add edges to connect $w_j$ with $N_H(w)\subseteq S$.
        \item Among all $n^2$ pairs $(u_i , w_j)$ with $u_i\in U$ and $w_j\in W$, choose a subset $C$ of $n^2/2$ pairs. For each chosen pair $(u_i,w_j) \in C$, add edge $\{u_i,w_j\}$.\label{step: C}
        \item Add edges to connect node $v$ with $N_H(v) \subseteq S$. We denote the resulting graph by $G_C$.
    \end{enumerate}   
    Each edge $\{u_i,w_j\}\in C$ can be seen as a duplicate of the edge $\{u,w\}$ in the sense that $N_{G_C}(u_i)=N_H(u)$ and $N_{G_C}(w_j)=N_H(w)$.
    Moreover, for each pair $(u_i,w_j)\in C$, its corresponding edge $\{u_i,w_j\}$ belongs to one unique copy of $H$ containing node $v$. We write such a copy of $H$ as $H_{i,j}$, then $\mathcal{H}_{C}=\{H_{i,j}: (u_i, w_j) \in C\}$ is the collection of all copies of $H$ in $G_C$ containing node $v$. To solve the membership-listing problem, node $v$ is required to output $\mathcal{H}_{C}$ after $r$ rounds of communication from the last topological change.

 There are $\binom{n^2}{n^2/2}$ distinct choices of $C$ at \Cref{step: C}, where each choice of $C$ corresponds to a distinct output $\mathcal{H}_{C}$ of node $v$.
    Since $v$ is not adjacent to each $u_i \in U$ and each $w_j \in W$, the output of $v$ only depends on the messages that $v$ receive from $N_H(v)$. Similar to the proof of \Cref{lem: memlist lb edge insertions other}, set $d=|N_H(v)| = O(1)$. The $d$ edges incident to $v$ are added in the last $d$ topological changes. Since each topological change is followed by $r$ rounds of communication, $v$ receives $x = dr + (d-1)r + \cdots + r= O(r)$ messages from its neighbors in total.
    To ensure the correctness of the output from $v$, we must have
    \[
    2^{x \cdot B} \geq \binom{n^2}{n^2/2} = 2^{\Omega(n^2)}, 
    \]
    where $B$ is the message size.
    Therefore, we have $B = \Omega(n^2/r)$.
    % Note that there are $\binom{n^2}{n^2/2}$ distinct choices of $C$ at \Cref{step: C}. Each $C$ corresponds to a distinct output $\mathcal{H}_{C}$ of node $v$.
    % Since $v$ is not adjacent to each $u_i \in U$ and each $w_j \in W$, the output of $v$ only depends on the received messages from $N_H(v)$. Set $d=|N_H(v)|$. In the last $d$ rounds, there are $d$ edges added to $v$ and $v$ receives $dr + (d-1)r + \cdots + r=\Theta(d^2)\cdot r$ messages from all its neighbors.
    % To ensure the correctness of the output from $v$, it requires
    % \[
    % 2^{\Theta(d^2)\cdot Br} \geq \binom{n^2}{n^2/2} = 2^{\Omega(n^2)} 
    % \]
    % Since $d$ is negligible compared with $n$, it has $B = \Omega(n^2/r)$ for any $r\geq 1$.
\end{proof}
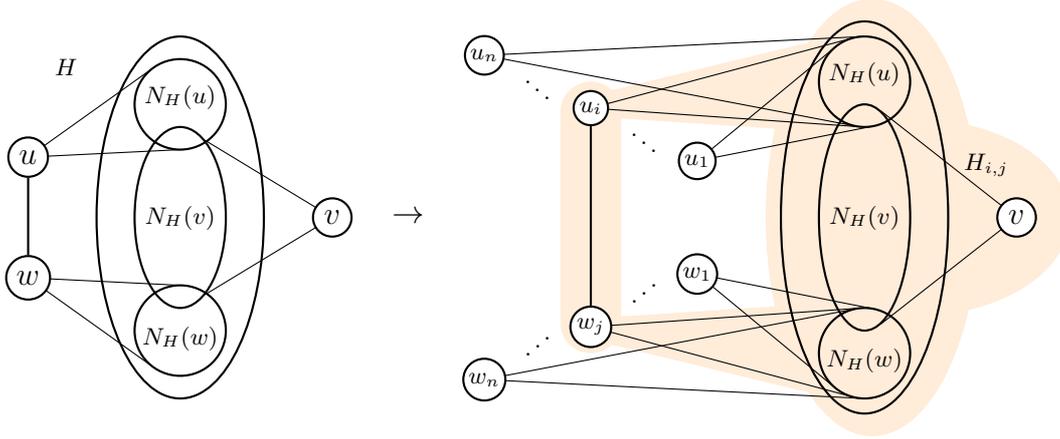
\begin{figure}[htbp]
    \centering
    \begin{tikzpicture}[node distance={10mm},font=\footnotesize,thick, 
vtx/.style = {draw, circle, inner sep=0.8mm, font=\large, fill=white}, 
dup/.style= {draw, circle, inner sep=0.4mm, font=\footnotesize, fill=white}, 
bg/.style={fill=orange!15, draw=none}, 
cone/.style={black,thin}, x=1mm, y=1mm
]
\tikzmath{\dx = 0.7;} 
    \begin{scope}
        \node at (-5,20) {$H$};
        \drawTangentLines{10}{15}{6}{-10}{8}{cone}
        \drawTangentLines{10}{-15}{6}{-10}{-8}{cone}
        \drawTangentLines{07}{0}{12}{30}{0}{cone}
        \draw (10,15) circle (6);
        \draw (10,-15) circle (6);
        \node at (10,16) {$N_H(u)$};
        \node at (10,-16) {$N_H(w)$};
        % \draw (1,0) circle (1);
        \draw (10,0) ellipse [x radius=06mm,y radius=12mm];
        \draw (10,0) ellipse [x radius=11mm,y radius=24mm];
        \node at (10,0) {$N_H(v)$};
        
        \node[vtx] (u) at (-10,8) {$u$};
        \node[vtx] (w) at (-10,-8) {$w$};
        \node[vtx] at (30,0) {$v$};
        \draw (u)--(w);
    \end{scope}

    \node[font=\large] at (40,0) {$\rightarrow$};

    \begin{scope}[xshift=90mm]
        \draw[bg] (11mm,0mm) ellipse (14mm and 29mm);
        \draw[bg, rounded corners=4mm] (-30,-18) rectangle (-22,18);

        % \filldraw[bg, rounded corners=20mm] 
        %     (0mm, 30mm) -- (0mm, -30mm) -- (45mm,0mm) -- cycle;
        \draw[bg] (17mm,0mm) ellipse (20mm and 13mm);
        \filldraw[bg, rounded corners=9mm] 
            (-35mm, 14mm) -- (10mm,25mm) -- (10mm,11mm) -- cycle;
        \filldraw[bg, rounded corners=9mm] 
            (-35mm, -14mm) -- (10mm,-25mm) -- (10mm,-11mm) -- cycle;
    
        \node at (26,7) {$H_{i,j}$};
        \drawTangentLines{1}{0}{18}{30}{0}{cone}
        \draw (10,18) circle (6);
        \draw (10,-18) circle (6);
        \node at (10,19) {$N_H(u)$};
        \node at (10,-19) {$N_H(w)$};
        % \draw (1,0) circle (1);
        \draw (10,0) ellipse [x radius=6mm,y radius=15mm];
        \draw (10,0) ellipse [x radius=11mm,y radius=26mm];
        \node at (10,0) {$N_H(v)$};

        \foreach \x/\y/\z in {50/n/n,30/i/j,10/1/1}{
            \drawTangentLines{10}{18}{6}{-5-\x*\dx}{4+\x*\dx*0.5}{cone}
            \drawTangentLines{10}{-18}{6}{-5-\x*\dx}{-4-\x*\dx*0.5}{cone}
            \node[dup] (u\x) at (-5-\x*\dx,4+\x*\dx*0.5) {$u_{\y}$};
            \node[dup] (w\x) at (-5-\x*\dx,-4-\x*\dx*0.5) {$w_{\z}$};
        }
        \node at (-33,18)    {$\ddots$};
        \node at (-33,-16) {\reflectbox{$\ddots$}};
        \node at (-19,11)  {$\ddots$};
        \node at (-19,-9)   {\reflectbox{$\ddots$}};
        
        \node[vtx] at (30,0) {$v$};
        \draw (u30)--(w30);
    \end{scope}
\end{tikzpicture}
    \caption{Left: graph $H$, node $v$, and edge $\{u,w\}$. Right: graph $G_C$, where the unique copy $H_{i,j}$ of $H$ that contains $\{u_i, w_j\}$ is highlighted.}
    %\caption{Graph $H$ is on the left with edge $\{u,w\}$ and node $v$ such that node $v$ is not adjacent to $u$ and $w$. Graph $G_C$ is on the right. It has node set $S\cup \{v\}\cup U \cup W$. It contains $n^2/2$ copies of $H$ which share node set $S \cup \{v\}$. For each pair $(u_i,w_j)\in C$, there exists one unique copy of $H$, $H_{i,j}$, which contains the edge $\{u_i,w_j\}$. For each subset $C$ of size $n^2/2$, it corresponds to a distinct graph $G_C$ and a distinct output $\mathcal{H}_C$ of node $v$.}
    \label{fig: memlist lb other}
\end{figure}

\subsection{Lower Bounds Under Node Insertions}\label{subsec: lb node insertion}
In this section, we establish the same $r$-round bandwidth complexity lower bounds in \Cref{subsec: lb edge insertion} under \emph{node insertions}. The proofs are very similar to the proofs in \Cref{subsec: lb edge insertion}. The modification required is to replace edge insertions with node insertions in the construction of the dynamic network. %Due to the modification, the condition for the lower bound $B=\Omega(n^2/r)$ changes from being a complete multipartite graph to being node-edge independent (\Cref{def:neindependent}).
%we offer a lower bound of $r$-round bandwidth complexity for $\MemList$ under node insertions, where $B=\Omega(n/r)$ for graphs except cliques, presented by \Cref{lem: memlist lb node insertions complete} and $B=\Omega(n^2/r)$ for graph except complete multipartite graphs, presented by \Cref{lem: memlist lb node insertions other}. The proofs are similar to the proofs used in \Cref{subsec: lb edge insertion}. Instead of edge insertions, the construction sequence uses node insertions.

\begin{theorem}\label{lem: memlist lb node insertions complete}
    For any graph $H$ that is not a clique and for any $r\geq 1$, 
    the $r$-round bandwidth complexity of $\MemList(H)$ under node insertions is $B=\Omega(n/r)$.
\end{theorem}
\begin{proof}
We use the same dynamic network construction in the proof of \Cref{lem: memlist lb edge insertions complete} shown in \Cref{fig:memlist-ind-pair}, replacing edge insertions with node insertions.
    %Suppose two nodes $u,v\in V(H)$ are not adjacent. Let $W=V(H)\setminus \{u,v\}$.
    %To get $B=\Omega(n/r)$, we construct graph $G_{U'}$ as illustrated in Figure \ref{fig:memlist-ind-pair} and it is constructed by node insertions with node $v$ last inserted. Set up $U = \{u_1, u_2, \ldots u_n\}$. Here is the insertion sequence of $\mathcal{G}$:
    \begin{enumerate}
        \item  Start from the induced subgraph $H[W]=H-\{u,v\}$.
        \item Pick $n/2$ nodes from $U = \{u_1, u_2, \ldots u_n\}$ to form node set $U'$. For each $u_i \in U'$, add node $u_i$ together with edges incident to $N_H(u)\subseteq W$.
        \item Add node $v$ together with edges incident to $N_H(v)\subseteq W$.  We denote the resulting graph by $G_{U'}$.
    \end{enumerate}   
    Since $v$ is inserted in the last topological change with $d = |N_H(v)|$ incident edges, $v$ receives $dr = O(r)$ messages before it needs to output all copies of $H$ in $G_{U'}$ containing $v$.
    Using the same analysis of \Cref{lem: memlist lb edge insertions complete}, for $v$ to produce the correct output, it is required that 
    \[
    2^{Bdr} \geq \binom{n}{n/2} = 2^{\Omega(n)},
    \]
    so $B = \Omega(n/r)$.
\end{proof}

\begin{theorem}\label{lem: memlist lb node insertions other}
For any graph $H$ that is not a complete multipartite graph and for any $r\geq 1$, 
    the $r$-round bandwidth complexity of $\MemList(H)$ under node insertions is $B=\Omega(n^2/r)$.
\end{theorem}
\begin{proof}
We use the same dynamic network construction in the proof of \Cref{lem: memlist lb edge insertions other} shown in \Cref{fig: memlist lb other}, replacing edge insertions with node insertions.
    %Suppose node $v\in V(H)$ and edge $\{u,w\}\in E(H)$ are independent.
    % satisfy that $\{v,u\} \notin E(H)$ and $\{v,w\} \notin E(H)$. 
    %Let $S=V(H)\setminus \{v,u,w\}$.
    %To get $B=\Omega(n^2/r)$, we construct graph $G_C$ as illustrated in \Cref{fig: memlist lb other} and it is constructed by node insertions with $v$ last inserted.
    %Set up $U = \{u_1, u_2, \ldots u_n\}$ and $W = \{w_1, w_2, \ldots w_n\}$. Let $C$ be a set of arbitrary $n^2/2$ pairs of $(u_i,w_j)$ with $u_i\in U, w_j\in W$. Here is the insertion sequence for $\mathcal{G}$:
        \begin{enumerate}
        \item Start from induced subgraph $H[S]=H-\{v,u,w\}$.
        \item For each $u_i \in U = \{u_1, u_2, \ldots u_n\}$, add node $u_i$ together with edges incident to $N_H(u)$.
            \item Among all $n^2$ pairs $(u_i , w_j)$ with $u_i\in U$ and $w_j\in W$, choose a subset $C$ of $n^2/2$ pairs. 
        \item For each $w_j \in W = \{w_1, w_2, \ldots w_n\}$, add node $w_j$ together with edges incident to $N_H(w)$. If $(u_i, w_j)\in C$, we also add the edge $\{u_i, w_j\}$.
        \item Add node $v$ together with edges incident to $N_H(v)$.  We denote the resulting graph by $G_{C}$.
    \end{enumerate}     
    Since $v$ is inserted in the last topological change with $d = |N_H(v)|$ incident edges, $v$ receives $dr = O(r)$ messages before it needs to output all copies of $H$ in $G_C$ containing $v$.
    Using the same analysis in the proof of \Cref{lem: memlist lb edge insertions other}, for $v$ to produce the correct output, it is required that 
    \[
    2^{Bdr} \geq \binom{n^2}{n^2/2} = 2^{\Omega(n^2)},
    \]
    so $B = \Omega(n^2/r)$.
\end{proof}

\subsection{Deletion Models}\label{subsec: deletion}

In this section, we finish the complexity landscape of the membership-listing problem by considering edge deletions and node deletions. For an arbitrary target subgraph $H$, we show $r$-round algorithms with bandwidth complexity $B=O(\log/r)$ for $\MemList(H)$ under edge insertions (\Cref{lem: memlist edge deletion}) and under node insertions (\Cref{lem: memlist node deletion}), and then we show a matching lower bound $\Omega((\log n)/r)$ for any target subgraph that is not a clique (\Cref{lem: memlist lb node deletion}).

\begin{theorem}\label{lem: memlist edge deletion}
    For any connected graph $H$ and for any $r\geq r_H$,  
    the $r$-round bandwidth of $\MemList(H)$ under edge deletions is $B=O((\log n)/r)$.
\end{theorem}

\begin{proof}
We begin with showing an $r_H$-round algorithm with bandwidth $O(\log n)$. Whenever an edge $\{u,v\}$ is deleted, both $u$ and $v$ send $\{\ID(u), \ID(v)\}$, which can be encoded as a message of $O(\log n)$ bits, to all their neighbors in the next round of communication. During the remaining $r_H - 1$ rounds of communication, all nodes that have received the message $\{\ID(u), \ID(v)\}$ broadcast the message to all their neighbors. 
For the correctness of the algorithm, let $G$ be the graph immediately before the deletion of $\{u,v\}$, and fix a subgraph $H$ of $G$ that contains the edge $\{u,v\}$. Consider any node $w \in V(H)$. Observe that $\min\{\dist_H(w,u), \dist_H(w,v)\}\leq r_H$, so within $r_H$ rounds of communication, $w$ receives $\{\ID(u), \ID(v)\}$, so $w$ has sufficient information to detect the deletion of $\{u,v\}$ so that $w$ can stop listing $H$.

%, node $u$ and $v$ can stop listing $H$ without communication needed. For each remaining node $w\in V(H)$, $w$ receives $\{\ID(u), \ID(v)\}$ within $r_H$ rounds of communication, since either $\dist_H(w,u)\leq r_H$ or $\dist_H(w,v)\leq r_H$. Once it receives $\{\ID(u), \ID(v)\}$, $w$ stops listing all copies of $H$ containing edge $\{u,v\}$.

To generalize the above algorithm to an $r$-round algorithm with bandwidth $O((\log n)/r)$, all we need to do is to break each round of the above algorithm into $d=\left\lfloor \frac{r}{r_H}\right\rfloor = \Theta(r)$ rounds by breaking the $O(\log n)$-bit message into $d$ blocks of equal size $O((\log n)/d) = O((\log n)/r)$ and send them using $d$ rounds.
%Consider the general $r$-round algorithm with any $r\geq r_H$. Set $d=\left\lfloor \frac{r}{r_H}\right\rfloor$. For each round of the above $r_H$-round algorithm, it can be simulated using $d$ rounds: break the edge $\ID$ into $d$ blocks equally and send it to all its neighbors, one block on every round. Therefore the $r$-round bandwidth of $\MemList(H)$ under edge insertions is $O((\log n)/r)$ for any $r\geq r_H$.
\end{proof}

\begin{theorem}\label{lem: memlist node deletion}
    For any connected graph $H$ and for any $r \geq r_H'$, 
    the $r$-round bandwidth of $\MemList(H)$ under node deletions is $B=O((\log n)/r)$.
\end{theorem}

\begin{proof}
The proof is similar to the proof of \Cref{lem: memlist edge deletion}. The main difference is that here we spread the $\ID$ of the deleted node and not the $\ID$ of the endpoints of the deleted edge.
%For node deletions, instead of spreading edge $\ID$ to all its neighbors, it can spread the deleted node $\ID$ directly. 
Again, we begin with describing an $r_H'$-round algorithm with bandwidth $O(\log n)$. Whenever a node $u$ is deleted, all $v\in N_G(u)$ send $\ID(u)$, which can be encoded as a message of $O(\log n)$ bits, to all their neighbors in the next communication round.
During the remaining $r_H' - 1$ rounds of communication, all nodes that have received the message $\ID(u)$ broadcast the message to all their neighbors. 
%and all received nodes forward it to all its neighbors during $r_H'$ rounds of communication. It requires $B=O(\log n)$.
For the correctness of the algorithm, let $G$ be the graph immediately before the deletion of $u$, and fix a subgraph $H$ of $G$ that contains $u$. Consider any node $w \in V(H)$. Observe that $\dist_H(w,v)\leq r_H'$ for some $v\in N_H(u)$, so within $r_H'$ rounds of communication, $w$ receives $\ID(u)$, so $w$ has sufficient information to detect the deletion of $u$ so that $w$ can stop listing $H$. 
%For each subgraph $H$ in $G$ containing node $u$ and each remaining node $w\in V(H)$, $w$ receives $\ID(u)$ within $r_H'$ rounds of communication, since $\dist_H(w,v)\leq r_H'$ for some $v\in N_H(u)$. Once it receives $\ID(u)$, $w$ stops listing all copies of $H$ containing node $u$. 
Similar to the proof of \Cref{lem: memlist edge deletion}, the above algorithm can be generalized to an $r$-round algorithm with bandwidth $O((\log n)/r)$.
%Given any $r\geq r_H'$, we can use $r$-round algorithm to simulate above $r_H'$-round algorithm with bandwidth $B=O((\log n)/r)$.
\end{proof}

\begin{theorem}\label{lem: memlist lb edge deletion}\label{lem: memlist lb node deletion}
    For any connected graph $H$ that is not a clique and for any $r \geq 1$,
    the $r$-round bandwidth of $\MemList(H)$ under edge deletions and node deletions is $B=\Omega((\log n)/r)$.
\end{theorem}
\begin{proof}
    Since $H$ is not a clique, there exist $u,v\in V(H)$ such that $\{u,v\}\notin E(H)$.     
    To prove the desired bandwidth complexity lower bound $B=\Omega((\log n)/r)$, we consider a dynamic network starting from a graph $G^*$ containing $\Omega(n)$ copies of $H$ such that they have \emph{distinct} nodes $u$ and a \emph{common} node $v$. Such a graph $G^\ast$ can be constructed by simply replacing $u$ in $H$ with an independent set $U=\{u_1, u_2, \ldots, u_n\}$, so $V(G^*)=V(H)\cup U \setminus \{u\}$, and for each $u_i\in U$, $\{u_i, w\}\in E(G^*)$ if and only if $\{u, w\} \in E(H)$ for some node $w\in V(H)$. 
 Therefore, each $u_i\in U$ belongs to one unique copy of $H$ in $G^*$, which contains node $v$. Denote that copy as $H_i$. Set $\mathcal{H} = \{H_1, H_2, \ldots, H_n\}$. Consider the following dynamic graph $\mathcal{G}$:
    \begin{enumerate}
        \item Start from the graph $G^*$.
        \item Pick an arbitrary node $u_i \in U$. 
        \begin{itemize}
            \item For the case of edge deletion, pick any node $w\in N_H(u)$ and delete the edge $\{u_i,w\}$.
            \item For the case of node deletion, delete the node $u_i$.
        \end{itemize}
        In either case, $v$ should stop listing $H_i$ after $r$ rounds of communication.
    \end{enumerate}
    Observe that there are $n$ distinct choices of $u_i$ from $U$. Each choice of $u_i\in U$ corresponds to a distinct output $\mathcal{H}\setminus \{H_i\}$ of node $v$. Since $v$ is not adjacent to any $u_j \in U$, the output of $v$ only depends on the messages that $v$ receive from $N_H(v)$. Set $d=|N_H(v)|$. 
    In $r$ rounds of communication, node $v$ receives $dr = O(r)$ messages, each with $B$ bits.
    To ensure the correctness of the output from $v$, it is required that
    \[
    2^{Bdr} \geq n,
    \]
    so $B = \Omega((\log n)/r)$.% for $\MemList(H)$ under edge deletions and node deletions respectively for any $r\geq 1$. 
\end{proof}

% \begin{figure}[htbp]
%     \centering
%     \includegraphics[scale = 0.2]{img/MemList-Deletion.jpg}
%     \caption{The labelings are wrong}
%     \label{fig:Memlist Deletion}
% \end{figure}

%\newpage
\section{One-Round Membership-Detection}\label{sect:memDect}
In this section, we investigate the one-round bandwidth complexity of the $\MemDetect(H)$ problem under any single type of topological change. A summary of our results is provided in \Cref{tab:memDect}. For the definitions of $\rad$, $\diam$, and $\nediam$, see \Cref{def:parameters,def:NE_parameters}. 
%\yijun{Some issues: How do we get rid of $\nediam>3$ in the case of $\rad=2$? Maybe just write $\nediam(H) \geq 3$ for the sake of simplicity? Similarly, how do I know that for $\diam=2$, we must have $\rad \in \{1,2\}$? This has to be discussed in the paper somewhere. Some lower bounds are written in terms of $r_H$ and $r_H'$ but the table is written in terms of other graph parameters -- need to cite the observations that connect these terminologies.}
%\yanyu{relationship between $\diam$ and $\nediam$ is addressed in \Cref{obs: diam nediam}; $\rad<=\diam$ should be textbook/ direct from definition. }
%\yanyu{edited impossibly column with additional reference for why complete multipartite}
%In this section, we provide a partial characterization of the bandwidth complexity of one-round dynamic \MemDetect{} for any graph $H$ under any single type of topological change. See \Cref{tab:memDect} for a summary of our results. \mingyang{add an overview}

\begin{table}[ht!]
\begin{center}
\begin{tabular}{lllllll}
\cline{2-7}
\multicolumn{1}{l|}{\textbf{}} & \multicolumn{5}{l|}{Complete multipartite graphs} & \multicolumn{1}{l|}{\multirow{3}{*}{Others}} \\ \cline{2-6}
\multicolumn{1}{l|}{\textbf{}} & \multicolumn{1}{l|}{\multirow{2}{*}{Stars}} & \multicolumn{3}{l|}{$s$-cliques} & \multicolumn{1}{l|}{\multirow{2}{*}{Others}} & \multicolumn{1}{l|}{} \\ \cline{3-4}
\multicolumn{1}{l|}{\textbf{}} & \multicolumn{1}{l|}{} & \multicolumn{1}{l|}{$s=3$} & \multicolumn{1}{l|}{$s \geq 3$} & \multicolumn{1}{l|}{} & \multicolumn{1}{l|}{} & \multicolumn{1}{l|}{} \\ \cline{2-7} 
\multicolumn{1}{l|}{\multirow{2}{*}{\textbf{Edge insertions}}} & \multicolumn{1}{l|}{$\Theta(1)$} & \multicolumn{1}{l|}{$O(\log n)$} & \multicolumn{1}{l|}{$O(\sqrt{n})$} & \multicolumn{1}{l|}{$\Omega(\log\log n)$} & \multicolumn{1}{l|}{$\Theta(n)$} & \multicolumn{1}{l|}{Impossible} \\
\multicolumn{1}{l|}{} & \multicolumn{1}{l|}{[\ref{Lem:Memdect star}]} & \multicolumn{1}{l|}{\cite{bonne2019distributed}} & \multicolumn{1}{l|}{\cite{bonne2019distributed}} & \multicolumn{1}{l|}{[\ref{lem: memdect lb tri}]} & \multicolumn{1}{l|}{[\ref{Lem: memlist Multipartite graph}][\ref{lem: memdect lb edge insertion}]} & \multicolumn{1}{l|}{[\ref{thm: nediam=2}]} \\ \cline{3-6}
\multicolumn{1}{l|}{\multirow{2}{*}{\textbf{Node insertions}}} & \multicolumn{1}{l|}{[\ref{memdetect_trivial_LB1}]} & \multicolumn{3}{l|}{$\Theta(n)$} & \multicolumn{1}{l|}{$\Theta(n)$} & \multicolumn{1}{l|}{[\ref{lem: locality edge}]} \\
\multicolumn{1}{l|}{} & \multicolumn{1}{l|}{[\ref{memdetect_trivial_LB2}]} & \multicolumn{3}{l|}{\cite{bonne2019distributed}} & \multicolumn{1}{l|}{[\ref{Lem: memlist Multipartite graph}][\ref{lem: memdect lb node insertion}]} & \multicolumn{1}{l|}{[\ref{lem: locality node insertion}]} \\ \cline{3-6}
\multicolumn{1}{l|}{\multirow{2}{*}{\textbf{Edge deletions}}} & \multicolumn{1}{l|}{} & \multicolumn{3}{l|}{$\Theta(1)$} & \multicolumn{1}{l|}{$O(\log n)$} & \multicolumn{1}{l|}{} \\
\multicolumn{1}{l|}{} & \multicolumn{1}{l|}{} & \multicolumn{3}{l|}{\cite{bonne2019distributed}} & \multicolumn{1}{l|}{[\ref{lem: memlist edge deletion}]} & \multicolumn{1}{l|}{} \\ \cline{2-7} 
\textbf{} &  &  &  &  &  &  \\ \cline{2-7} 
\multicolumn{1}{l|}{\textbf{}} & \multicolumn{1}{l|}{\multirow{3}{*}{$\diam = 1$}} & \multicolumn{4}{l|}{$\diam = 2$} & \multicolumn{1}{l|}{\multirow{3}{*}{$\diam \geq 3$}} \\ \cline{3-6}
\multicolumn{1}{l|}{\textbf{}} & \multicolumn{1}{l|}{} & \multicolumn{2}{l|}{\multirow{2}{*}{$\rad = 1$}} & \multicolumn{2}{l|}{$\rad = 2$} & \multicolumn{1}{l|}{} \\ \cline{5-6}
\multicolumn{1}{l|}{\textbf{}} & \multicolumn{1}{l|}{} & \multicolumn{2}{l|}{} & \multicolumn{1}{l|}{$\nediam = 2$} & \multicolumn{1}{l|}{$\nediam = 3$} & \multicolumn{1}{l|}{} \\ \cline{2-7} 
\multicolumn{1}{l|}{\multirow{2}{*}{\textbf{Node deletions}}} & \multicolumn{1}{l|}{$0$} & \multicolumn{2}{l|}{$\Theta(1)$} & \multicolumn{1}{l|}{$\Theta(1)$} & \multicolumn{1}{l|}{$O(\log n)$} & \multicolumn{1}{l|}{Impossible} \\
\multicolumn{1}{l|}{} & \multicolumn{1}{l|}{\cite{bonne2019distributed}} & \multicolumn{2}{l|}{[\ref{memdetect_trivial_LB2}][\ref{thm: memdect node deletion r=1}]} & \multicolumn{1}{l|}{[\ref{memdetect_trivial_LB2}][\ref{thm: memdect node deletion M}]} & \multicolumn{1}{l|}{[\ref{lem: memlist node deletion}]} & \multicolumn{1}{l|}{[\ref{thm: diam geq 3 impossible memdetect}]} \\ \cline{2-7} 
\end{tabular}
\end{center}
    \caption{The bandwidth complexity of one-round $\MemDetect(H)$.}
    \label{tab:memDect}
\end{table}

To see that the table for node deletions covers all cases, observe that $\diam=2$ implies $\rad \in \{1,2\}$ because $\rad \leq \diam$, and from \Cref{obs: diam nediam} we infer that $\diam=2$ implies $\nediam \in \{2,3\}$. In the table for the remaining three types of topological changes, refer to the paragraph at the end of \Cref{subsec: local constraints} for why complete multipartite graphs characterize one-round solvability.

%In the table, the class of graphs $H$ with $\diam(H) = \rad(H) = \nediam(H) = 2$ corresponds to the ``proper'' %\yijun{Now, I prefer ``proper'' over ``rich''}
%complete multipartite graphs, which are complete multipartite graphs where each part contains at least two nodes. 
For node deletions, our $O(1)$-bandwidth algorithm for the case where $\diam(H) = \rad(H) = \nediam(H) = 2$ works for all complete multipartite graphs $H$. Recall from \Cref{thm: nediam=2} that a connected graph $H$ with at least three nodes is complete multipartite if and only if $\nediam(H) = 2$.
For node deletions, the only case where we are unable to obtain a tight bound is when $\diam(H) = \rad(H) = 2$ and $\nediam(H) = 3$.
A notable example of such a graph $H$ is the $5$-cycle $C_5$.

In \Cref{sect:easy}, we collect some easy observations about the membership-detection problem. In \Cref{sect:memdect_LB}, we show our new bandwidth complexity lower bounds.
In \Cref{sec:1r-memd-upp-nd}, we show our new bandwidth complexity upper bounds.
%In \Cref{sect:memdect_open}, we discuss the remaining open problems.

\subsection{Easy Observations}\label{sect:easy}

In this section, we collect some easy observations about the one-round bandwidth complexity of the membership-detection problem.

\begin{observation}\label{Lem:Memdect star}
    For any star graph $H=K_{1,s}$ with $s \geq 2$, the one-round bandwidth complexity of $\MemDetect(H)$ is $O(1)$ under any type of topological change: edge insertions, node insertions, edge deletions, and node deletions.
\end{observation}

\begin{proof}
    Let $H=K_{1,s}$. 
    In each round, after the topological change, each node sends a one-bit message $b$ to all its neighbors. The message $b$ is an indicator with value $1$ if its degree is at least $s$ and value $0$ otherwise. For each node $u$, if its current degree is less than $s$ and all the received messages in the current round are $0$, it outputs \No{}.
    Otherwise, $u$ outputs \Yes{}. This algorithm works under any type of topological change.%\yijun{lower bound missing}
    %\yijun{Shall we also have a theorem for node deletion and $\diam=1$?}
    %\mingyang{$diam=1$ is a clique. we can use result from \cite{bonne2019distributed}}
\end{proof}

\begin{observation} \label{thm: diam geq 3 impossible memdetect}
    For any target subgraph $H$ with $\diam(H) \geq 3$, there is no one-round algorithm for $\MemDetect(H)$ under node deletions.
\end{observation}

\begin{proof}
    % This follows directly from \Cref{obs: r_H'} and \Cref{lem: locality node deletion}.
    This observation follows directly from \Cref{def: r_H'} and \Cref{lem: locality node deletion} with $T = 1$. %\yanyu{old equivalent between contraction and node-edge distance not needed now.}
\end{proof}

\begin{observation} \label{memdetect_trivial_LB1}
    If $|E(H)| \geq 2$, then any algorithm for $\MemDetect(H)$ under edge insertions or edge deletions must have non-zero bandwidth complexity.
\end{observation}
\begin{proof}
    If $|E(H)| \geq 2$, then there exist a node $u$ and an edge $e$ in $H$ such that $u$ is not an endpoint of $e$. Without communication, $u$ cannot detect the deletion or insertion of $e$. 
\end{proof}

\begin{observation}\label{memdetect_trivial_LB2}
    If $\diam(H) \geq 2$, then any algorithm for $\MemDetect(H)$ under node insertions or node deletions must have non-zero bandwidth complexity.
\end{observation}
\begin{proof}
    If $\diam(H) \geq 2$, then there exists an independent set $S=\{u,v\}$ of two nodes in $H$. Without communication, $u$ cannot detect the deletion or insertion of $v$. 
\end{proof}

%\yijun{Can someone clean up this section? Remove the results that we cannot prove, have a discussion about the remaining open problems, and fix all the issues to make the writing of this section ``complete.''}

\subsection{\texorpdfstring{$\Omega(n)$}{Omega(n)} Lower Bounds Under Node and Edge Insertions}\label{sect:memdect_LB}

In this section, we show that, for both node insertions and edge insertions, the one-round bandwidth complexity of $\MemDetect(H)$ is $\Omega(n)$ for any complete multipartite graph $H$ that is neither a star nor a clique.

%In this section, we offer a lower bound of bandwidth complexity for one-round $\MemDetect(H)$ under edge insertions, where $B=\Omega(n)$ bits for any complete multipartite graph except cliques or stars.
\paragraph{Proof idea} Before presenting formal proofs, we first describe the high-level ideas behind the proofs. For any given graph $H$ that is not a \emph{clique}, we can find two non-adjacent nodes $u$ and $v$ with a common neighbor $w$. Consider the network $G$ resulting from removing the edge $\{u,w\}$ in $H$ and replacing the node $u$ with an independent set $U'$ of size $\Theta(n)$.
The network $G$ is constructed by a series of edge insertions, where the edges incident to $v$ and $w$ are added \emph{last}. This ensures that $v$ and $w$ have only $O(1)$ rounds to learn the set $U'$.

Next, we insert an edge $\{u_j, w\}$ to the network $G$ for some choice of $u_j \notin V(H) \setminus \{u\}$. If $u_j \in U'$, then $\{u_j\} \cup V(H)\setminus \{u\}$ induces a copy of $H$, so $v$ should output $\Yes$. If $u_j \notin U'$, then no copy of $H$ is formed, so $v$ should output $\No$. Immediately after the insertion of $\{u_j, w\}$, $w$ can inform $v$ of $\ID(u_j)$. For $v$ to output correctly for all possible choices of $\ID(u_j)$, $v$ must learn the set $U'$, which requires $\Omega(n)$ bits of information. As $v$ is incident to $O(1)$ edges only, we obtain a one-round bandwidth complexity lower bound $\Omega(n)$ under edge insertions. The proof for node deletions is similar.

To realize the proof idea, the selection of $u$ and $v$ needs to be done carefully. In particular, it is required that $N_H(u) \setminus \{w\} \neq \emptyset$, since otherwise removing $\{u,w\}$ from $H$ renders $u$ an isolated node. This explains why $H$ cannot be a \emph{star}.

%Given graph $H$ satisfying above constraints, we can find two non-adjacent nodes $u,v$ with their common neighbor $w$. To get lower bound $\Omega(n)$, we start from the node set $V(H)$ and duplicate $u$ into $U'$, an arbitrary node set of size $\Omega(n)$. For each node $u_i\in U'$, we construct a unique copy of $H$ with the edge $\{w, u_i\}$ missing. Finally, pick one node $u_j$ and insert the edge $\{w, u_j\}$. The bandwidth constraint implies that $v$ cannot distinguish between two choices of $U'$, so once we complete the construction with a certain choice  $u_j \in U'$, $v$ cannot tell if this $u_j$ is really in $U'$ because there is another choice of $U'$ not containing $u_j$ that $v$ cannot distinguish.\yijun{need to discuss the need for the assumption that the graph is not a star}

\begin{theorem}\label{lem: memdect lb edge insertion} For any complete multipartite graph $H$ that is neither a clique nor a star, the one-round bandwidth complexity of $\MemDetect(H)$ under edge insertions is $\Omega(n)$.
\end{theorem}
\begin{proof}
Let $S \subseteq V(H)$ be a \emph{largest} part of the complete multipartite graph $H$. Since $H$ is not a clique, $|S|\geq 2$. Select $u$ and $v$ to be any two distinct nodes in $S$. Observe that $u$ and $v$ are non-adjacent, as $S$ is an independent set. Select $w$ to be any node in $N_H(u)$. Since $H$ is not a star, $N_H(u) = N_H(v) = V(H) \setminus S$ contains at least two nodes, so $N_H(u) \setminus \{w\} \neq \emptyset$.

\paragraph{Dynamic network} Consider the dynamic network $\mathcal{G}$ defined by the following construction:
    \begin{enumerate}
        \item The initial graph $G^0$ is the subgraph of $H$ induced by $V(H) \setminus \{u,v,w\}$ together with a set of $n$ isolated nodes $U=\{u_1,u_2,\ldots, u_n\}$ and two isolated nodes $v$ and $w$.
        \item Select $U' \subseteq U$ as an arbitrary subset of size exactly $n/2$. 
        
        For each $u_i\in U'$, insert an edge between $u_i$ and each node in $N_H(u) \setminus \{w\}$. \label{step: set up U'}
        \item Insert an edge between $v$ and each node in $N_H(v) \setminus \{w\}$. \label{step: node v added}
        \item Insert an edge between $w$ and each node in $N_H(w)\setminus\{u\}$. \label{step: w edges}
        \item Pick an arbitrary node $u_j\in U$. 
        
        Insert an edge $\{w, u_j\}$.\label{step: node w added} %\yanyu{ is the one edge excepted or connected, seems to be the latter.
        %Edited, please verify}
    \end{enumerate}
If a step involves more than one edge insertion, then these edge insertions can be done sequentially in any order. We emphasize that our choice of $\{u,v,w\}$ ensures that $N_H(u) \setminus \{w\} \neq \emptyset$, so \Cref{step: set up U'} is not vacuous. The network right before \Cref{step: node w added} consists of a set of isolated nodes $U\setminus U'$ and the graph $G$ resulting from removing the edge $\{u,w\}$ in $H$ and replacing the node $u$ with an independent set $U'$. 

\paragraph{Output} We make the following observations about the output of $v$.
\begin{itemize}
    \item If $u_j \in U'$, then $v$ belongs to a subgraph isomorphic to $H$ after the insertion of $\{w, u_j\}$ in \Cref{step: node w added}, so $v$ should output $\Yes$.
    \item If $u_j \notin U'$, then $v$ does not belongs to a subgraph isomorphic to $H$ after the insertion of $\{w, u_j\}$ in \Cref{step: node w added}, so $v$ should output $\No$.
\end{itemize}
If $u_j \in U'$, then indeed $\{u_j\} \cup V(H)\setminus \{u\}$ induces a subgraph isomorphic to $H$. Now consider the case $u_j \notin U'$ and suppose that $v$ belongs to a subgraph $H'$ isomorphic to $H$ after the insertion of $\{w, u_j\}$ in \Cref{step: node w added}. Recall that $v \in S$, where $S$ is chosen as a \emph{largest} part of the complete multipartite graph $H$. Since $U'\cup S \setminus\{u\}$ is an independent set, at most $|S|$ nodes from $U'\cup S \setminus\{u\}$ can be included in $H'$. To ensure that $|V(H')| = |V(H)|$, the following two statements hold.
\begin{itemize}
    \item $H'$ contains all nodes in $V(H) \setminus S$. In particular, $w$ is included in $H'$.
    \item $H'$ contains exactly $|S|$ nodes in $U'\cup S \setminus\{u\}$. In particular, at least one node $u_i \in U'$ is included in $H'$.
\end{itemize}
Since both $\{u_i, v\}$ and $\{u_i, w\}$ are non-edges, $u_i$, $v$, and $w$ belong to the same part of the complete multipartite graph $H'$, contradicting the fact that $\{v,w\}$ is an edge. Therefore, such a subgraph $H'$ does not exist.

\paragraph{Information} The output of $v$ at the end of \Cref{step: node w added} depends on the messages that $v$ receive in \Cref{step: node v added,step: w edges,step: node w added}. We divide them into two parts.
\begin{itemize}
    \item The first part consists of all the messages that $v$ receives in \Cref{step: node v added,step: w edges,step: node w added}, excluding the message sent from $w$ to $v$ in \Cref{step: node w added}. We write $\mathcal{M}_1$ to denote the collection of these messages. Since $u_j$ is not adjacent to $v$, $\mathcal{M}_1$ is {independent} of the choice of $u_j$ and only depends on the choice of $U'$.
    \item The second part is the message sent from $w$ to $v$ in \Cref{step: node w added}. This message depends on $\ID(u_j)$ and all the messages that $w$ receives in \Cref{step: w edges}. We write $\mathcal{M}_2$ to denote the collection of the messages that $w$ receives in \Cref{step: w edges}. Observe that $\mathcal{M}_2$ is {independent} of the choice of $u_j$ and only depends on the choice of $U'$.
\end{itemize}
To summarize, the output of $v$ at the end of \Cref{step: node w added} is a function of $\mathcal{M}_1$, $\mathcal{M}_2$, and $\ID(u_j)$, where $\mathcal{M}_1$ and $\mathcal{M}_2$ are {independent} of the choice of $u_j$ and only depends on the choice of $U'$. Since $|N_H(v)|$ and $|N_H(w)|$ are both $O(1)$, $\mathcal{M}_1$ and $\mathcal{M}_2$ consist of $O(1)$ messages of $B$ bits, where $B$ is the bandwidth of the algorithm. Intuitively, this means that $v$ decides its output according to $\ID(u_j)$ and $O(B)$ bits of information extracted from the choice of $U'$.

\paragraph{Bandwidth complexity} There are $\binom{n}{n/2} = 2^{\Theta(n)}$ choices of  $U'\subseteq U$. Suppose the bandwidth complexity is $B = o(n)$, then there exist two distinct choices $U_1'$ and $U_2'$ of $U'\subseteq U$ that lead to identical $\mathcal{M}_1$ and $\mathcal{M}_2$. Select $u_p$ as any node in $U_1' \setminus U_2'$. Consider the following two scenarios:
    \begin{itemize}
        \item Pick $U' = U_1'$ in \Cref{step: set up U'} and pick $u_j = u_p$ in \Cref{step: node w added}. Since $u_p \in U_1'$, $v$ should output $\Yes$.
        \item Pick $U' = U_2'$ in \Cref{step: set up U'} and pick $u_j = u_p$ in \Cref{step: node w added}. Since $u_p \notin U_2'$, $v$ should output $\No$.
    \end{itemize}
    However, the output of $v$ is identical in both scenarios, as  $v$ receives the same messages in both scenarios, so the algorithm is incorrect. Therefore, we must have $B = \Omega(n)$.
\end{proof}

\begin{theorem}\label{lem: memdect lb node insertion}
    For any complete multipartite graph $H$ that is neither a clique nor a star, the one-round bandwidth complexity of $\MemDetect(H)$ under node insertions is $\Omega(n)$.
\end{theorem}
\begin{proof}
The proof is similar to the proof of \Cref{lem: memdect lb edge insertion}, with the same choice of $u$, $v$, and $w$. The main difference is that here the dynamic network is constructed using node insertions and not edge insertions.
%Let $S \subseteq V(H)$ be a \emph{largest} part of the complete multipartite graph $H$. Since $H$ is not a clique, $|S|\geq 2$. Select $u$ and $v$ to be any two distinct nodes in $S$. Observe that $u$ and $v$ are non-adjacent, as $S$ is an independent set. Select $w$ to be any node in $N_H(u)$. Since $H$ is not a star, $N_H(u) = N_H(v) = V(H) \setminus S$ contains at least two nodes, so $N_H(u) \setminus \{w\} \neq \emptyset$.

\paragraph{Dynamic network} Consider the dynamic network $\mathcal{G}$ defined by the following construction:
    \begin{enumerate}
        \item The initial graph $G^0$ is the subgraph of $H$ induced by $V(H) \setminus \{u,v,w\}$. 
        \item Let $U=\{u_1,u_2,\ldots, u_n\}$ be a set of $n$ nodes that are currently not in $G^0$. 
        
        Select $U' \subseteq U$ as an arbitrary subset of size exactly $n/2$. 
        \begin{itemize}
            \item For each $u_i\in U'$, insert the node $u_i$ together with edges incident to  $N_H(u) \setminus \{w\}$.
            \item For each $u_i\in U \setminus U'$, insert the node $u_i$ without any incident edges.
        \end{itemize}
         \label{step: set up U'2}
        \item Insert the node $v$ together with edges incident to $N_H(v) \setminus \{w\}$. \label{step: node v added2}
        \item Pick an arbitrary node $u_j\in U$. 
        
        Insert the node $w$ together with edges incident to $\{u_j\} \cup N_H(w)\setminus\{u\}$. \label{step: w edges2}
        %\item Pick an arbitrary node $u_j\in U$. Insert an edge $\{w, u_j\}$.\label{step: node w added} %\yanyu{ is the one edge excepted or connected, seems to be the latter.
        %Edited, please verify}
    \end{enumerate}
If a step involves more than one node insertion, then these node insertions can be done sequentially in any order.  

\paragraph{Output} Same as the proof of \Cref{lem: memdect lb edge insertion}, we have the following observations.
\begin{itemize}
    \item If $u_j \in U'$, then $v$ belongs to a subgraph isomorphic to $H$ after the insertion of $w$ in \Cref{step: w edges2}, so $v$ should output $\Yes$.
    \item If $u_j \notin U'$, then $v$ does not belongs to a subgraph isomorphic to $H$ after the insertion of $w$ in \Cref{step: w edges2}, so $v$ should output $\No$.
\end{itemize}

\paragraph{Information} The output of $v$ at the end of \Cref{step: w edges2} depends on the messages that $v$ receives in \Cref{step: node v added2,step: w edges2}. We divide them into two parts.
\begin{itemize}
    \item The first part consists of all the messages that $v$ receives in \Cref{step: node v added2,step: w edges2}, excluding the message sent from $w$ to $v$ in \Cref{step: w edges2}. We write $\mathcal{M}$ to denote the collection of these messages. Since $u_j$ is not adjacent to $v$, $\mathcal{M}$ is {independent} of the choice of $u_j$ and only depends on the choice of $U'$.
    \item The second part is the message sent from $w$ to $v$ in \Cref{step: w edges2}. This message is {independent} of the choice of $U'$ and only depends on the choice of $u_j$.
\end{itemize}
To summarize, the output of $v$ at the end of \Cref{step: w edges2} is a function of $\mathcal{M}$ and $\ID(u_j)$. Since $|N_H(v)| = O(1)$, $\mathcal{M}$ consists of $O(1)$ messages of $B$ bits, where $B$ is the bandwidth of the algorithm. Same as the proof of \Cref{lem: memdect lb edge insertion}, this means that $v$ decides its output according to $\ID(u_j)$ and $O(B)$ bits of information extracted from the choice of $U'$.

\paragraph{Bandwidth complexity} The rest of the proof is the same as the proof of \Cref{lem: memdect lb edge insertion}. There are $\binom{n}{n/2} = 2^{\Theta(n)}$ choices of $U'\subseteq U$. Suppose the bandwidth complexity is $B = o(n)$, then there exist two distinct choices $U_1'$ and $U_2'$ of $U'\subseteq U$ that lead to identical $\mathcal{M}$. Select $u_p$ as any node in $U_1' \setminus U_2'$. Consider the following two scenarios:
    \begin{itemize}
        \item Pick $U' = U_1'$ in \Cref{step: set up U'2} and pick $u_j = u_p$ in \Cref{step: w edges2}. Since $u_p \in U_1'$, $v$ should output $\Yes$.
        \item Pick $U' = U_2'$ in \Cref{step: set up U'2} and pick $u_j = u_p$ in \Cref{step: w edges2}. Since $u_p \notin U_2'$, $v$ should output $\No$.
    \end{itemize}
    However, the output of $v$ is identical in both scenarios, as  $v$ receives the same messages in both scenarios, so the algorithm is incorrect. Therefore, we must have $B = \Omega(n)$.  
\end{proof}

\subsection{\texorpdfstring{$O(1)$}{O(1)} Upper Bounds Under Node Deletions}
\label{sec:1r-memd-upp-nd}

In this section, we show two new $O(1)$ upper bounds on the one-round bandwidth complexity of $\MemDetect(H)$ under node deletions.

%In this section, we characterize the bandwidth complexity of one-round algorithms for $\MemDetect(H)$ under node deletions.

\begin{theorem}
\label{thm: memdect node deletion r=1}
    For any target subgraph $H$ satisfying $\diam(H)=2$ and $\rad(H)=1$, there is a one-round and $O(1)$-bandwidth algorithm for $\MemDetect(H)$ under node deletions.
\end{theorem}
\begin{proof}
    Since $\rad(H)=1$, there is a node in any copy of $H$ that is adjacent to every other node in the same copy of $H$. 
    We pick one such node for each copy and call this node the central node and other nodes the fringe nodes. 
    Initially, all nodes list every copy of $H$ they belong to, and output \Yes{} if the number of copies is at least one. 
    In subsequent rounds, using one-bit messages, the central node $u$ in every copy informs each fringe node $v$ once it detects that all copies of $H$ with $u$ as the central node and $v$ as a fringe node have been destroyed due to the deletion of a fringe node. When a central node $u$ of a copy is removed, the fringe nodes remove all the copies of $H$ such that $u$ is the central node in that copy.
    Each node outputs \Yes{} if and only if it is still in some copy of $H$, either as a central node or as a fringe node.
    The bandwidth complexity is $O(1)$ since the algorithm only requires sending one-bit messages. 
\end{proof}

Next, we present our main contribution of this section: a one-round and $O(1)$-bandwidth algorithm for $\MemDetect(H)$ that works for any complete multipartite graph $H$.

\paragraph{Warm up} Before presenting the algorithm for an arbitrary complete multipartite graph $H$, we warm up with the special case $H=C_4$. 
For notational simplicity, in the subsequent discussion, we write $N^r(u)=N_{G^r}(u)$ to denote the neighborhood of $u$ in $G^r$, where $G^r$ is the network at the end of the $r$th round. % and define the degree of $u$ after round $r$ as $d_u^r = |N^r(u)|$.
% \yijun{This notation $N^r(u)$ is too specialized and is only used in Section 6, so I moved it here. As far as I check, the notation  $d_u^r$ is not used, so I commented it out for now.} 

For a node $w$, two distinct neighbors $u\in N^r(w)$ and $v\in N^r(w)$, and a round number $r$, let $\Count_w^r(u,v)$ denote the number of nodes in $V(G^r) \setminus \{w\}$ that are adjacent to both $u$ and $v$. Observe that $w$ is contained in a $C_4$ in $G^r$ if and only if there exist two distinct neighbors $u\in N^r(w)$ and $v\in N^r(w)$ such that $\Count_w^r(u,v) \geq 1$. Therefore, to design a one-round algorithm to solve $\MemDetect(C_4)$, we just need to make sure that at the end of each round $r$, each node $w$ learns $\Count_w^r(u,v)$ for all $u\in N^r(w)$ and $v\in N^r(w)$.

We show that this can be done using one-bit messages, so the one-round bandwidth complexity of $\MemDetect(C_4)$ is $O(1)$.
Initially, since $G^0$ is known to all nodes, each node $w$ can locally compute $\Count_w^0(u,v)$ for all $u\in N^0(w)$ and $v\in N^0(w)$.
In each round, if a node detects that one of its neighbors is deleted, it sends a signal \Del{} to all its neighbors. 
If a node $w$ receives \Del{} from two neighbors $u$ and $v$ in round $r$, then $w$ updates $\Count_w^r(u,v) = \Count_w^{r-1}(u,v)-1$, as this indicates that a common neighbor of $u$ and $v$ in $V(G^{r-1}) \setminus \{w\}$ is deleted in round $r$. Otherwise, we have $\Count_w^r(u,v) = \Count_w^{r-1}(u,v)$. 

%At the end of each round, $w$ outputs \Yes{} if it has at least 1 pair of neighbors $u,v$ such that $C_w^r(u,v) > 0$, and \No{} otherwise. 
% \yanyu{updated warm up with short proof}

% This is a one-round, $O(1)$-bandwidth algorithm for detecting $C_4$ under node deletions. 
% The bandwidth is $O(1)$ as we only send one-bit \Del{} messages. 
% For correctness, we prove that $C_w^r(u,v)= \Count_w^r(u,v)$ holds for all nodes $w$ and its neighbors $u,v$ after each round $r$. 
% This suffices since $w$ is in some copy of $C_4$ at the end of round $r$ if and only if $\Count_w^r(u,v)>0$ for some pair of neighbors $u,v$. 
% The invariant holds initially by our computation of $C_w^0(u,v)$. 
% For any round $r$, let the deleted node be $y$, and consider a fixed node $x$. 
% If $y$ is a neighbor of $x$, then the pair $y,z$ is no longer in $N^r(x)$ for any other neighbor $z$ of $x$ so the invariant holds vacuously. 
% If $y$ is at a distance of $3$ or more from $x$, neither $C_x^r(w,z)$ nor the true value $\Count_x^r(w,z)$ will change for any pair of neighbors $w,z$, so the invariant holds as well. 
% Finally, if $y$ is at a distance 2 from $x$, then it is adjacent to 1 or more neighbors of $x$. 
% By our algorithm, node $x$ receives \Del{} from both neighbors $w$ and $z$ if and only if a node adjacent to both $w$ and $z$ was deleted, so both $C_x^r(w,z)$ and $\Count_x^r(w,z)$ decrease by 1, maintaining the invariant.\yijun{Correctness proof is kind of trivial and tedious - can be removed}

\paragraph{Proof idea} 
To generalize the above algorithm from $C_4$ to an arbitrary complete multipartite graph $H$, we observe that, for any $v \in V(H)$, every node $u\in V(H)\setminus \{v\}$ is either a neighbor of $v$ or is adjacent to \emph{all} nodes in $N_H(v)$. 
    We call $V(H) \setminus (\{v\} \cup N_H(v))$ the \emph{independent} nodes of $v$, as they form an independent set because they belong to the same part as $v$ in the complete multipartite graph $H$.
    %, as  since they are the nodes that are in the same independent set as $v$ in $H$. 
    
    The deletion of a neighbor of $v$ in $H$ can be detected by $v$ immediately.
    The deletion of any independent node of $v$ in $H$ can be detected by \emph{every} neighbor of $v$ immediately without communication and can be detected by $v$ in one round of communication using one-bit messages.

    Similar to the algorithm for $C_4$, for each node $v$, each subset $S \subseteq N^r(v)$, and each round number $r$, we can let $v$ count the number of nodes in $V(G^r) \setminus (S \cup \{v\})$ that are adjacent to all nodes in $S$ at the end of round $r$. This information is sufficient for $v$ to decide if it belongs to a copy of $H$.

Now we present the formal proof realizing the above proof idea. By \Cref{thm: nediam=2}, the following theorem applies to any target subgraph $H$ with $\nediam(H) = 2$.

\begin{theorem}
\label{thm: memdect node deletion M}
For any connected complete multipartite graph $H$, 
    %If $H\in \mathcal{M}$ is a complete multipartite graph, 
    there is a one-round and $O(1)$-bandwidth algorithm for $\MemDetect(H)$ under node deletions.%\yijun{maybe just write $\nediam(H) = 2$ and $\diam(H)=\rad(H)=2$ as the condition of the theorem}
\end{theorem}
\begin{proof}
Let $k$ be the number of parts of the complete multipartite graph $H$. Since $H$ is connected, $k \geq 2$. For each $i \in [k]$, let $S_i \subseteq V(H)$ be the $i$th part of $H$. To realize the proof idea above, for each node $v$ in $G^r$ and each $S \subseteq N^r(v)$, consider the following terminology:
\begin{itemize}
    \item $\Count_v^r(S) =$ the number of nodes in $V(G^r) \setminus (S\cup \{v\})$ that are adjacent to all nodes in $S$.
    \item For each $i \in [k]$, \[\Count_v^r(S,i)=\begin{cases}
			\Count_v^r(S), & \text{if $G^r[S]$ contains a subgraph isomorphic to $H - S_i$,}\\
            0, & \text{otherwise.}
		 \end{cases}\] 
\end{itemize}
We claim that there exists a subgraph of $G^r$ containing $v$ isomorphic to $H$ if and only if there exist $i \in [k]$ and $S \subseteq N^r(v)$ such that $\Count_v^r(S,i) \geq |S_i| - 1$.
\begin{itemize}
    \item Suppose $\Count_v^r(S,i) > |S_i| - 1$, then a subgraph of $G^r$ containing $v$ isomorphic to $H$ can be obtained by combining the following pieces:
   \begin{itemize}
       \item Any subgraph of $G^r[S]$ isomorphic to $H - S_i$.
       \item Any $|S_i| -1$ nodes $U=\{u_1, \ldots, u_{|S_i| -1}\}$ in $V(G^r) \setminus (S\cup \{v\})$ that are adjacent to all nodes in $S$.
       \item The node $v$.
   \end{itemize} 
    Here $U \cup \{v\}$ plays the role of the $i$th part $S_i$ of $H$.
    \item Suppose there exists a subgraph of $G^r$ containing $v$ isomorphic to $H$. Let $X$ be the node set of the subgraph. Suppose $v$ is mapped to a node in $S_i$ in the isomorphism and $U=\{u_1, \ldots, u_{|S_i| -1}\}$ are the other nodes in $G^r$ that are mapped to $S_i$, then $\Count_v^r(X\setminus (U\cup\{v\})) \geq |U|=|S_i|-1$. This is because $G^r[X\setminus (U\cup\{v\})]$ contains a subgraph isomorphic to $H-S_i$ and all nodes in $U \cup\{v\}$ are adjacent to all nodes in $X\setminus (U\cup\{v\})$. 
\end{itemize}

By the above claim, to design a one-round algorithm for $\MemDetect(H)$, it suffices that each node $v$ learns $\Count_v^r(S,i)$ for all $i \in [k]$ and all $S \subseteq N^r(v)$ at the end of each round $r$. Since $G^r[S] = G^0[S]$ is already known to $v$, 
$\Count_v^r(S,i)$ can be calculated from $\Count_v^r(S)$. Similar to the algorithm for $C_4$, using one-bit messages, we can let each node $v$ learn $\Count_v^r(S)$ for all $S \subseteq N^r(v)$, as follows.

Initially, since $G^0$ is known to all nodes, each node $v$ can locally compute $\Count_v^0(S)$ for all $S \subseteq N^0(v)$.
In each round, if a node detects that one of its neighbors is deleted, it sends a signal \Del{} to all its neighbors. 
If a node $v$ receives \Del{} from all nodes in $S \subseteq N^r(v)$, then $v$ updates $\Count_v^r(S) = \Count_v^{r-1}(S)-1$, as this indicates that a common neighbor of all nodes in $S$ in $V(G^{r-1}) \setminus (S \cup \{v\})$ is deleted in round $r$. Otherwise, we have $\Count_v^r(S) = \Count_v^{r-1}(S)$. Hence the one-round bandwidth complexity of $\MemDetect(H)$ under node deletions is $O(1)$.
%Let $C = \{c_1, \ldots, c_k\}$ be the multiset of the sizes of the $k$ parts of $H$. 
%be $c_1, \ldots, c_k$, then we write $H = K_{c_1, \ldots, c_k}$ or $H = K_C$, where $C = \{c_1, \ldots, c_k\}$ is a multiset.  
\end{proof}

\section{One-Round Listing}\label{sect:list}
In this section, we investigate the one-round bandwidth complexity of the problem of $\List(H)$ for edge deletions and node deletions. For edge deletions, we obtain a complete characterization. For node deletions, we obtain an almost complete characterization, except for the case where $\rad(H)= \diam(H) = 2$.  See \Cref{tab:listing} for a summary of our results. Refer to \Cref{def:parameters,def:NE_parameters} for the definition of $\rad$, $\nerad$, and $\diam$.

\begin{table}[ht]
\begin{center}
% \usepackage{multirow}
\begin{tabular}{lllll}
\cline{2-5}
\multicolumn{1}{l|}{\textbf{}} & \multicolumn{1}{l|}{\multirow{2}{*}{$\nerad = 1$}} & \multicolumn{2}{l|}{$\nerad = 2$} & \multicolumn{1}{l|}{\multirow{2}{*}{$\nerad \geq 3$}} \\ \cline{3-4}
\multicolumn{1}{l|}{\textbf{}} & \multicolumn{1}{l|}{} & \multicolumn{1}{l|}{$\rad=1$} & \multicolumn{1}{l|}{$\rad=2$} & \multicolumn{1}{l|}{} \\ \cline{2-5} 
\multicolumn{1}{l|}{\multirow{2}{*}{\textbf{Edge deletions}}} & \multicolumn{1}{l|}{$0$} & \multicolumn{1}{l|}{$\Theta(1)$} & \multicolumn{1}{l|}{$\Theta(\log n)$} & \multicolumn{1}{l|}{Impossible} \\
\multicolumn{1}{l|}{} & \multicolumn{1}{l|}{[\ref{lem: list edge del rad=1}]} & \multicolumn{1}{l|}{[\ref{list_trivial_LB}][\ref{lem: list edge del rad=2 O(1)}]} & \multicolumn{1}{l|}{[\ref{thm list edel nerad=rad=2}][\ref{logn list edge del}]} & \multicolumn{1}{l|}{[\ref{thm list rad 3}]} \\ \cline{2-5} 
\textbf{} &  &  &  &  \\ \cline{2-5} 
\multicolumn{1}{l|}{\textbf{}} & \multicolumn{1}{l|}{\multirow{2}{*}{$\rad = 1$}} & \multicolumn{2}{l|}{$\rad = 2$} & \multicolumn{1}{l|}{\multirow{2}{*}{$\rad \geq 3$}} \\ \cline{3-4}
\multicolumn{1}{l|}{\textbf{}} & \multicolumn{1}{l|}{} & \multicolumn{1}{l|}{$\diam = 2$} & \multicolumn{1}{l|}{$\diam \geq 3$} & \multicolumn{1}{l|}{} \\ \cline{2-5} 
\multicolumn{1}{l|}{\multirow{2}{*}{\textbf{Node deletions}}} & \multicolumn{1}{l|}{$0$} & \multicolumn{1}{l|}{$O(\log n)$} & \multicolumn{1}{l|}{$\Theta(\log n)$} & \multicolumn{1}{l|}{Impossible} \\
\multicolumn{1}{l|}{} & \multicolumn{1}{l|}{[\ref{thm:list_zero}]} & \multicolumn{1}{l|}{[\ref{thm:ND_list_UB}]} & \multicolumn{1}{l|}{[\ref{thm:ND_list_LB}][\ref{thm:ND_list_UB}]} & \multicolumn{1}{l|}{[\ref{lem: list node del rad>=3}]} \\ \cline{2-5} 
\end{tabular}
\end{center}%\mingyang{Maybe we can update this table.}
    \caption{The bandwidth complexity of one-round $\List(H)$.}
    \label{tab:listing}
\end{table}

In \Cref{sect:list_ED}, we examine the listing problem under edge deletions.
In \Cref{sect:list_ND}, we examine the listing problem under node deletions.
%In \Cref{sect:list_Problems}, we discuss the remaining open problem of determining the right bandwidth complexity bound for the case where $\rad(H)= \diam(H) = 2$ under node deletions.

\subsection{Listing Under Edge Deletions}\label{sect:list_ED}
In this section, we give a complete characterization of the bandwidth complexity of one-round algorithms for $\List(H)$ under edge deletions.% in  \Cref{tab:listingED}.

%\mingyang{table edited}\mingyang{Maybe add a caption}

\begin{theorem}\label{thm list rad 3}
    If $\nerad(H) \geq 3$, then there is no one-round algorithm for $\List(H)$ under edge deletions.
\end{theorem}

\begin{proof}
    Let $H$ be a connected graph with $\nerad(H) \geq 3$, and the initial graph $G^0$ be a single copy of $H$.
    At least one node must be listing this copy of $H$ initially, call this node $u$. Since $\nerad(H) \geq 3$, there exists some edge $\{v,w\}$ such that neither $v$ nor $w$ is a neighbor of $u$.
    Suppose we delete this edge $\{v,w\}$. Only $v$ and $w$ can detect the edge deletion immediately and perhaps send a message to their respective neighbors in this round. However, $u$ is neither a neighbor of $v$ nor $w$, so it cannot learn of this edge deletion by the end of this round. Hence $u$ is not able to correctly stop listing the copy of $H$.
\end{proof}

\Cref{thm list rad 3} allows us to restrict our attention to the case where $\nerad(H) \in \{1,2\}$. We observe that the case $\nerad(H) = 1$ can be solved with zero bandwidth, leaving $\nerad(H) = 2$ as the only nontrivial scenario.

\begin{theorem} \label{lem: list edge del rad=1}
    If $\nerad(H) = 1$, then there exists a one-round algorithm for $\List(H)$ under edge deletions, with zero bandwidth required.
\end{theorem}

\begin{proof}
    If $\nerad(H) = 1$, then $H$ must be a star graph with a center $u$ that is adjacent to all other nodes in $H$.

    \paragraph{Algorithm} Initially, each copy $H_i$ of $H$ in the initial graph $G^0$ is listed by its center $u_i$. Whenever an edge is deleted, any center $u_j$ that is adjacent to the deleted edge stops listing $H_j$ immediately. Note that a node can be a center for multiple subgraphs: It is possible that a node $x$ is both $u_j$ and $u_k$ at the same time, with $j \neq k$.

    \paragraph{Proof of correctness} Initially, all copies $H_i$ of $H$ are correctly listed by exactly one node $u_i$.
    Suppose during a round, $H_j$ is destroyed due to the deletion of some edge $e$. Since $H_j$ is a star with center $u_j$, $e$ must be adjacent to $u_j$, so $u_j$ can stop listing $H_j$ correctly.
    Lastly, no new copies of $H$ can be formed by edge deletions. Hence the algorithm is correct.
\end{proof}

The condition $\nerad(H) = 1$ in \Cref{lem: list edge del rad=1} is necessary, as non-zero bandwidth is required whenever $\nerad(H) \geq 2$.

\begin{observation} \label{list_trivial_LB}
    If $\nerad(H) \geq 2$, then any algorithm for $\List(H)$ under edge deletions must have non-zero bandwidth complexity.
\end{observation}
\begin{proof}
    If $\nerad(H) \geq 2$, then there exist a node $u$ and an edge $e$ in $H$ such that $u$ is not an endpoint of $e$. Without communication, $u$ cannot detect the deletion of $e$. 
\end{proof}

\begin{theorem}\label{lem: list edge del rad=2 O(1)}
    If %$\nerad(H) = 2$ and 
    $\rad(H) = 1$, then there exists a one-round algorithm for $\List(H)$ under edge deletions, with bandwidth complexity $O(1)$.
\end{theorem}

\begin{proof}
    If $\rad(H) = 1$, then $H$ must be a graph with a center $u$ adjacent to all other nodes, and with some edges between the nodes in  $V(H) \backslash \{u\}$. The choice of the center is not unique in general.

    \paragraph{Algorithm}
    \begin{itemize}
        \item Initially, each copy $H_i$ of $H$, in the initial graph $G^0$ is listed by any one $u_i$ of its centers.
        \item Whenever an edge $\{v,w\}$ is deleted, if $v = u_j$ or $w = u_j$ is the center responsible for listing $H_h$, then they stop listing $H_j$ immediately.
        \item Both $v$ and $w$ send a one-bit message $\Del$ to their respective neighbors, indicating that they have an incident edge deleted.
        \item If any center $u_j$ receives the message $\Del$ from two distinct neighbors $v$ and $w$ and $\{v,w\}$ is an edge in $H_j$, then $u_j$ stops listing $H_j$ immediately.
    \end{itemize}

    \paragraph{Proof of correctness} Initially, every copy $H_i$ of $H$ is correctly listed by exactly one node $u_i$.
    Suppose during a round, $H_j$ is destroyed due to the deletion of some edge $e$. If $e$ was adjacent to $u_j$, then $u_j$ stops listing $H_j$ correctly. Otherwise, $e = \{v,w\}$ where both $v$ and $w$ are neighbors of $u_j$, as $u_j$ is a center of $H_j$. In this case, $u_j$ receives the message $\Del$ from both $v$ and $w$, allowing it to stop listing $H_j$ correctly as well.
    Again, same as the proof of \Cref{lem: list edge del rad=1}, no new copies of $H$ can be formed by edge deletions. Hence the algorithm is correct.
\end{proof}

\Cref{list_trivial_LB,lem: list edge del rad=2 O(1)} together establish that the tight one-round bandwidth complexity bound for  $\List(H)$ under edge deletions is $\Theta(1)$ for any target subgraph $H$ with $\nerad(H) = 2$ and $\rad(H) = 1$. Next, we show that this bound increases to $\Theta(\log n)$ for any target subgraph $H$ with $\nerad(H) = 2$ and $\rad(H) = 2$.

\begin{theorem} \label{thm list edel nerad=rad=2}
    If $\nerad(H) = 2$ and $\rad(H) = 2$, then any one-round algorithm for $\List(H)$ under edge deletions must have bandwidth complexity $\Omega(\log n)$.
\end{theorem}

\begin{proof}
    Let $|V(H)| = m$ and $V(H) = \{u_1, u_2, \ldots, u_m\}$. We construct the initial graph $G^0$ in the following way:
    
    \begin{itemize}
        \item Start with a single copy of $H$.
        \item Let $n$ be any positive integer. Replace each node $u_i$ with an independent set of $n$ nodes $S_i = \{u_{i,1}, u_{i,2}, \ldots, u_{i,n}\}$. 
        \item In other words, $V(G^0) = \bigcup_{i \in [m]} S_i$ and for each pair of nodes $u_{i, a}$ and $u_{j, b}$, $\{u_{i, a},u_{j, b}\} \in E(G^0)$ if and only if $i \neq j$, $\{u_i,u_j\} \in E(H)$. That is, $S_i \cup S_j$ induces a complete bipartite graph if and only if $\{u_i,u_j\} \in E(H)$.
        \item Now we have $|V(G^0)| = mn = \Theta(n)$, and $G^0$ contains $n^m$ copies of $H$ in the form $\{u_{1,i_1}, u_{2,i_2}, \ldots, u_{m,i_m}\}$. There may be more copies of $H$ in other forms, but they are irrelevant to our proof.
    \end{itemize}

    By the pigeonhole principle, one of the nodes must list at least $n^m / mn = n^{m-1} / m$ copies of $H$ initially. Without loss of generality, let this node be $u_{1,1}$.
    
    Since $\rad(H) = \nerad(H) = 2$, there must exist two nodes $u_i, u_j \in V(H)$ such that $\{u_1,u_i\}$, $\{u_i, u_j\} \in E(H)$, and $\{u_1, u_j\} \notin E(H)$. Without loss of generality, let $u_i$ and $u_j$ be $u_{m-1}$ and $u_m$, respectively. 

    Now, consider the set of $n^{m-1}/m$ copies of $H$ being listed by $u_{1,1}$ initially. Since there are $n$ copies of the node $u_2$ in $G^0$, by the pigeonhole principle again, at least $1/n$ of the copies of $H$ in this set must have the same $u_{2,i}$. Without loss of generality, let it be $u_{2,1}$.

    Repeating the same argument for $u_3, \ldots, u_{m-1}$, we may assume that $u_{1,1}$ must be listing at least $n/m$ copies of $H$ in the form $\{u_{1,i}, u_{2,1}, u_{3,1}, \ldots, u_{m-1, 1}, u_{m, j}\}$ with $i \in [n]$ and $j\in [n]$. Let this set of $n/m$ copies of $H$ be $L$. Since $L$ has at least $n/m$ distinct elements, there must be at least $\sqrt{n/m}$ copies of $H$ with distinct values of $i$, or at least $\sqrt{n/m}$  copies of $H$ with distinct values of $j$ in $L$. 

    \paragraph{Case 1} Suppose $u_{1,1}$ is listing at least $\sqrt{n/m}$ copies of $H$ in the form $\{u_{1,i}$, $u_{2,1}$, $u_{3,1}$, $\ldots$, $u_{m-1, 1}$, $u_{m, j}\}$ where the values of $i$ are distinct. Let $I$ be the set of distinct values of $i$ here.

    Recall that $\{u_1, u_{m-1}\} \in E(H)$, $\{u_{m-1}, u_m\} \in E(H)$, and $\{u_1, u_m\} \notin E(H)$. Hence $\{u_{1,i}, u_{m-1, 1}\} \in E(G^0)$ for all $i \in I$. Suppose we now delete the edge $\{u_{1,i^*} , u_{m-1, 1}\}$ for some $i^* \in I \setminus\{ 1\}$, then $u_{1,1}$ must stop listing any copies of $H$ containing this edge by the end of this round, and there is at least one such copy of $H$ in $L$. Since there is no edge between $u_{1,i^*}$ and $u_{1,1}$, $u_{m-1, 1}$ must send a message to $u_{1,1}$ that allows $u_{1,1}$ to distinguish $i^*$ from the total pool of at least $\sqrt{n/m}-1  = \Omega(\sqrt{n})$ possibilities, thus requiring $\Omega(\log \sqrt{n}) = \Omega(\log n)$ bandwidth.

    \paragraph{Case 2} Suppose $u_{1,1}$ is listing at least $\sqrt{n/m}$ copies of $H$ in the form $\{u_{1,i}$, $u_{2,1}$, $u_{3,1}$, $\ldots$, $u_{m-1, 1}$, $u_{m, j}\}$ where the values of $j$ are distinct. Let $J$ be the set of distinct values of $j$ here.

    Suppose we now delete the edge $\{u_{m-1,1}, u_{m, j^*}\}$ for some $j^* \in J$, then by a similar argument to the previous case, since there is no edge between $u_{m, j^*}$ and $u_{1,1}$, $u_{m-1, 1}$ must send a message to $u_{1,1}$ that allows it to distinguish $j^*$ from the total pool of at least $\sqrt{n/m}$ possibilities, thus requiring $\Omega(\log n)$ bandwidth as well.
\end{proof}

With a simple algorithm, we can show that the lower bound of \Cref{thm list edel nerad=rad=2} is \emph{tight}.

\begin{theorem} \label{logn list edge del}
    If $\nerad(H) = 2$ and $\rad(H) = 2$, then there exists a one-round algorithm for $\List(H)$ under edge deletions, with bandwidth complexity $O(\log n)$.
\end{theorem}

\begin{proof}
    The proof idea is that, with $O(\log n)$ bandwidth, a node can send the exact identity of the deleted edge to all its neighbors.

    \paragraph{Algorithm}
    Initially, each copy $H_i$ of $H$ in the original graph $G^0$ is listed by exactly one node $u_i$ from $\necenter(H_i)$. Recall from \Cref{def:NE_parameters} that $u_i$ is a node such that the node-edge distance  $\dist(u_i, e) \leq \nerad(H) = 2$ for all $e \in E(H_i)$. Whenever an edge $\{v,w\}$ is deleted, both endpoints $v$ and $w$ send a $O(\log n)$-bit message to their respective neighbors containing the $\ID$s of $v$ and $w$, allowing them to determine exactly which edge has been deleted. If some $u_i$ determines that an edge $e \in E(H_i)$ has been deleted, then it stops listing $H_i$ immediately.

    \paragraph{Proof of correctness}
    Initially, all copies $H_i$ of $H$ are correctly listed by exactly one node $u_i$.
        Suppose during a round, $H_j$ is destroyed due to the deletion of some edge $e$. Since $\nerad(H_j) = 2$ with $u_j$ being a node from $\necenter(H_j)$, $e$ must be incident to $u_j$ or incident to some neighbor of $u_j$. In either case, $u_j$ can determine the exact identity of $e$ and deduce that it is part of $H_j$, thus stopping listing $H_j$ immediately. As no new copies of $H$ can be formed by edge deletions, the algorithm is correct.
\end{proof}

\subsection{Listing Under Node Deletions}\label{sect:list_ND}
In this section, we show that many proofs from \Cref{sect:list_ED} can be adapted to the case of node deletions with minor modifications. However, we are unable to give a tight bound in the case where $\rad(H) = \diam(H) = 2$.
%See \Cref{tab:listingND} for a summary of our results for one-round \List\ under node deletions.

\begin{theorem} \label{lem: list node del rad>=3}
    If $\rad(H) \geq 3$, then there is no one-round algorithm for $\List(H)$ under node deletions.
\end{theorem}

\begin{proof}
%The proof is similar to the proof of \Cref{thm list rad 3}.
    Let $H$ be a connected graph with $\rad(H) \geq 3$, and the initial graph $G^0$ be a single copy of $H$.    
    At least one node must be listing this copy of $H$ initially, and call this node $u$. Since $\rad(H) \geq 3$, there exists some node $v$ such that $\dist(u,v) \geq 3$.
    Suppose we delete the node $v$. Since $\dist(u,v) \geq 3$, $v$ is not a neighbor of $u$ and is not adjacent to a neighbor of $u$. Thus $u$ cannot learn of the deletion of $v$ by the end of this round. Hence $u$ is not be able to correctly stop listing this copy of $H$.
\end{proof}

\begin{theorem}\label{thm:list_zero}
    If $\rad(H) = 1$, then there exists a one-round algorithm for $\List(H)$ under node deletions, with zero bandwidth required.
\end{theorem}

\begin{proof}
    If $\rad(H) = 1$, then $H$ must be a star graph. 
    The algorithm and proof of correctness are almost identical to those in \Cref{lem: list edge del rad=1}. The only difference is that each $u_j$ stops listing $H_j$ when it detects that some node (rather than an edge) in $H_j$ is deleted.
\end{proof}

\begin{theorem}\label{thm:ND_list_LB}
    If $\rad(H) = 2$ and $\diam(H) \geq 3$, then any one-round algorithm for $\List(H)$ under node deletions must have bandwidth complexity $\Omega(\log n)$.
\end{theorem}

\begin{proof}
    Let $v,w \in V(H)$ be two nodes such that $\dist(v,w) \geq 3$. Let $|V(H)| = m$, and $V(H) = \{v, w, u_1, u_2, \ldots, u_{m-2}\}$. We construct the initial graph $G^0$ in the following way:
    
    \begin{itemize}
        \item Start with a single copy of $H$.
        \item  Let $n$ be any positive integer. Replace $v$ with an independent set of $n$ nodes $S_v = \{v_1, v_2, \ldots, v_n\}$. Replace $w$ with an independent set of $n$ nodes $S_w = \{w_1, w_2, \ldots, w_n\}$. Therefore, $V(G^0) = S_v \cup S_w \cup \{u_1, u_2, \ldots, u_{m-2}\}$.
        %For each of $v$ and $w$, duplicate it $n-1$ times to get the sets of nodes $S_v = \{v_1, v_2, \ldots, v_n\}$ and $S_w = \{w_1, w_2, \ldots, w_n\}$.
        % More formally, $V(G^0) = S_v \cup S_w \cup \{u_1, u_2, \ldots, u_{m-2}\}$ and $E(G^0)$ consists of the following edges.
        % \begin{itemize}
        %     \item $\{u_i, u_j\} \in E(G^0)$ if $\{u_i, u_j\} \in E(H)$.
        %     \item $\{u_i, v_j\} \in E(G^0)$ if $\{u_i, v\} \in E(H)$.
        %     \item $\{u_i, w_j\} \in E(G^0)$ if $\{u_i, w\} \in E(H)$.
        %     \item $\{v_i, w_j\} \in E(G^0)$ for all $v_i$ and $w_j$.
        % \end{itemize}
        \item Each $v_i \in S_v$ has the same set of neighbors as $v$ in $H$. Each $w_i \in S_w$ has the same set of neighbors as $w$ in $H$.
        % \item i.e., $\{v_i, u_j\} \in E(G^0)$ if and only if $\{v,u_j\} \in E(H)$. Same for $w_i$.
        \item Now we have $|V(G^0)| = 2n+m-2 = \Theta(n)$, and $G^0$ contains $n^2$ copies of $H$ in the form $\{v_{i}, w_{j}, u_1, u_2, \ldots, u_{m-2}\}$. There may be more copies of $H$ in other forms, but they are irrelevant to our proof.
    \end{itemize}

    Observe that if the eccentricity of a node is at least $3$ in $G^0$, it cannot list any of the $n^2$ copies of $H$ initially. Otherwise, by the same argument as in \Cref{lem: list node del rad>=3}, the node is not able to correctly stop listing some copy of $H$ when a node at distance at least $3$ away is deleted. Since $\dist(v,w) \geq 3$ in $H$, the eccentricity of all nodes in $S_v$ and $S_w$ are at least $3$ in $G^0$. Hence at most $m-2$ nodes $\{u_1, u_2, \ldots, u_{m-2}\}$ can be listing any of the $n^2$ copies of $H$ initially.

    By the pigeonhole principle, one of them must list at least $n^2 / (m-2)$ copies of $H$ initially. Without loss of generality, let this node be $u_1$.
    
    Consider the relationship between $u_1$, $v$, and  $w$ in $H$. Since $\dist(v,w) \geq 3$, $u_1$ cannot be adjacent to both $v$ and $w$. Without loss of generality, suppose $u_1$ is not adjacent to $v$. However, as we have argued above, the eccentricity of $u_1$ cannot be at least $3$ in $G^0$, so we must have $\dist(u_1, v) = 2$ in $H$.

    Now consider the set of $n^2/(m-2)$ copies of $H$ being listed by $u_1$ initially. Since there are $n$ copies of the node $w$ in $G^0$, by the pigeonhole principle again, at least $1/n$ of the copies of $H$ in this set must have the same $w_i$. Without loss of generality, let it be $w_1$.

    Hence, $u_1$ must be listing at least $n/(m-2)$ copies of $H$ in the form $\{v_i, w_1, u_1, u_2, \ldots, u_{m-2}\}$. That is, $u_1$ is listing at least $n/(m-2)$ copies of $H$ with distinct values of $i$. Let $I$ be the set of distinct values of $i$.

    Recall that $\dist(u_1, v) = 2$ in $H$. Hence $\dist(u_1, v_i) = 2$ for all $i \in I$ in $G^0$. Suppose we now delete the node $v_{i^*}$ for some $i^* \in I$, then $u_1$ must stop listing any copies of $H$ containing $v_{i^*}$ by the end of this round, and there is at least one such copy. However, $v_{i^*}$ is not adjacent to any other $v_i$ or any $w_j$, so $u_1$ and $v_{i^*}$ have at most $m-3$ common neighbors, $\{u_2, \ldots, u_{m-2}\}$, in $G^0$ that can detect the deletion of $v_{i^*}$ immediately. Thus, $u_1$ can only receive at most $m-3$ messages to help it distinguish $i^*$ from the total pool of at least $n/(m-2) = \Omega(n)$ possibilities, requiring $\Omega(\log n)$ bandwidth.
\end{proof}

With a simple algorithm, we can show that the above lower bound is tight for any target subgraph $H$ with $\rad(H) = 2$.

\begin{theorem}\label{thm:ND_list_UB}
    If $\rad(H) = 2$, then there exists a one-round algorithm for $\List(H)$ under node deletions, with bandwidth complexity $O(\log n)$.
\end{theorem}

\begin{proof}
    The algorithm is similar to the $O(\log n)$-bandwidth algorithm for $\List(H)$ under edge deletions in \Cref{logn list edge del}.
    
    \paragraph{Algorithm} Initially, each copy $H_i$ of $H$ in the initial graph $G^0$ is listed by exactly one $u_i$ of its centers. See \Cref{def:parameters} for the definition of centers. Whenever a node $v$ is deleted, each neighbor of $v$ sends a $O(\log n)$-bit message to its respective neighbors containing the $\ID$ of $v$, allowing them to determine exactly which node has been deleted. If a center $u_i$ determines that a node $w \in V(H_i)$ has been deleted, it stops listing $H_i$ immediately.

    \paragraph{Proof of correctness} Initially, all copies $H_i$ of $H$ are correctly listed by exactly one node $u_i$.

    Suppose during a round, $H_j$ is destroyed due to the deletion of some node $v$. Since $\rad(H_j) = 2$ with $u_j$ being a center, $v$ must be adjacent to $u_j$ or a neighbor of $u_j$. In either case, $u_j$ is able to determine the exact identity of the deleted node and stop listing $H_j$ immediately.
As no new copies of $H$ can be formed by node deletions, the algorithm is correct.
\end{proof}

\section{Conclusions and Open Problems}\label{sect:conclusions}
In this work, we substantially extend the study of dynamic distributed subgraph finding initiated by Bonne and Censor-Hillel~\cite{bonne2019distributed} in the \emph{deterministic} setting. We establish \emph{tight} one-round bandwidth bounds for triangle finding in bounded-degree dynamic networks: $\Theta(\log \log n)$ for membership-detection under edge insertions only (\Cref{cor-1}), and $\Theta(\log \log \log n)$ for detection when both edge and node insertions are allowed (\Cref{cor-2}). Before our work, no lower bound was known for these two problems. Moreover, we provide a \emph{complete characterization} of the $r$-round bandwidth complexity of the membership-listing problem across all subgraphs and types of topological changes (\Cref{table:memList1}). Despite these advances, many intriguing open problems remain.

\begin{description}
    \item[Beyond bounded-degree networks] While we obtain tight bounds for triangle finding in bounded-degree networks, the current upper and lower bounds remain unmatched for general unbounded-degree networks. Can stronger lower bounds be established for networks with higher degrees?

    Specifically, for membership-detection, in \Cref{lem: memdect lb ks}, we establish a new lower bound of $\Omega(\log\log n)$ for the one-round bandwidth complexity of $\MemDetect(K_s)$ under edge insertions for any $s \geq 3$. While we provide a matching upper bound for $\MemList(K_3)$ in \emph{bounded-degree} networks in \Cref{thm: Delta memlist ub loglogn}, this lower bound is not yet known to be tight for \emph{unbounded-degree} networks, where the current best upper bound is $O(\log n)$ for $s = 3$ and $O(\sqrt{n})$ for $s \geq 4$~\cite{bonne2019distributed}. Closing these gaps remains an intriguing open question.

    \item[Randomized algorithms] While we focus on deterministic algorithms in this paper, many of our lower bounds extend to randomized algorithms, as shown in \Cref{sect:random}. Yet the role of randomness in reducing bandwidth complexity for dynamic distributed subgraph finding is not well understood: Which problems exhibit an advantage for randomized over deterministic algorithms?

    \item[Round-bandwidth tradeoffs] While our complete characterization of the membership-listing problem applies to $r$-round algorithms with an arbitrary round number $r$, the remainder of our results---and much of the existing literature---primarily focuses on the one-round scenario.

    A particularly illustrative case is the membership-listing of cliques: In the one-round setting, the bandwidth complexity is $\Theta(\sqrt{n})$, whereas allowing two rounds reduces the bandwidth complexity to $\Theta(1)$~\cite{bonne2019distributed}. This stark contrast shows the potential benefits of additional communication rounds in lowering bandwidth requirements. Exploring how increased round numbers influence bandwidth complexity remains an interesting avenue for future research.

    \item[Toward complete characterizations of the remaining problems] In this work, we provide partial characterizations for one-round membership-detection and listing. Can these characterizations be completed? What can be said about the detection problem? We discuss several specific open questions in \Cref{sect:memdect_open,sect:list_Problems}.
\end{description}

\subsection{Membership Detection}\label{sect:memdect_open}
In this work, we provide a complete characterization of the one-round bandwidth complexity of $\MemDetect(H)$ for all choices of $H$ under \emph{node insertions}. However, as shown in \Cref{tab:memDect}, several open questions remain for the other three types of topological changes.

For edge deletions, there is still a gap between the current upper bound of $O(\log n)$ (\Cref{lem: memlist edge deletion}) and the lower bound of $\Omega(1)$ (\Cref{memdetect_trivial_LB2}) for any complete multipartite graph $H$ that is neither a clique nor a star. In particular, the following question remains unresolved.

\begin{question}
    What is the one-round bandwidth complexity of $\MemDetect(C_4)$ under edge deletions?
\end{question}

For node deletions, the only unsolved case is when $\diam(H) = \rad(H) =2$ and $\nediam(H)=3$: There is still a gap between the current upper bound of $O(\log n)$ (\Cref{lem: memlist node deletion}) and the lower bound of $\Omega(1)$ (\Cref{memdetect_trivial_LB2}). In particular, the following question remains unresolved.

\begin{question}
    What is the one-round bandwidth complexity of $\MemDetect(C_5)$ under node deletions?
\end{question}

\subsection{Listing}\label{sect:list_Problems}
For $\List(H)$, in \Cref{tab:listing}, the only remaining open problem is to determine the one-round bandwidth complexity of $\List(H)$ under node deletions
for the case of $\rad(H) = \diam(H) = 2$. 

\begin{question}
    What is the one-round bandwidth complexity of $\List(H)$ under node deletions
for the case of $\rad(H) = \diam(H) = 2$. 
\end{question}

It is challenging to obtain a non-trivial lower bound for this case, as the high degree of symmetry in $H$ makes it difficult to construct a ``worst-case'' initial graph $G^0$. %Hence the lower bound for a general one-round $\List(H)$ algorithm under node deletions is still unknown in this case. 
However, we conjecture that there is a matching lower bound $\Omega(\log n)$ due to some of our observations below.

First, all the $\List(H)$ algorithms we have used in our upper bounds so far are \underline{stable} ones, where a \underline{stable} algorithm is defined as follows.

\begin{definition}
    A one-round $\List(H)$ algorithm, under edge deletions or node deletions, is \underline{stable} if it satisfies the following property: For all $u \in V(G^0)$ in the initial graph $G^0$, if $u$ lists some copy $H_i$ of $H$ initially, then $u$ only stops listing $H_i$ in the exact round when $H_i$ is destroyed. That is, $u$ does not stop listing $H_i$ prematurely. 
\end{definition}

If we restrict our attention to \underline{stable} algorithms, the desired $\Omega(\log n)$ lower bound can be obtained.

\begin{theorem}\label{thm:stableLB}
    If $\rad(H) = \diam(H) = 2$, then any one-round \underline{stable} algorithm for $\List(H)$ under node deletions must have bandwidth complexity $\Omega(\log n)$.    
\end{theorem}

\begin{proof}
    The setup for this proof is identical to the one for \Cref{thm list edel nerad=rad=2}, so we omit the details and only highlight the important parts.
    
    Let $|V(H)| = m$, and $V(H) = \{u_1, u_2, \ldots, u_m\}$. Construct the initial graph $G^0$ in the same way as in the proof for \Cref{thm list edel nerad=rad=2}. Now we have $|V(G^0)| = mn = \Theta(n)$, and $G^0$ contains $n^m$ copies of $H$ in the form $\{u_{1,i_1}, u_{2,i_2}, \ldots, u_{m,i_m}\}$.

    Using the same repeated application of the pigeonhole principle, we may assume, without loss of generality, that the node $u_{1,1}$ that is listing at least $n/m$ copies of $H$ in the form $\{u_{1,i}, u_{2,1}, u_{3,1}, \ldots, u_{m-1, 1}, u_{m, j}\}$. Additionally, $\dist(u_1, u_m) = 2$ in $H$.

    Now we delete all nodes $u_{i,j}$, for $i \in [2,m-1]$ and $j \in [2,n]$, over several rounds. Note that no copy of $H$ in the form $\{u_{1,i}, u_{2,1}, u_{3,1}, \ldots, u_{m-1, 1}, u_{m, j}\}$ is destroyed in this process. Since we assume the algorithm is \underline{stable}, $u_{1,1}$ must still be listing at least $n/m$ copies of $H$ in this form.

    Using the same analysis as in the proof for \Cref{thm list edel nerad=rad=2}, we only need to consider the following two scenarios.

    \paragraph{Case 1} Suppose $u_{1,1}$ is listing at least $\sqrt{n/m}$ copies of $H$ in the form $\{u_{1,i}$, $u_{2,1}$, $u_{3,1}$, $\ldots$, $u_{m-1, 1}$, $u_{m, j}\}$ where the values of $i$ are distinct. Let $I$ be the set of distinct values of $i$ here.

    We now delete the node $u_{1,i^*}$ for some $i^* \in I \setminus \{1\}$, then, $u_{1,1}$ must stop listing any copies of $H$ containing $u_{1,i^*}$ by the end of this round, and there is at least one such copy. However, $u_{1,i^*}$ is not adjacent to $u_{1,1}$ or any $u_{m,j}$, so $u_{1,i^*}$ and $u_{1,1}$ have at most $m-2$ common neighbors, $\{u_{2,1}, u_{3,1}, \ldots, u_{m-1, 1}\}$, in the current network that can detect the deletion of $u_{1,i^*}$ immediately. Thus, $u_{1,1}$ can only receive at most $m-2$ messages to help it distinguish $i^*$ from the total pool of at least $\sqrt{n/m}-1  = \Omega(\sqrt{n})$ possibilities, thus requiring $\Omega(\log \sqrt{n}) = \Omega(\log n)$ bandwidth.

    \paragraph{Case 2} Suppose $u_{1,1}$ is listing at least $\sqrt{n/m}$ copies of $H$ in the form $\{u_{1,i}$, $u_{2,1}$, $u_{3,1}$, $\ldots$, $u_{m-1, 1}$, $u_{m, j}\}$ where the values of $j$ are distinct. Let $J$ be the set of distinct values of $j$ here.

     We now delete the node $u_{m,j^*}$ for some $j^* \in J$. By a similar argument to the previous case, $u_{1,1}$ can only receive at most $m-2$ messages to help it distinguish $j^*$ from the total pool of at least $\sqrt{n/m}$ possibilities, thus requiring $\Omega(\log n)$ bandwidth as well.
\end{proof}

It appears to be quite challenging to design a \underline{non-stable} algorithm for $\List(H)$ under node deletions to break the lower bound of \Cref{thm:stableLB}, so we conjecture that the same lower bound also applies to general algorithms.
%for $\rad(H) = \diam(H) = 2$, any one-round algorithm for $\List(H)$, under node deletions, also has bandwidth complexity $\Omega(\log n)$. 

\bibliographystyle{alpha}
\bibliography{reference}
%\newpage
\appendix
% \yanyu{add back due to some missing reference, can discuss if we include. Technically this should not be new, but maybe cuz it's too simple i could not find a citation. The only citation is the website from wikipedia 
% \href{https://www.graphclasses.org/classes/gc_1237.html}{link}}
% \input{archive/structure}
% \input{archive/backup}
% \input{archive/archive}

\section{Randomized Subgraph Finding}\label{sect:random}

In the appendix, we show that many of our lower bounds can be extended to the randomized setting using the approach of Bonne and Censor-Hillel~\cite{bonne2019distributed}. In their work~\cite{bonne2019distributed}, the lower bounds are proved via hard dynamic instances that consist of a sequence of topological changes, where the final topological change creates or removes a copy of $H$. In the deterministic setting, the algorithm must be correct for \emph{every} admissible sequence. To argue about randomized algorithms, they~\cite{bonne2019distributed} allow a constant error probability: The creation or removal of a copy of $H$ due to the final topological change is correctly identified with probability at least $1-\varepsilon$, for some constant $\varepsilon>0$. While Bonne and Censor-Hillel~\cite{bonne2019distributed} did not specify a single formal randomized model, the above setting is compatible with several natural variants: global vs.\ local failure probabilities; one-sided vs.\ two-sided error; local error relative to nodes, topological changes, or copies of $H$.

All the deterministic lower bounds by Bonne and Censor-Hillel~\cite{bonne2019distributed} follow a common template: There is a \emph{fixed} node $v$ such that each admissible update sequence yields a distinct output (e.g., creation or removal of a particular copy of $H$) for $v$. Since $v$ must output correctly in all cases, if there are $X$ possible sequences, $v$ must acquire $\Omega(\log X)$ bits of information.

To extend such a lower bound to the randomized setting, let $\mathcal{D}$ be the uniform distribution over the $X$ admissible update sequences in the hard instance. If there is a randomized algorithm that allows $v$ to output correctly with probability at least $1-\varepsilon$ on $\mathcal{D}$, then there exists a fixing of the randomness that yields a deterministic algorithm that is correct on a $(1-\varepsilon)$-fraction of the sequences. The deterministic lower bound still applies to this derandomized algorithm after replacing $X$ with $(1-\varepsilon)X$, which preserves the asymptotic $\Omega(\log X)$ bound.

\paragraph{\Cref{sect:memList,sect:memDect}: The lower bounds extend verbatim}
All lower bounds in \Cref{sect:memList,sect:memDect} follow exactly the paradigm above: Each sequence of topological updates maps to a distinct outcome; a predetermined node $v$ must be correct in every case; and the proof quantifies over the admissible update sequences. Hence, by the randomized extension of Bonne and Censor-Hillel~\cite{bonne2019distributed}, all these lower bounds extend to the randomized setting with the same asymptotics.

\paragraph{\Cref{sect:list}: Extension with a small modification}
In the lower bound proofs (\Cref{thm list edel nerad=rad=2} and \Cref{thm:ND_list_LB}) in \Cref{sect:list}, the node $v$ responsible for producing the output is not fixed \emph{a priori}, so the preceding randomized extension does not apply immediately. However, a simple workaround suffices. In the following discussion, we focus on the edge deletion model (\Cref{thm list edel nerad=rad=2}). The case of node deletion (\Cref{thm:ND_list_LB}) can be handled similarly.

We use the same hard instance in the proof of \Cref{thm list edel nerad=rad=2}, with a minor modification that, in the final step, a \emph{uniformly random} edge in the constructed graph is deleted. Suppose a randomized algorithm succeeds with probability at least $1-\varepsilon$ over the hard instance. 

By a union bound over all $e \in E(H)$, there exists a fixing of the randomness so that, for at least a $1-\varepsilon|E(H)|$ fraction of the copies of $H$, the algorithm correctly reacts to the deletion of \emph{any} edge of that copy. As long as $\varepsilon|E(H)|<1$, the original argument in the proof of \Cref{thm list edel nerad=rad=2} can be applied to the derandomized algorithm after restricting attention to a $1-\varepsilon|E(H)|$ fraction of the copies of $H$. The asymptotic lower bounds therefore extend to randomized algorithms.

\paragraph{\Cref{sect:cliqueLB}: The randomized extension fails}
It can be formally proved that the randomized extension technique does not work for our clique-finding lower bounds (\Cref{lem: memdect lb ks} and \Cref{thm:mixed}) in \Cref{sect:cliqueLB}, as we can design randomized algorithms that break the two lower bounds on our hard instances. 

The reason that the technique does not work is that our current lower bound PROOF relies on the correctness of the algorithm on nearly all instances, and not just a constant fraction of the instances. For example, one part of the lower bound argument is that we want to argue that for a length-2 path $u\!-\!x\!-\!v$, if the bandwidth and the number of rounds are not enough, one endpoint $u$ cannot learn the $\ID$ of the other endpoint $v$, so when a new edge incident to $u$ appears, $u$ cannot tell whether the edge is $\{u,v\}$ (so that a triangle is formed) or the edge is $\{u,w\}$ for some other node $w$ (so that a triangle is not formed).  However, as long as $u$ can learn $X$ bits of information from $v$, this is already sufficient for $u$ to differentiate $\ID(v)$ from $1 - O(1/2^X)$ fraction of the identifiers, meaning that the algorithm may already succeed on an overwhelming fraction of update sequences. Our lower bound argument crucially requires that we choose $w$ adversarially according to the behavior of the algorithm. This is fundamentally different from all other lower bound proofs in this paper.

For simplicity, we restrict attention to triangles (i.e., $\MemDetect(K_3)$ and $\Detect(K_3)$) in the subsequent discussion. The hard instance in our lower bounds in \Cref{sect:cliqueLB} is constructed as follows. Start with $t$ length-$2$ paths. Then perform one edge insertion. There are two cases: (1) inserting an edge connecting the two endpoints of a length-$2$ path, forming a triangle, or (2) inserting an edge concatenating two length-$2$ paths, not forming a triangle. If node insertion is allowed (in the setting of \Cref{thm:mixed}), there is a third case: (3) a new node is inserted with edges incident to the two endpoints of a length-$2$ path, not forming a triangle.

We give a randomized algorithm that succeeds with probability $1 - O(1/\poly\log\log\log n)$ using bandwidth $O(\log\log\log\log n)$ over our hard instances, confirming that the lower bound proofs in \Cref{sect:cliqueLB} do \emph{not} admit a randomized extension. Let each node $v$ generate a short identifier $\ID^\ast(v)$ consisting of $O(\log\log\log\log n)$ uniformly random bits. Whenever a new edge appears (either by edge insertion or due to node insertion), the two endpoints first exchange their own $\ID^\ast$s as well as the multiset of their neighbors' $\ID^\ast$s. In addition, each endpoint of the inserted edge sends a signal to all of its neighbors.

Suppose the inserted edge $\{u,v\}$ closes a triangle $\{u,v,w\}$ (Case 2 above). We show that the triangle can be detected in the setting of both \Cref{lem: memdect lb ks} and  \Cref{thm:mixed}.
\begin{itemize}
    \item Consider $\MemDetect(K_3)$ under edge insertion (the setting of \Cref{lem: memdect lb ks}).  Then $w$ deterministically detects the triangle because it receives two signals from its neighbors. Moreover, $u$ (and symmetrically $v$) detects the triangle provided that $\ID^\ast(u), \ID^\ast(v), \ID^\ast(w)$ are all distinct, which occurs with probability $1 - O(1/\poly\log\log\log n)$.
    \item Consider $\Detect(K_3)$ under both node and edge insertion (the setting of \Cref{thm:mixed}). In this setting,  We only let $u$ and $v$ (the endpoints of the inserted edge)  detect the triangle, as node $w$ cannot do this because it cannot differentiate between Cases 1 and 3. Similarly, $u$ and $v$ successfully detect the triangle with probability $1 - O(1/\poly\log\log\log n)$.
\end{itemize}
 
\paragraph{An alternative approach} An alternative way to show that our lower bounds in \Cref{sect:memList,sect:memDect} extend to the randomized setting is via reductions from two-party communication problems. We take \Cref{lem: memlist lb edge insertions complete} as an example and demonstrate its randomized extension via a reduction from the \emph{set-disjointness} problem: Alice holds a subset $S_A \in [n]$, Bob holds a subset of $S_B \in [n]$, and they want to decide whether $S_A \cap S_B = \emptyset$. It is well-known~\cite{roughgarden2016communication} that the problem requires $\Omega(n)$ bits of communication, even allowing a constant error probability of $\varepsilon < 1/2$.

In the proof of \Cref{lem: memlist lb edge insertions complete}, the hard instance construction is parameterized by a choice of subset $U' \subseteq U = \{u_1, u_2, \ldots u_n\}$. Different choice of $U'$ leads to different outputs of $v$: Each $u_i \in U'$ corresponds to a distinct copy $H_i$ of $H$ that must be listed by $v$. Suppose there is an $r$-round algorithm $\mathcal{A}$ with bandwidth $B=o(n/r)$ that succeeds with probability $1-\varepsilon$ over the hard instance. We show how to use the algorithm to obtain an $o(n)$-bit set-disjointness protocol with the same success probability $1-\varepsilon$, contradicting the known $\Omega(n)$ set-disjointness lower bound and yields the desired $\Omega(n/r)$ bandwidth lower bound.

We simulate the dynamic graph construction and algorithm $\mathcal{A}$ between Alice and Bob with $U' = \{u_i \, | \, i \in S_A\}$. When above process finishes, Bob decides that $S_A \cap S_B = \emptyset$ if $i \notin S_B$ for all $H_i$ listed by $v$. Since  $\mathcal{A}$ is an $r$-round algorithm with bandwidth $B=o(n/r)$ and node $v$ has degree at most two in our hard instance construction, the number of bits communicated between Alice and Bob is $O(Br)=o(n)$.
 
A similar reduction from set-disjointness applies to all other lower bound proofs in \Cref{sect:memList}. For the lower bounds proofs (\Cref{lem: memdect lb edge insertion} and \Cref{lem: memdect lb node insertion}) in \Cref{,sect:memDect}, we need to consider a different communication problem: In the \emph{index} problem, Alice holds a subset $S_A \subseteq [n]$, Bob holds an index $b \in [n]$,  Bob needs to decide whether $b \in S_A$ after receiving some message from Alice (the communication is only one-way). It is well-known~\cite{roughgarden2016communication} that the problem requires $\Omega(n)$ bits of communication, even allowing a constant error probability of $\varepsilon < 1$.

We take \Cref{lem: memdect lb edge insertion} as an example, where the main difference from the proof of \Cref{lem: memlist lb edge insertions complete} is that here the hard instance construction is parameterized by not only $U' \subseteq U = \{u_1, u_2, \ldots, u_n\}$ but also a choice of $u_j \in U$, and then $v$ should decide whether $u_j \in U'$. The reduction from the index problem is then obtained by setting $U' = \{v_i \, | \, i \in S_A\}$ and $j = b$.

\end{document}